\documentclass[11pt,a4paper]{article}
\usepackage{times}

\usepackage{multicol}
\usepackage{multirow}
\usepackage{url}
\usepackage{wrapfig}
\usepackage[hidelinks,pagebackref]{hyperref}
\renewcommand*\backref[1]{\ifx#1\relax \else (see page #1) \fi}
\usepackage[utf8]{inputenc}
\usepackage[T1]{fontenc}
\usepackage[small]{caption}
\usepackage{amsfonts}
\usepackage{graphicx}
\usepackage{mathdots}
\usepackage{color}
\usepackage{amsmath}
\usepackage{amssymb}
\usepackage[numbers,sort]{natbib}
\usepackage{nicefrac}
\usepackage{amsthm}
\usepackage{booktabs}
\usepackage{paralist}
\usepackage{xcolor}
\usepackage{algorithm}
\usepackage{algorithmic}
\usepackage{cleveref}
\usepackage{thm-restate}
\newcommand{\appsymb}{$\bigstar$}
\usepackage{authblk}
\usepackage{fullpage}
\usepackage[inline]{enumitem}

\usepackage{subcaption}
\usepackage{todonotes}
\usepackage{kbordermatrix}
\usepackage{pgfplots}
\usepackage{tikz}
\usepackage{tikz-qtree}

\newcommand{\pfnote}[1]{}
\newcommand{\ssnote}[1]{}
\newcommand{\pfsidenote}[1]{}
\newcommand{\nbnote}[1]{}

\newtheorem{theorem}{Theorem}[section]
\newtheorem{proposition}[theorem]{Proposition}

\newtheorem{definition}{Definition}[section]
\newtheorem{example}{Example}[section]

\newtheorem{remark}{Remark}[section]

\newcommand{\calF}{\mathcal{F}}
\newcommand{\calP}{\mathcal{P}}
\newcommand{\calT}{\mathcal{T}}

\newcommand{\swap}{{{\mathrm{swap}}}}
\newcommand{\spear}{{{\mathrm{Spear}}}}
\newcommand{\pref}{\succ}
\newcommand{\sharpp}{{{\mathrm{\#P}}}}
\newcommand{\np}{{{\mathrm{NP}}}}
\newcommand{\fpt}{{{\mathrm{FPT}}}}

\newcommand{\pos}{{\mathrm{pos}}}
\newcommand{\dist}{{\mathrm{dist}}}

\newcommand{\reals}{{\mathbb{R}}}

\newcommand{\calD}{{\mathcal{D}}}

\newcommand{\POS}{{{\mathrm{POS}}}}

\newcommand{\EMD}{{{\mathrm{emd}}}}

\newcommand{\rID}{{{\mathrm{rID}}}}
\newcommand{\ID}{{{\mathrm{ID}}}}
\newcommand{\id}{{{\mathrm{id}}}}
\newcommand{\UN}{{{\mathrm{UN}}}}
\newcommand{\AN}{{{\mathrm{AN}}}}
\newcommand{\an}{{{\mathrm{an}}}}
\newcommand{\ST}{{{\mathrm{ST}}}}

\newcommand{\stt}{{{\mathrm{st}}}}

\newcommand{\sgn}{{{\mathrm{sgn}}}}
\newcommand{\Euc}{{{\mathrm{Euc}}}}
\newcommand{\MR}{{{\mathrm{MR}}}}
\newcommand{\AMR}{{{\mathrm{AMR}}}}

\newcommand{\relvot}{{{\mathrm{freq}}}}

\newcommand{\normphi}{{{\mathrm{norm}\hbox{-}\phi}}}

\definecolor{darkgreen}{rgb}{0,0.5,0}
\definecolor{darkpink}{rgb}{0.75,0.25,0.25}
\definecolor{RED}{rgb}{1,0,0}

\newcommand{\figurecut}[1]{}
\newcommand{\cutfornow}[1]{}

\newcommand{\hideforaamas}[1]{}

\title{Drawing a Map of Elections\thanks{This paper merges results
    from an AAMAS-2020 paper \cite{szu-fal-sko-sli-tal:c:map}, an IJCAI-2021 paper \cite{boe-bre-fal-nie-szu:c:compass}, and the PhD thesis
    of \citet{szu:thesis:map-of-elections}.
    In particular, most of Sections~\ref{sec:map-cultures} and~\ref{sec:experiments}, as well as significant fragments of Section~\ref{sec:real-life}, were not published in either of the two conference papers, or were
    published in very rudimentary forms. Code and data for the paper are available at \url{https://github.com/Project-PRAGMA/Journal---Drawing-a-Map-of-Elections}.}}

\author[1,2]{Stanisław Szufa}
\author[3]{Niclas Boehmer}
\author[4]{Robert Bredereck}
\author[1]{Piotr Faliszewski}
\author[5]{Rolf Niedermeier} 
\author[6]{Piotr Skowron}
\author[7]{Arkadii Slinko}
\author[8]{Nimrod Talmon}

\affil[1]{\small
  AGH University,
  \{faliszew,szufa\}@agh.edu.pl}
 \affil[2]{\small
  CNRS, LAMSADE, Universit\'{e} Paris Dauphine -- PSL,
  s.szufa@gmail.com}
\affil[3]{\small
  Hasso Plattner Institute, University of Potsdam,
  niclas.boehmer@hpi.de}
\affil[4]{\small
  Institut für Informatik, TU Clausthal, 
  robert.bredereck@tu-clausthal.de}
\affil[5]{\small
  Technische Universit\"at Berlin, Algorithmics and Computational 
  Complexity}
 \affil[6]{\small
  University of Warsaw,
  p.skowron@mimuw.edu.pl}
 \affil[7]{\small
  University of Auckland,
  a.slinko@auckland.ac.nz}
 \affil[8]{\small
  Ben-Gurion University,
  talmonn@bgu.ac.il}
\date{}

\begin{document}

\maketitle

\begin{abstract}
  Our main contribution is the introduction of the map of elections
  framework. A map of elections consists of three main elements: (1) a
  dataset of elections (i.e., collections of ordinal votes over given
  sets of candidates), (2) a way of measuring similarities between
  these elections, and (3) a representation of the elections in the 2D
  Euclidean space as points, so that the more similar two elections
  are, the closer are their points.  In our maps, we mostly focus on
  datasets of synthetic elections, but we also show an example of a
  map over real-life ones. To measure similarities, we would have
  preferred to use, e.g., the isomorphic swap distance, but this is
  infeasible due to its high computational complexity.  Hence, we
  propose polynomial-time computable positionwise distance and use it
  instead. Regarding the representations in 2D Euclidean space, we
  mostly use the Kamada-Kawai algorithm, but we also show two
  alternatives.
  We develop the necessary theoretical results to form our maps and
  argue experimentally that they are accurate and credible. Further,
  we show how coloring the elections in a map according to various
  criteria helps in analyzing results of a number of experiments.  In
  particular, we show colorings according to the scores of winning
  candidates or committees, running times of ILP-based winner
  determination algorithms, and approximation ratios achieved by
  particular algorithms.
\end{abstract}

\section{Introduction}\label{sec:intro}

Alongside theoretical research, experimental studies are in the very
heart of \emph{computational social
  choice}~\cite{bra-con-end-lan-pro:b:comsoc-handbook}.  Computational
aspects of elections, such as the problems of winner determination~
\citep{bou-lu:c:chamberlin-courant,fal-lac-pet-tal:c:csr-heuristics,csa-lac-pic:c:schulze-method,mic-sha:c:predict-election-outcome},
identifying and analyzing various forms of
manipulation~\citep{wal:c:where-hard-veto,dav-kat-nar-wal-xia:j:borda-nanson-baldwin-manipulation,gol-lan-mat-per:c:rank-dependent,kel-has-haz:j:approx-manipulation,bar-lan-yok:c:approval-hamming-manipulation,lu-zha-rab-obr-vor:c:manipulation-issue-selection} or
control~\citep{erd-fel-rot-sch:j:experimental-control,bre-che-nie-wal:c:parliamentary,wil-vor:c:control-social-influence},
measuring candidate performance~\citep{bri-sch-suk:c:mov-tournaments,boe-bre-fal-nie:c:counting-bribery,boe-fal-jan-pet-pie-sch-sko-szu:c:pb-performance}, and
preference elicitation~\citep{lu-bou:c:max-regret,ore-fil-bou:c:efficient,ben-nat-pro-sha:c:pb-preference-elicitation}
are nowadays often investigated through experiments. For example,
researchers evaluate running times of
algorithms~\cite{wal:j:hard-manipulation,fal-sli-sta-tal:j:clustering,wan-sik-she-zha-jia-xia:c:practical-algo-stv},
or test what approximation ratios appear in
practice~\cite{sko-fal-sli:j:multiwinner,kel-has-haz:j:approx-manipulation}. It
is also common to test non-computational properties of elections and
voting rules---for instance to evaluate how frequently a given
phenomenon occurs (e.g., how frequently a given voting rule is
manipulable~\cite{wal:j:hard-manipulation,gol-lan-mat-per:c:rank-dependent,erd-fel-rot-sch:j:experimental-control},
or how frequently particular candidates
win~\cite{elk-fal-las-sko-sli-tal:c:2d-multiwinner}; naturally, the
cited papers are just a few examples).
Yet, designing convincing experiments is not easy and, in particular,
it is not clear what election data to use.

The main contribution of this paper is setting up the ``map of
elections'' framework, which allows one to build and visualize
election datasets, and helps in planning and evaluating experiments.  Further, we argue that
the framework is useful by showcasing a number of its applications.
We focus on ordinal elections, i.e., on the setting where an election
consists of a set of candidates and a collection of voters that rank
the candidates from the most to the least desirable ones. For results
regarding approval elections we point the reader to the follow-up work of
\citet{szu-fal-jan-lac-sli-sor-tal:c:map-approval}.

\subsection{Motivating Example}\label{sec:mot-example}
Imagine we are interested in 
the Harmonic Borda~\cite{fal-sko-sli-tal:c:paths} (HB) multiwinner
voting rule. Under this rule we are given an ordinal election and an
integer $k$, and the rule chooses a committee of $k$ candidates that
minimizes the sum of dissatisfaction scores assigned by the voters; it
is a variant of the classic proportional approval voting rule,
PAV~\citep{Thie95a,kil:chapter:approval-multiwinner} (see
Section~\ref{sec:dec-dod-hb} for a detailed definition).  Finding a
winning committee under this rule is
$\np$-hard~\cite{sko-fal-lan:j:collective}, but such a committee can
be computed, e.g., using ILP solvers (i.e., solvers for integer linear
programming problems).
We want to assess how quickly an ILP solver can compute the winning
committees.
 
Ideally, we would want to try all elections of a given size. For
example, elections with 100 candidates and 100 voters are common in
the multiwinner
literature~\citep{elk-fal-las-sko-sli-tal:c:2d-multiwinner,cel-hua-vis:c:fairness-multiwinner,fal-lac-pet-tal:c:csr-heuristics};
see also the overview of
\citet{boe-fal-jan-kac-lis-pie-rey-sto-szu-was:c:guide}. Naturally,
this is infeasible. Instead, a natural approach is to generate
elections according to several standard distributions, referred to as
\emph{statistical cultures},
and test the algorithms on them.\footnote{It would also be natural to
  consider real-life elections (e.g., from
  Preflib~\cite{mat-wal:c:preflib}, or those collected by Boehmer and
  Schaar~\cite{boe-sch:c:real-life-elections}, currently also included
  in Preflib).  However, many sources of real-life elections lead to
  data over relatively few candidates (often just three or
  four~\cite{tid-pla:b:modeling-elections,mat-for-gol:c:empirical-study},
  although elections with 15 and more candidates are not uncommon
  either).  Accordingly, such elections would not suffice for our
  experiments (yet, they are useful in other cases; see, e.g., the
  works of Brandt et al.~\cite{bra-gei-str:c:ehrhart-theory} and Ayadi
  et al.~\cite{aya-ben-lan-pet:c:stv-incomplete}).} Indeed, many of
the
above-cited papers focus on some subset of the
following four models
(see Section~\ref{sec:prelim} for detailed descriptions of the
distributions):
\begin{enumerate}
\item The impartial culture (IC) model, where all votes are generated
  uniformly at random and independently (this is one of the most
  popular models in experiments, often used as a baseline for other
  experiments~\citep{wal:j:hard-manipulation,erd-fel-rot-sch:j:experimental-control,gol-lan-mat-per:c:rank-dependent,wan-sik-she-zha-jia-xia:c:practical-algo-stv,kel-has-haz:j:approx-manipulation,sko-fal-sli:j:multiwinner}).
  
\item The Polya-Eggenberger urn
  model~\citep{berg1985paradox,mcc-sli:j:similarity-rules}, which
  proceeds similarly to the impartial culture one, but already
  generated votes have increased probability of being generated
  again---as specified by a ``contagion'' parameter---and, hence, the
  model introduces correlations.  The model is quite common in various
  experiments~\cite{wal:j:hard-manipulation,erd-fel-rot-sch:j:experimental-control,fal-sli-sta-tal:j:clustering,kel-has-haz:j:approx-manipulation,sko-fal-sli:j:multiwinner}.
  
\item The Mallows model~\citep{mal:j:mallows}, and its mixtures,
  where the probability of generating a vote depends on its similarity
  to a preselected central one (a dispersion parameter controls the concentration
  of the generated votes). Mallows model is particularly natural for
  settings with some sort of a ground
  truth~\cite{gol-lan-mat-per:c:rank-dependent,sko-fal-sli:j:multiwinner,bra-gei-str:c:ehrhart-theory,mic-sha:c:predict-election-outcome,ben-sko:c:comparing-rules}.

\item The Euclidean model~\cite{enelow1984spatial,enelow1990advances},
  where candidates and voters are drawn as points in a Euclidean space
  and the closer a candidate is to a voter, the more this voter
  appreciates this candidate, is also common in recent
  experiments~\cite{bra-gei-str:c:ehrhart-theory,elk-fal-las-sko-sli-tal:c:2d-multiwinner,fal-sli-sta-tal:j:clustering,lac:c:perpetual-voting};
  the way of generating the ideological positions of the candidates
  and voters is the parameter of the model.

\end{enumerate}

\noindent
However, which of these models should we use and how should we set
their parameters? Perhaps we should also use some other models as well,
possibly generating elections that are
single-peaked~\cite{bla:b:polsci:committees-elections},
single-crossing~\cite{mir:j:single-crossing,rob:j:tax},
group-separable~\cite{ina:j:group-separable,ina:j:simple-majority}, or
that are structured in some other way?
Intuitively, we would
like to have a set of elections that would be as varied as possible for our experiment,
so that, on the one hand, we would not spend too much time on very
similar elections---for which we expect nearly identical results of
the experiments---and, on the other hand, we would not miss
interesting families of elections---for which the results would be
hard to predict.

Unfortunately, researchers working on elections only recently started to ask and carefully
answer questions like the ones given above, though, of course, there
are examples of older experimental papers that address these issues in
detail; see, e.g., the work of \citet{bra-gei-str:c:ehrhart-theory}.
This view is supported by a recent analysis by \citet{boe-fal-jan-kac-lis-pie-rey-sto-szu-was:c:guide} of over $160$ papers that
include numerical experiments on elections and were published between
2009 and 2023 in the IJCAI, AAAI, and AAMAS conference series. In
particular, the authors found that in about half of the papers that
considered ordinal elections (of which there were a bit over $120$),
researchers only used a single data source.\footnote{If one's paper were to address a specific practical problem only using
real-life data from the relevant data source can be
justified. However, many of these papers relied on the highly unrealistic
impartial culture model, which may put some doubt on the relevance and generalizability of
their results.} While the situation has been improving over time and in
recent papers researchers do use more varied data, the process is
slow\footnote{According to
\citet{boe-fal-jan-kac-lis-pie-rey-sto-szu-was:c:guide}, in 2023 still
about half of the papers studying ordinal elections used only one or two
data sources.} and principled advice on how experiments should be conducted is missing.

As reported by
\citet{boe-fal-jan-kac-lis-pie-rey-sto-szu-was:c:guide}, nearly
two-thirds of papers that include experiments on elections rely on
synthetically generated data, without looking at any elections based
on real-life data.  While this is natural---real-life data relevant to
a particular setting is not easy to come by---it reinforces the need
for a good understanding of how elections generated synthetically
relate to each other and to elections occurring in practice.  Indeed,
many papers do not seem to be making informed decisions on which data
sources to include and either rely on intuitions or on copying
approaches of previous authors working on similar topics (see the
discussion of election sizes in the work of
\citet{boe-fal-jan-kac-lis-pie-rey-sto-szu-was:c:guide}, where the
authors discuss typical application scenarios).

Our main motivation behind the map of elections framework is to help
with solving the issues mentioned above. In particular, the maps help
with analyzing and using diverse data sources and with relating synthetic
elections to real-life ones.

\begin{figure}
    \begin{subfigure}[b]{0.49\textwidth}
      \centering
      \includegraphics[width=8.cm, trim={0.2cm 0.2cm 0.2cm 0.2cm}, clip]
        {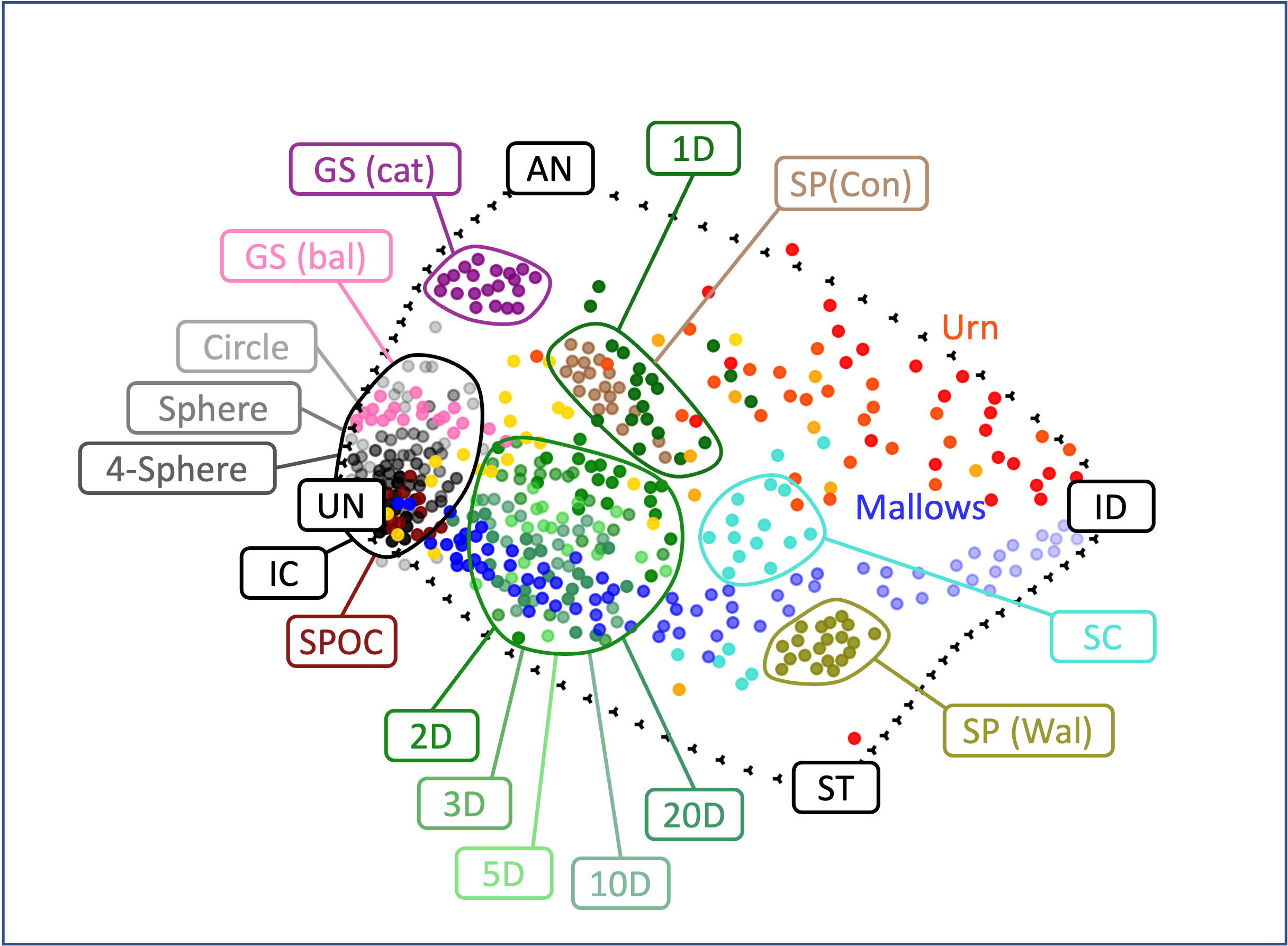}
      \caption{\label{fig:10x100SC}\footnotesize A map of synthetic elections.}
  \end{subfigure}
  \hfill
   \begin{subfigure}[b]{0.49\textwidth}
    \centering            
    \includegraphics[width=8cm, trim={0.2cm 0.2cm 0.2cm 0.2cm}, clip]{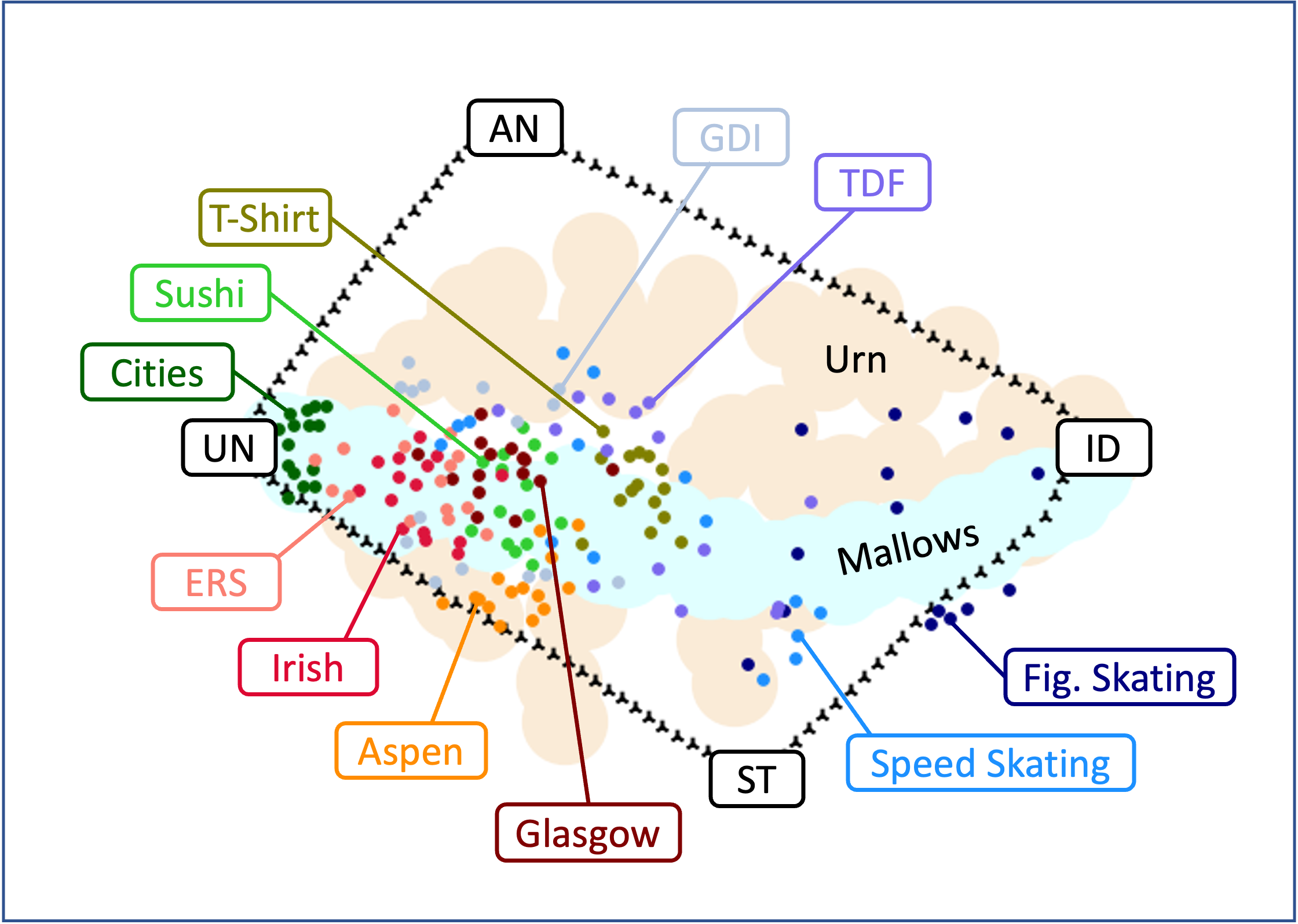}
    \caption{\label{fig:10x100RW}\footnotesize 
    A map of various real-life elections. }
  \end{subfigure}
  \caption{\label{fig:10x100+compass} Two examples of maps of
    elections.  Each point on each of the maps represents an election
    with $10$ candidates and $100$ voters. Elections on the left were
    generated using various statistical cultures (as reflected by
    their colors), whereas those on the right were derived from
    real-life data (after appropriate preprocessing). We define
    various statistical cultures in \Cref{sec:cultures}. Maps of
    synthetic elections are discussed in detail in
    \Cref{sec:map-cultures}, whereas the map of real-life elections is
    discussed in \Cref{preflib_info}; the maps will become fully
    comprehensible only after reading these sections. The maps include
    four special points (UN, ID, ST, AN), to which we refer as the
    ``compass,'' and their connecting paths.  The map on the right
    also includes pale blue and orange areas, which represent places
    where Mallows and urn elections would land, respectively.}
\end{figure}

\subsection{Maps of Elections}

Our main idea is that given a set of elections we may arrange them as
points on a plane, so that the Euclidean distances between the points
would indicate the similarity between the corresponding
elections. Such points form a map, which can then be colored in a
number of ways. For example, each color may correspond to the
statistical culture (or, a real-life dataset) from which the given
election comes, or to some feature of the election.  We show two
examples of such maps and the former type of coloring in
Figure~\ref{fig:10x100+compass}.  In \Cref{sec:experiments}, we show a
number of colorings of the latter type, including ones related to the
motivating example from \Cref{sec:mot-example}.

Preparing a map requires three main components: (1)~the set of
elections that we want to work with, (2)~the way of measuring
similarities between the elections, and (3)~the way of embedding
points on the plane that respects these similarities.  We describe and
evaluate these components in detail in \Cref{sec:positionwise} and
\Cref{sec:map-cultures}. Below we give a quick overview.

\paragraph{Election Datasets and Compass.}
We mostly focus on synthetic elections, generated from various
statistical cultures, such as those mentioned in the motivating
example above. Specifically, we build several election datasets, where each of them consists of 480 elections that
come from several well-known cultures, with various parameter
settings,\footnote{We additionally use 84 elections that help us
  navigate through the map.}
with each dataset involving a different number of candidates.
The idea is that if one wants to perform some computational experiment
on elections, a priori it is not clear which models are most
interesting to consider.  Our datasets contain elections sampled from many
different ones, so that one could use these datasets for initial studies and only
then decide which models to use in follow-up, more focused
experiments.  That said, our choices are, of course, somewhat
arbitrary and one may as well prefer other combinations of models and
their parameters.  We complement our results on the synthetic models
by also considering elections derived from real-life data, available
in Preflib~\citep{mat-wal:c:preflib} or collected by
\citet{boe-sch:c:real-life-elections} (the latter ones are now
included within Preflib).

In addition to the base content, all our datasets also include a
number of elections that help in navigating the maps. In particular,
these elections include what we call ``the compass,'' that is,
four elections that capture four different types of (dis)agreement among the voters:
\begin{enumerate}
\item The identity election, $\ID$, which models perfect agreement and
  contains voters that all have the same preference order.

\item The uniformity election, $\UN$, which  models perfect lack of
  agreement (or, chaos) and contains voters with all possible
  preference orders.

\item The stratification election, $\ST$, which models partial agreement, where
  voters agree that half of the candidates are better than those in the other
  half, but lack agreement on anything else.

\item The antagonism election, $\AN$, which models a conflict where half
  of the voters have one preference order and the other half has the
  opposite one.
\end{enumerate}
Besides the compass elections themselves, we also include \emph{paths} of intermediate
elections in between them.
For example, a path from $\ID$ to $\AN$ can be seen as consisting of
elections where, initially, all voters agree on a single preference order,
but  as we move toward $\AN$, more and more voters report the opposite one.

\paragraph{Similarities Between Elections.}
To measure the similarity of our elections, we need a notion of a
distance. Furthermore, this distance needs to be both neutral and
anonymous---i.e., independent of the names of both the candidates and
the voters. Indeed, we want to compare elections generated from
statistical cultures, where this information is random, as well as
elections that come from different real-life datasets, whose
candidates and voters are unrelated.  Such distances were 
studied by \citet{fal-sko-sli-szu-tal:j:isomorphism}, who---in
particular---proposed the isomorphic variants of the swap and Spearman
distances, which satisfy our requirements. However, while these
distances have many attractive
features~\citep{boe-fal-nie-szu-was:c:metrics}, they also are
$\np$-hard to compute, which makes them difficult to use on elections
with more than a few candidates. Since we want our framework to be
scalable, we provide a new distance, which we call
\emph{positionwise}. This distance is less precise than the isomorphic
ones, but still satisfies our desiderata, is computable in polynomial
time, and appears to give meaningful results.

Briefly put, the key idea of the positionwise distance is to not work
directly on elections and votes, but to use an aggregate
representation. Specifically, given an election with $m$ candidates,
its \emph{frequency matrix} is an $m \times m$ matrix where the
columns correspond to the candidates, the rows correspond to the
positions in the votes, and the entries give fractions of the votes
where a given candidate was ranked on a given position. Positionwise
distance measures the distance between such matrices. An added benefit
of using frequency matrices is that we show a nearly 1-to-1
correspondence between them and elections with a given number of
voters.  Consequently, instead of working directly with our compass
elections, it is often more convenient to work with their matrices.
Generally, we often speak of matrices and elections interchangeably.

\paragraph{Embedding Algorithms.}
There is a number of algorithms that given a set of objects and
distances between each pair of these objects find their embedding as
points on a plane (or, other Euclidean space), so that the Euclidean
distances between the points resemble those between the objects that
they represent.  (getting perfect embedding is often impossible; as an
example, consider embedding 10 points on a plane, where each pair is
supposed to be at the same distance). These algorithms may either focus on
obtaining embeddings that are more visually pleasing, or may strive
for as high accuracy as possible.

In most of our visualizations, we use the Kamada-Kawai
algorithm~\citep{kam-kaw:j:embedding}, which belongs to the latter
group and which gave the most accurate embeddings according to a
number of criteria (another example of such an algorithm is the
Multidimensional Scaling algorithm, MDS,~\citep{kru:j:mds,lee:j:mds},
but it did not perform as well).
A brief idea of the Kamada-Kawai algorithm is that it considers the
points to be embedded as vertices in the 2D space that are connected
with springs.  Each spring has its resting length corresponding to the
desired distance (the positionwise distance in our case) between its two endpoints, but typically it is either
stretched or compressed. The algorithm iteratively moves the vertices
to minimize the energy of the system. The original algorithm uses the
Newton-Raphson method for this purpose, but we use the faster
optimization method presented by Barzilai and
Borwein~\citep{bar-bor:j:kam-kaw}. For a comparison of the two
variants of the Kamada-Kawai algorithm with the two optimization methods,
see the work of \citet{pos-has-now-pla:c:kam-kaw-bar-bor}; we point to
the work of \citet{mt:sapala} for the implementation that we use.

We note that in some cases the force-directed algorithm of
\citet{fru-rei:j:graph-drawing}, which belongs to the group of
algorithms that prioritize creating a visually pleasing embedding,
also can be useful. In particular, it produces embeddings where the
election points are more spread over the plane, which in some
circumstances can help with the analysis of the map; such a more
spread map still shows the main relations between various elections,
but the points do not overlap and it is easier to see their colors.
We conduct a quantitative comparisons of the different embeddings
produced by these algorithms in \Cref{app:embeddings}.

\subsection{Use Cases for the Maps }\label{sec:usecases}

There is a number of ways in which maps of elections can help in
conducting numerical experiments. Below we give three general types of
experiments where maps can be helpful:

\begin{description}
\item[Finding Relations Between Elections.] By their very nature, maps
  arrange elections in 2D space, so that similar elections are close
  to each other. By examining a map and comparing the positions of
  particular elections to each other and to the characteristic compass
  elections, we learn about their nature and properties. For example,
  in \Cref{sec:map-cultures} we form a dataset of elections
  generated from various statistical cultures and find how they relate
  to each other by drawing a map where each election has a color that
  specifies the statistical culture from which it was generated. Among
  other findings, we observe how the Conitzer and Walsh models for
  generating single-peaked elections produce elections that are quite
  different from each other, and how the Conitzer model is similar to
  the 1D Euclidean model. This type of experiment also allowed
  \citet{boe-bre-fal-nie-szu:c:compass,boe-fal-kra:c:mallows-normalization}
  to discover that the dispersion parameter in the Mallows model may
  need normalization, and to propose such a normalization.
  Perhaps most importantly, by looking at positions where elections
  land on the maps, we get a better understanding as to which areas of
  the election space are not covered by our datasets. This encourages
  forming more diverse datasets for further experiments (indeed, we
  form such a dataset in \Cref{sec:map-cultures}).
  Further, experiments where we analyze the positions of elections on
  the map allow us to learn about real-life data. For example, while
  there are statistical cultures that generate highly polarized
  elections---i.e., elections close to the antagonism compass
  election---real-life elections rarely fall in the respective region
  of the map (see Section~\ref{sec:real-life}).  By analyzing the
  proximity of real-life elections to those generated from statistical
  cultures, we may also learn which cultures (have a chance to)
  generate realistic data. We show some results in this direction in
  Section~\ref{sec:real-life}; in a follow-up paper,
  \citet{boe-bre-elk-fal-szu:c:frequency-matrices} pursue this
  direction more systematically.

\item[Visualizing Election Properties.] Given a dataset of elections,
  we may want to learn something about some of their features, such as
  the Borda scores of their winners (see
  \Cref{sec:map:borda-copeland}), the running times an an ILP solver
  takes to compute their winners according to some $\np$-hard voting
  rules (see \Cref{sec:map:ilp-times}), or many others.  In this case,
  we color the elections on the map according to a given feature. This
  way we obtain a comprehensible bird's eye view at the results for
  all the elections, without resorting to computing averages or other
  statistics (naturally, these statistics are also important---the
  maps simply offer a different perspective and an ability to spot
  issues that otherwise would not be as readily visible).
  In the most beneficial case, we learn that the considered
  features closely depend on elections' positions on the map. For
  example, in \Cref{sec:map:borda-copeland} we find that Borda scores of the winners tend to
  increase as we move from the uniformity compass election toward the
  identity one. In \Cref{sec:map:ilp-times} we see a similar phenomenon for the running
  time required to find Harmonic Borda winners using an ILP solver.
  Another possibility is that we would spot some pattern regarding the
  considered feature and statistical cultures present on the map.  For
  example, in \Cref{sec:map:ilp-times} we observe that an ILP
  formulation for computing the results of the Dodgson rule tends to
  require a close-to-fixed amount of time for elections from all our
  statistical cultures (for a given fixed size of elections) except
  for two, where on one it is much faster and on the other it is much
  slower. Maps that include diverse statistical cultures facilitate
  such observations because they include varied data\footnote{As shown
    by \citet{boe-fal-jan-kac-lis-pie-rey-sto-szu-was:c:guide}, many
    computational social choice papers that include experiments on
    elections limit their attention to a small collection of
    statistical cultures. When performing experiments on such limited
    data it is very easy to overlook such phenomena as described
    here.} and allow one to see the results ``all at once.''
  An experiment of the type described here was recently provided by
  \citet{con-gra-men:c:explanations-map}, where the map was used
  to analyze the sizes of explanations regarding why a particular
  candidate won in a given election.
  
\item[Comparing Algorithm Performance.] The third use case is when we
  want to compare, e.g., how two different algorithms perform on
  elections from a given dataset.  As an example, in
  \Cref{sec:map:approx} we take two approximation algorithms for the
  $\np$-hard Chamberlin--Courant multiwinner voting
  rule~\citep{cha-cou:j:cc} and compute the approximation ratios they
  achieve on our elections. Then, we color the elections on the map
  according to the algorithm that performed better (with the blue
  color for one algorithm and red for the other) and its performance
  (the better is the achieved approximation ratio, the darker is the
  color). Such an experiment is particularly useful for areas of the
  map where, on average, both algorithms perform similarly but where
  on particular elections one of the algorithms is notably better than
  the other one: Various colorings of the map allow us to spot such places
  more easily than using basic statistics (such as the average
  approximation ratio on elections from a given statistical
  culture). Recently, a similar visualization was used by
  \citet{del-pet:t:indifferent-stv-map}, who analyzed two different
  ways of generalizing the STV voting rule to elections where voters
  can express indifferences.
\end{description}

We want to stress that, in principle, all the experiments and types of
experiments that we described above could be performed without maps
and without the visualizations that they offer, possibly letting one
reach the same conclusions and spotting the same phenomena as using
the maps. The main advantage of using the maps is that it may be
significantly easier to make relevant observations and reach
particular conclusions.
We also stress that the researchers can take advantage of our maps
even without including the actual plots in their papers. For example,
they may run their initial experiments using our diverse, synthetic
datasets, analyze the maps to see where interesting phenomena happen,
and then focus only on those areas.  Indeed, this is how the map
framework was used, e.g., in the works of
\citet{boe-bre-fal-nie:c:counting-bribery} and
\citet{boe-bre-fal-nie:c:counting-bribery-real-life}.
That said, the maps also have a number of drawbacks and there are
challenges that one faces when using them. We describe them throughout
the paper and summarize them in \Cref{sec:conclusions}.

\subsection{Related Work}\label{sec:related}

To the best of our knowledge, this paper is the first to provide a
principled framework for analyzing and visualizing election
datasets. In this respect, we are not aware of previous works to which
we can directly compare our approach. Nonetheless, we do build on
numerous ideas that were already present in the literature, and we
continued working on the topic.  Below we review both our follow-up works
and other related literature.

\paragraph{The Map Framework.}
This paper is based on two conference papers, where the first one
proposed the map framework and showed its usefulness through initial
experiments~\cite{szu-fal-sko-sli-tal:c:map}, and the second one
added, e.g., the compass elections, analysis of position and frequency
matrices, and studies of real-life
data~\cite{boe-bre-fal-nie-szu:c:compass}. The paper is also based on
fragments of the PhD thesis of Stanisław
Szufa~\citep{szu:thesis:map-of-elections}. Both conference papers also
included content that we decided to omit---sometimes to maintain
focus, sometimes because further research showed that our initial
ideas were misguided, and sometimes because they deserved more
detailed treatment than we could have provided here. We also show some
additional results, not present in either of the two papers. Indeed,
most of Sections~\ref{sec:map-cultures} and~\ref{sec:experiments}
present new material; Section~\ref{sec:real-scores} also is new. Yet,
on the whole, the current paper can largely be seen as a union of the
two initial works and fragments of the PhD thesis. Next we give an
overview of the follow-up papers that regard the framework.

\citet{boe-bre-elk-fal-szu:c:frequency-matrices} extended the map
framework by observing that the positionwise distance can be applied not
only to elections, but also to distributions over votes---such as
impartial culture, the Mallows distribution, or various distributions
of single-peaked votes---by measuring distances between the expected
frequency matrices of such distributions. They have provided a number
of ways of computing such matrices and have shown that they can be
used for learning parameters of the models that, e.g., are as close as
possible to various real-life elections. In this sense, their work put
a principled layer over the analysis of real-life elections that we provide
in \Cref{sub:recom}.

The work of \citet{boe-bre-elk-fal-szu:c:frequency-matrices} can also
be seen as presenting one more argument---on top of low computational
complexity and apparent usefulness---for using the positionwise
distance in the map framework. Yet, a priori, it is not really clear
if the positionwise distance really is making a good compromise
between computational complexity and expressivity. Perhaps an even
simpler distance would give equally good (or, only slightly worse)
results? Or, maybe, the results given by the positionwise distance are
too far off and we should seek a different approach?  This issue was
addressed by \citet{boe-fal-nie-szu-was:c:metrics}, who compared
several distances and analyzed their properties. In the end, they
concluded that the compromise offered by the positionwise distance
indeed is appealing. However, if one considers small enough elections
or has sufficient computing power, using the isomorphic swap distance
certainly \emph{is} preferable, as done, e.g., by
\citet{fal-kac-sor-szu-was:c:microscope}, who explain the positions of
elections on the map---generated according to the isomorphic swap
distance---using the notions of agreement, diversity, and polarization
among the voters. \citet{boe-cai-fal-fan-jan-kac:c:position-matrices}
have shown a number of downsides and difficulties that may arise when
using the positionwise distance, but they also concluded that on
typical datasets, such as those that we use in this paper, the issues
are not significant.

Maps of elections already proved useful in planning and conducting
several experiments and we mentioned some such examples in
\Cref{sec:usecases}.

While our work focuses on ordinal elections, in principle the map
framework can be applied to any type of objects, provided that one can
find an appropriate distance and, possibly, a set of compass points to
navigate the map. So far, this was done by
\citet{szu-fal-jan-lac-sli-sor-tal:c:map-approval} for the case of
approval elections (i.e., elections where each vote is a set of
candidates that a given voter finds appealing), by
\citet{boe-hee-szu:c:map-stable-roommates} for the case of the stable
roommates problem (in this problem we have a group of people who have
ordinal preferences over the other ones, regarding who to be teamed-up
with), and by \citet{fal-lac-sor-szu:c:map-of-rules} for the case of
approval-based multiwinner voting rules. Further,
\citet{fal-kac-sor-szu-was:c:microscope} have used the same approach
for visualizing particular votes within a single election.

\paragraph{Distances.}
There is a number of works in computational social choice that use
various forms of distances, either between votes or between
whole elections. For example, this is the case in the distance
rationalizability
framework~\cite{nit:j:closeness,mes-nur:b:distance-realizability,ser-ada:c:dr,elk-fal-sli:j:distance-rational,had-wil:j:dr}
(see also the overview of \citet{elk-sli:b:rationalization}) or in
various works that use distances to define forms of manipulating
elections~\cite{obr-elk:c:optimal-manipulation,bau-hog:c:generalized-distance-bribery,dey:j:local-distance-bribery,ana-dey:j:distance-manipulation}
or to analyze the space of
elections~\cite{obr-elk-fal-sli:c:swap-geometry,obr-elk-fal:c:convexity}.
The main difference between these works and ours is that they consider
distances between elections with the same candidate sets, whereas we
allow these sets to differ (although we do require them to be of the
same cardinality).  We point the readers interested in an encyclopedic
listing of various kinds of distances to the work of
\citet{dez-dez:b:encyclopedia}.

\paragraph{Election Data.}
Our datasets mostly consist of synthetic data, generated according to
a number of prominent (and not-so-prominent) statistical cultures. We
describe these cultures and give appropriate references in
Section~\ref{sec:cultures}.  In \Cref{preflib_info}, we also draw a
map of real-world elections.  In this map, we include elections from
all datasets from Preflib\footnote{We refer to the contents of Preflib
  prior to extending it with the datasets collected by
  \citet{boe-sch:c:real-life-elections}.} that contain at least ten
candidates and in which votes are not heavily incomplete (see
\Cref{app:selection} for an overview of the Preflib datasets that we
considered and the reasons for rejecting many of them).  These
elections include the now-classic Sushi preferences collected by
\citet{kam:c:sushi}, elections from the city councils of Aspen and
Glasgow~\citep{openstv}, and surveys regarding costs of living and
populations of a number of
cities~\citep{car-cha-kri-vou:j:optimizing}.  Moreover, we included
three datasets capturing sports competitions, collected by
\citet{boe-sch:c:real-life-elections}.  Naturally, there is also a
number of other sources of real-life data that one might use, but
which we do not include in our experiments (also, oftentimes, these
elections have below ten candidates).  For example, there are
extensive datasets from field experiments held during French
presidential
elections~\citep{las-str:j:french-election-2002,bau-ige:b:french-elections-2007,bau-ige-leb-gav-las:j:french-election-2012,bou-bla-bau-dur-ige-lan-lar-las-leb-mer:t:french-elections-2017}
and during German state and federal
ones~\citep{alo-gra:j:german-approval-voting}.

\paragraph{Embedding Algorithms.}
There is a number of algorithms that one could use to obtain an
embedding of points in space in a way that attempts to respect their
prescribed distances. For example, one could consider the Principal
Component Analysis~\citep{min:j:pca}, $t$-Distributed Stochastic
Neighbor Embeddings ($t$-SNE)~\citep{maa:c:tsne,maa-hin:j:tsne}, or
Locally Linear Embeddings~\citep{don-gri:j:lle,mha-wan:c:lle}.  We
have tried these algorithms, but the results where highly
unsatisfactory on our data. Instead, we focus on the algorithm of
\citet{kam-kaw:j:embedding} (as adapted by \citet{mt:sapala}), which
work by simulating a system of springs. In \Cref{app:embeddings} we
also compare it to the algorithms of \citet{fru-rei:j:graph-drawing}
and to (metric) Multidimensional Scaling~\citep{kru:j:mds,lee:j:mds},
which also perform well.

\section{Preliminaries}\label{sec:prelim}

We write $\reals_+$ to denote the set of non-negative real
numbers. For an integer $n$, we write $[n]$ to denote the set
$\{1,\ldots, n\}$.  Given a vector $x = (x_1, \ldots, x_m)$, we
interpret it as an $m \times 1$ matrix, i.e., we use column vectors.
For a matrix $X$, we write $x_{i,j}$ to refer to the entry in its
$i$-th row and $j$-th column.  For two finite sets $C$ and $D$ of the
same cardinality, we write $\Pi(C,D)$ to denote the set of bijections
from $C$ to $D$.
\smallskip

\subsection{Elections}
An election is a pair $E = (C,V)$, where $C = \{c_1, \ldots, c_m\}$ is
a set of candidates and $V = (v_1, \ldots, v_n)$ is a collection of
voters. Each voter~$v_i$ has a preference order, i.e., a ranking of
the candidates from the most to the least desirable one. To simplify
notation, $v_i$ refers both to the voter and to his or her preference
order; the exact meaning will always be clear from the context.
We write $v \colon a \pref b$ to indicate that voter $v$ ranks
candidate $a$ ahead of candidate $b$, and we write $\pos_v(c)$ to
denote the position of candidate $c$ in $v$'s preference order (the
top-ranked candidate has position $1$, the next one has position $2$,
and so on).  For an election $E = (C,V)$ and two candidates
$a, b \in C$, we write $M_E(a,b)$ to denote the fraction of voters in
$V$ who prefer $a$~to~$b$.
By $d_\swap(v,u)$ we denote the swap distance between votes $u$ and
$v$ (over the same candidate set $C$), i.e., the minimum number of
swaps of adjacent candidates needed to turn vote $u$ into vote $v$.
By $d_\spear(v,u)$ we denote the Spearman's distance between $v$
and~$u$. It is defined as $d_\spear(v,u) = \sum_{c \in C}|\pos_v(c) -
\pos_u(c)|$.

Let $C$ and $D$ be two equal-sized candidate sets.  If $v$ is a
preference order over $C$ and $\delta$ is a bijection in $\Pi(C,D)$,
then by $\delta(v)$ we mean a preference order obtained from $v$ by
replacing each candidate $c \in C$ with candidate $\delta(c) \in
D$. Similarly, for an election $E = (C,V)$, where
$V = (v_1, \ldots, v_n)$, by $\delta(V)$ we mean
$(\delta(v_1), \ldots, \delta(v_n))$, and by $\delta(E)$ we mean
$(\delta(C),\delta(V)) = (D, \delta(V))$.

\subsection{Structured Preferences}\label{sec:structured}

Our datasets will include  elections where the voters' preferences have some
particular structure, such as \emph{single-peaked} or
\emph{single-crossing} ones (or several others). Such elections have been frequently
studied and used in the literature \cite{elk-lac-pet:t:domains}  to model various features observed in
real-life scenarios. Further, many problems that are intractable in
the general case become solvable in polynomial time for structured
elections (see the overview by \citet{elk-lac-pet:t:domains} for a
detailed discussion of structured domains, including motivations as
well as their axiomatic and algorithmic properties).

Single-peaked preferences, introduced by
Black~\cite{bla:b:polsci:committees-elections}, capture settings where
it is possible to order the candidates in such a way that as we move
along this order, then each voter's appreciation for the candidates
first increases and then decreases (one example of such an order is
the classic left-to-right spectrum of political opinions; another one
regards preferences over temperatures in an office). Recently, Peters and
Lackner~\cite{pet-lac:j:spoc} introduced the notion of preferences
single-peaked on a circle, where instead of ordering the candidates as
a line, we arrange them cyclically (such preferences may appear, e.g.,
when choosing a video-conference time and the voters are in different
time zones).

\begin{definition}[Black~\cite{bla:b:polsci:committees-elections},
  Peters and Lackner~\cite{pet-lac:j:spoc}]\label{def:sp}
  Let $C$ be a set of candidates and let $c_1 \lhd c_2 \lhd \cdots
  \lhd c_m$ be a strict, total order over $C$, referred to as the
  \emph{societal axis}. Let $v$ be a preference order over $C$.  We
  say that $v$ is \emph{single-peaked} with respect to $\lhd$ if for
  each $\ell \in [m]$ the set of the $\ell$ top-ranked candidates
  according to $v$ forms an interval within $\lhd$.
  We say that $v$ is \emph{single-peaked on a circle (SPOC)} with respect to $\lhd$ if for each
  $\ell \in [m]$, the set of the $\ell$ top-ranked candidates either forms
  an interval within $\lhd$ or a complement of an interval.
  An election is single-peaked (is single-peaked on a circle)
  if there is an axis such that each voter's preference order
  is single-peaked (single-peaked on a circle) with respect to
  this axis.
\end{definition}

We note that if an election is single-peaked then each voter ranks one
of the two extreme candidates from the societal axis last. Elections
single-peaked on a circle are not restricted in this way.

\begin{example}

  Consider the candidate set $\{a,b,c,d\}$ and the axis $a \lhd b \lhd
  c \lhd d$. The Vote $b \pref c \pref a \pref d$ is
  single-peaked with respect to~$\lhd$ 
  (sets $\{b\}$,
  $\{b,c\}$, $\{a,b,c\}$, and $\{a,b,c,d\}$ form intervals
  within~$\lhd$), whereas 
  $b \pref a \pref d \pref c$
  is single-peaked on a circle with respect to~$\lhd$, but is not
  single-peaked ($\{a,b,d\}$ is not an interval within 
  $\lhd$). The Vote $a \pref c \pref b \pref d$ is not even single-peaked
  on a circle with respect~to~$\lhd$.
\end{example}

Single-crossingness is based on a similar principle as
single-peakedness, but with respect to the voters. The idea is that we
can order the voters in such a way that for each pair of candidates
$a$ and $b$, as we move along this order, the relative preference
over $a$ and $b$ changes at most once (in the left-to-right political
spectrum example, one could think that voters are ordered from the
extreme-left one to the extreme-right one, initially the voters prefer
the more left-wing candidate, and then they switch to the more
right-wing one). Single-crossing preferences were introduced by
Mirr\-lees~\cite{mir:j:single-crossing} and Roberts~\cite{rob:j:tax}
to capture preference orders that arise in the context of taxation.
\begin{definition}[Mirrlees~\cite{mir:j:single-crossing}, Roberts~\cite{rob:j:tax}]
  An election $E = (C,V)$ is single crossing if it is possible to
  order the voters in such a way that for each pair of candidates $a,b
  \in C$, the set of voters that prefer $a$ to $b$ either forms a
  prefix or a suffix of this order.
\end{definition}

We say that a set of preference orders $\calD$ is a \emph{single-crossing
  domain} if every election where each voter has a preference order
from $\calD$ is single-crossing.
For a discussion and analysis
of single-crossing domains, we point the reader to
the work of Puppe and Slinko~\cite{pup-sli:j:single-crossing}.

\begin{example}
  Consider the candidate set $C = \{a,b,c,d\}$ and the following votes:
  \begin{align*}
    v_1 \colon a \pref b \pref c \pref d, &&
                                            v_5 \colon c \pref b \pref d \pref a, \\
    v_2 \colon b \pref a \pref c \pref d, &&
                                            v_6 \colon c \pref d \pref b \pref a, \\
    v_3 \colon b \pref c \pref a \pref d, &&
                                            v_7 \colon d \pref c \pref b \pref a. \\
    v_4 \colon b \pref c \pref d \pref a,                                              
  \end{align*}
  The election $E = (C,V)$, where $V = (v_1, \ldots, v_7)$ is single-crossing. Since deleting
  or cloning voters cannot turn a single-crossing election into one that is not single-crossing,
  we see that $\calD = \{v_1, \ldots, v_7\}$ is a single-crossing domain. Note that $E$ is also a single-peaked election with respect to $a \lhd b \lhd
  c \lhd d$.\footnote{We want to remark that this is merely a coincidence, as there are single-peaked elections that are not single-crossing and the other way around. For instance, $(C,\{d \pref a \pref b \pref c,d \pref c \pref b \pref a\})$ is single-crossing but not single-peaked, whereas 
$(C,\{a \pref b \pref c \pref d,a \pref b \pref d \pref c, b \pref a \pref d \pref c,b \pref a \pref c \pref d\})$
  is single-peaked with respect to $c \pref a \pref b \pref d$ but not single-crossing. 
  }
\end{example}

We are also interested in group-separable elections, introduced by
\citet{ina:j:group-separable, ina:j:simple-majority}. Inada's original
definition says that an election is group-separable if each subset $A$
of at least three candidates from this election can be partitioned
into two nonempty subsets, $A'$ and $A''$, such that every voter
either ranks all members of $A'$ ahead of all members of $A''$ or the
other way round.  The intuition is that for every candidate set $A$,
some candidates have one feature (say, those in $A'$) and the others
have an opposite one (those in $A''$). Some voters prefer one variant
of the feature and some prefer the other. Indeed, these features are
organized hierarchically in a tree, as captured in a tree-based
definition of group-separable elections, provided by
\citet{kar:j:group-separable}. For examples of this feature-based
interpretation, as well as for computational properties of
group-separable elections, see the work of
\citet{fal-kar-obr:c:group-separable}. While these definitions are
equivalent, we will find the tree-based one more convenient and
present it in the following.

Fix a candidate set $C$ and consider an ordered, rooted tree $\calT$
with $|C|$ leaves, where each node either is a leaf or has at least
two children, ordered in a sequence (for the sake of simplicity, we will speak of these children as
ordered from left to right).
Further, each leaf has a unique candidate from $C$ as a label. 
We refer to such trees as $C$-based.  A frontier of $\calT$ is
the preference order obtained by reading the lables of the leaves from
left to right. A preference order $v$ is \emph{consistent} with $\calT$ if we
can obtain $v$ as a frontier of $\calT$ by reversing the order of
children of (some of) the internal nodes.

  \begin{figure}
    \begin{subfigure}[b]{0.49\textwidth}
      \centering
    \begin{tikzpicture}
      \Tree [ .$\calT_1$ [ [  [ .$a$ ] [ .$b$ ]  ] [  .$c$  ] ] [ [ .$d$ ] [ .$e$ ] ] ]      
    \end{tikzpicture}
      \caption{A balanced tree $\calT_1$.}
    \end{subfigure}    
    \begin{subfigure}[b]{0.49\textwidth}
      \centering
      \begin{tikzpicture}
        \tikzset{level distance=24pt}
        \tikzset{sibling distance=14pt}
        \Tree [ .$\calT_2$  $a$ [ $b$ [ $c$ [ $d$  $e$ ] ]  ] ] ]
      \end{tikzpicture}
      \caption{A caterpillar tree $\calT_2$.}
    \end{subfigure}    
    \caption{\label{fig:g-s}Trees from Example~\ref{ex:g-s}.}
  \end{figure}
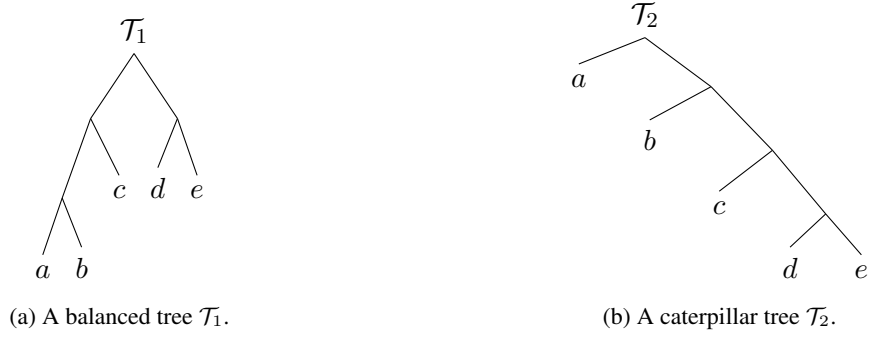

\begin{definition}[\citet{ina:j:group-separable,
    ina:j:simple-majority,kar:j:group-separable}]
  An election $E = (C,V)$ is group-separable if there is a $C$-based
  tree $\calT$ such that every vote from $V$ is consistent with
  $\calT$.
\end{definition}
We are
particularly interested in the following two types of group-separable
elections:
\begin{enumerate}
\item We say that an election is \emph{balanced group-separable} if it is
  group-separable with respect to a complete binary tree (where every level
  except possibly the last one is filled with nodes).
  
\item We say that an election is \emph{caterpillar group-separable} if
  it is group-separable with respect to a binary caterpillar tree (in
  such a tree, every node is either a leaf or has exactly two
  children, of which at least one is a leaf).
\end{enumerate}

\newcommand{\fvote}[5]{#1 \pref #2 \pref #3 \pref #4 \pref #5}

\begin{example}\label{ex:g-s}
  Consider the candidate set $C = \{a,b,c,d,e\}$,  the following votes
  \begin{align*}
    v_1 \colon \fvote a b c d e, &&
    v_2 \colon \fvote b a c d e, &&
    v_3 \colon \fvote e d c a b, \\
    v_4 \colon \fvote a b d e c, &&
    v_5 \colon \fvote b c e d a, &&
    v_6 \colon \fvote b c d e a, &&
  \end{align*}
  and the trees $\calT_1$ and $\calT_2$ from Figure~\ref{fig:g-s} (we see
  that $\calT_1$ is balanced while $\calT_2$ is caterpillar). Vote
  $v_1$ is consistent with both trees as it is their frontier. Votes
  $v_2$ and $v_3$ are consistent with $\calT_1$ but not $\calT_2$ (in
  particular, because for consistency with $\calT_2$ candidate $a$ has
  to be either ranked first or last), where $v_4$, $v_5$, and $v_6$
  are consistent with $\calT_2$ but not $\calT_1$ (in particular,
  because for consistency with $\calT_1$ candidates $a$, $b$, and $c$
  must be ranked consecutively).
\end{example}

Finally, we consider Euclidean preferences, discussed in detail, e.g.,
by Enelow and Hinich~\cite{enelow1984spatial,enelow1990advances}, which are
defined geometrically:
Each candidate and each voter corresponds to a point in
a Euclidean space and voters form their preferences by ranking the
candidates with respect to their distance.

\begin{definition}\label{def:t-euclidean}
  Let $t$ be a positive integer.  An election $E = (C,V)$ is
  $t$-Euclidean if it is possible to associate each candidate and each
  voter with his or her ideal point in $t$-dimensional Euclidean space
  $\mathbb{R}^t$ in such a way that the following holds: For each
  voter $v$ and each two candidates $a, b \in C$, $v$ prefers $a$ to
  $b$ if and only if $v$'s point is closer to the point of $a$ than to
  the point of $b$.
\end{definition}

It is well known that $1$-dimensional Euclidean elections are
both single-peaked and single-crossing. We also note that in a $2$-dimensional Euclidean election where
the ideal points are arranged on a circle, the voters have SPOC
preferences.

\subsection{Statistical Cultures}\label{sec:cultures}

Below we describe a number of ways of generating random elections
(i.e., statistical cultures).  For each of the models we either
describe explicitly how an election with $m$ candidates and $n$ voters
is generated, or we describe the process of generating a single vote
(and then it is implicit that this process is repeated $n$
times to obtain an election).\smallskip

\begin{description}
  \item[Impartial Culture and Related Models.]
Under the impartial culture model (IC), every preference order is sampled
with the same probability. That is, to generate a vote we choose a
preference order uniformly at random.
Under impartial anonymous culture (IAC), we require that each
\emph{voting situation} appears with the same probability (a voting
situation specifies how many votes with a given preference order are
present in a profile; thus IAC generates anonymized
preference profiles uniformly at random).
The Impartial Anonymous Neutral Culture (IANC) model additionally
abstracts away from the names of the
candidates~\cite{ege-gir:j:isomorphism-ianc}. For elections with 100
candidates, which are our main focus, these three models are nearly
the same, so we only consider IC.

\item[Pólya-Eggenberger Urn Model.]
  The Pólya-Eggenberger urn
  model~\cite{berg1985paradox,mcc-sli:j:similarity-rules} is
  parametrized with a nonnegative number $\alpha$, known as the level of
  contagion, and proceeds as follows: Initially, we have an urn with
  one copy of each of the $m!$ possible preference orders. To generate
  a vote, we draw a preference order from the 
  urn uniformly at random (this is the generated vote), and 
  we return it to the urn together with additional $\alpha m!$
  copies. For larger $\alpha$s 
  the generated votes are more correlated.  For $\alpha = 0$ the model
  is equivalent to IC, for $\alpha = \nicefrac{1}{m!}$ it is
  equivalent to IAC, and for $\alpha = \infty$ it produces unanimous
  elections.

\item[Mallows Model.]
  The Mallows model~\cite{mal:j:mallows} is parameterized by the
  dispersion parameter $\phi \in [0,1]$ and the center preference
  order $v$ (we choose it uniformly at random and then use it for all
  the generated votes). We generate each vote independently at random,
  where the probability of generating vote $u$ is proportional to
  $\phi^{d_\swap(v,u)}$. For $\phi = 1$, the model is equivalent to
  impartial culture, whereas for $\phi = 0$ all generated votes are
  identical to the center $v$.  See the work of Lu and Boutilier for
  an effective algorithm for sampling from the Mallows
  model~\cite{lu-bou:j:sampling-mallows}.  Instead of using the
  dispersion parameter directly, we will use its normalized variant,
  which we discuss in \Cref{sec:datasets}; we also point the reader to
  the recent work of \citet{boe-fal-kra:c:mallows-normalization} for a
  comparison of the two Mallows parameterizations.

  \smallskip
		
\item[Single-Peaked Models.]
We consider two ways of generating single-peaked elections, one
studied by Walsh~\cite{wal:t:generate-sp} and one studied by
Conitzer~\cite{con:j:eliciting-singlepeaked}; 
thus we refer to them as the \emph{Walsh model} and the \emph{Conitzer
  model}. 
Under both models, we first choose the axis (uniformly at
random). To generate a vote, we proceed as follows:
\begin{enumerate}
\item Under Walsh's model, we choose a single-peaked preference
  order (under the given axis) uniformly at random.
  Walsh~\cite{wal:t:generate-sp} provided a sampling algorithm for
  this task.

\item To generate a vote under the Conitzer model for the axis $c_1
  \lhd c_2 \lhd \cdots \lhd c_m$, we first choose some candidate $c_i$
  (uniformly at random) to be ranked on top (so, at this point, $c_i$
  is the only ranked candidate). Then, we perform $m-1$ steps
  as follows: Let $\{c_j, c_{j+1}, \ldots,
  c_{k}\}$ be the set of the currently ranked candidates. We choose
  the next-ranked candidate from the set $\{c_{j-1}, c_{k+1}\}$
  uniformly at random (if $j=1$ or $k=m$ then this set contains a single element and randomization
  is not needed).
\end{enumerate}
The Conitzer model is sometimes referred to as the \emph{random peak}
model, whereas the Walsh model can be dubbed as \emph{impartial
  culture over single-peaked votes}.

To generate a single-peaked on a circle vote, we use the Conitzer
model, except that we take into account that the axis is cyclical
(consequently, each SPOC vote is equally likely to be generated and,
hence, we could also refer to this model as impartial culture over SPOC
votes, or as the Walsh model for SPOC elections).  \smallskip

\item[Single-Crossing Models.]
  We would like to generate single-crossing elections uniformly at
  random, 
  but we are not aware of an efficient sampling algorithm.
  Thus, to generate a single-crossing election, we first generate a
  single-crossing domain $\calD$ and then draw $n$ votes from it
  uniformly at random. To generate this 
  domain for a candidate set $C = \{c_1, \ldots, c_m\}$, we use the
  following procedure:
  \begin{enumerate}
  \item We let $v$ be a preference order
    $c_1 \pref c_2 \pref \cdots \pref c_m$ and we output $v$ as the
    first member of our domain.
  \item We repeat the following steps until we output
    $c_m \pref c_{m-1} \pref \cdots \pref c_1$: (a) We draw candidate
    $c_j$ uniformly at random and we let $c_i$ be the candidate ranked
    right ahead of $c_j$ in $v$ (if $c_j$ is ranked on top, then we
    repeat); (b) If $i < j$ then we swap $c_i$ and $c_j$ in $v$ and
    output the new preference order.
  \item We randomly permute the names of the candidates.
  \end{enumerate}
  Our 
  domains always have cardinality $(\nicefrac{1}{2})m(m-1)+1$.  We
  believe that studying the process of sampling single-crossing
  elections (uniformly at random) in more detail is an interesting
  direction for future research. The approach given here is quite ad
  hoc, and we use it due to the lack of better alternatives.

\item[Group-Separable Model.] Given a candidate set $C$ and a
  $C$-based tree $\calT$, we generate each vote independently as
  follows: For each of the internal nodes in $\calT$, we reverse the
  order of its children with probability $\nicefrac{1}{2}$. The
  frontier of the resulting tree is the generated vote.  This process
  generates each vote compatible with a given tree with the same
  probability and, hence, it is a form of impartial culture over
  group-separable votes.  We consider either balanced trees or
  caterpillar trees.

\item[Euclidean Models.]
To generate a $t$-Euclidean election, we choose the ideal
points for the candidates and the voters, and then derive the voters'
preferences as in Definition~\ref{def:t-euclidean}. Given $t \in \{1,
2, \ldots \}$, we consider the following two ways of generating the
ideal points:
\begin{enumerate}
\item In the $t$-dimensional hypercube model ($t$D-Hypercube), we choose all the ideal points
  uniformly at random from 
  $[-1,1]^t$.
\item In the $t$-dimensional hypersphere model ($t$D-Hypersphere), we
  choose all the ideal points uniformly at random from the $t$-dimensional hypersphere
  centered at $(0, \ldots, 0)$, with radius $1$.
\end{enumerate}
For $t \in \{1,2,3\}, $ we refer to $t$D-Hypercube models as
1D-Interval, 2D-Square, and 3D-Cube, respectively. Similarly, by
Circle and Sphere we mean the $t$D-Hypersphere models for $t
\in \{2,3\}$.\smallskip

\end{description}

For the urn and Mallows models it is not completely clear what values
of the contagion and dispersion parameters to use. In our later
discussion we will suggest how to pick the parameter values to cover a broad spectrum of
elections that can be generated from these models. For the other
models, either there are no parameters to set, or we proposed
particular approaches above (such as using uniform distributions on
the hypercubes and hyperspheres for Euclidean elections, or the
particular type of trees for group-separable elections). These
choices are, of course, somewhat arbitrary, but we believe that they,
at least, cover some possible extreme approaches.

\subsection{(Isomorphic) Distances Between Elections}\label{sec:distances}

For a given set $X$, a pseudometric 
over $X$ is a function $d \colon X \times X \rightarrow \reals_+$ such
that for each $x,y,z \in X$ it holds that
\begin{inparaenum}[(1)]
\item $d(x,x) = 0$, 
\item $d(x,y) = d(y,x)$, and 
\item $d(x,z) \leq d(x,y) + d(y,z)$. 
\end{inparaenum}
In our case, we take $X$ to be the set of all elections with a given
number 
of candidates and a given number 
of voters.  As our goal is to compare elections generated from
statistical cultures, where the names of the candidates or the voters
are chosen randomly, we require distances to be invariant under
permutations of the names of candidates and voters. We refer to
such distances as neutral/anonymous (\emph{neutrality} refers to
invariance with respect to permuting candidate names and
\emph{anonymity} has an analogous meaning for the case of voters).

Recently, \citet{fal-sko-sli-szu-tal:j:isomorphism} introduced several
pseudometrics between elections that satisfy our basic anonymity/neutrality conditions.  The
main idea is that given two elections,
we find mappings between their candidates and between their voters,
and then we sum up the distances between the matched pairs of
 votes assuming the candidate mapping (using, e.g., the swap distance or the Spearman's distance);
we seek mappings that give the smallest final distance. These distances are
known as the isomorphic swap distance and the isomorphic Spearman distance.
The names stem from the underlying distance between the votes and the
fact that two elections are isomorphic---i.e., one can be obtained
from the other by renaming the candidates and reordering the
voters---if and only if their distances are equal to zero.  While
these distances are intuitively very appealing---they capture even the
smallest differences in the structure between elections and have very
natural interpretations---they are $\np$-hard to compute, hard to
approximate, and the known $\fpt$ algorithms are too
slow~\citep{fal-sko-sli-szu-tal:j:isomorphism}.

\section{Positionwise Distance and the Backbone Map}\label{sec:positionwise}
In this section we define our main tool, the positionwise distance,
and apply it to form our first map.  The idea of the positionwise
distance is that given two elections we first derive their aggregate
representations (as position matrices) and then compute the distances
between these representations. This way we lose some precision as
compared to the isomorphic swap or Spearman distances (e.g., because
some nonisomorphic elections have the same aggregate representations
and, hence, are at positionwise distance zero) but we gain efficient,
polynomial-time computability. In this section we first define and
analyze the position and frequency matrices, i.e., the aggregate
representations that we use and which facilitates the low
computational complexity of our distance,\footnote{Indeed, just using
  an aggregate representation does not need to lead to a
  polynomial-time computable distance. In one of the conference papers
  on which this paper is based we have studied a distance based on the
  weighted pairwise majority matrix, which turned out to be
  $\np$-complete to compute~\citep{szu-fal-sko-sli-tal:c:map}. For
  more details on this distance, see the work of
  \citet{boe-fal-nie-szu-was:c:metrics}.}  and then we define the
positionwise distance and provide the basic setup for the map of
elections. Finally, we reflect on why we believe that positionwise
distance is a good choice for our maps (for elections with larger
numbers of candidates).

\subsection{Position and Frequency Matrices}

Given an election, its position matrix specifies for each candidate
how many voters rank him or her at each possible position. Frequency matrices
are defined analogously, but focus on the fractions of votes instead of absolute vote counts.
Formally, let $E = (C,V)$ be an election, where
$C = \{c_1, \ldots, c_m\}$ and $V = (v_1, \ldots, v_n)$.  For a
candidate $c \in C$ and position $i \in [m]$, we write~$\#\pos_E(c,i)$
to denote the number of voters in election $E$ that rank $c$ on
position~$i$, and by~$\#\pos_E(c)$ we mean the vector:
\[
  ( \#\pos_E(c,1), \#\pos_E(c,2), \ldots, \#\pos_E(c,m) ).
\]
A \emph{position matrix} of election $E$, denoted $\#\pos(E)$, is the
$m \times m$ matrix that has
vectors~$\#\pos_E(c_1), \ldots, \#\pos_E(c_m)$ as its columns.\footnote{Technically, due to different possible orderings of the candidates,
an election may have several different position matrices; for our
purposes we typically assume some arbitrary order because the
positionwise distance will internally reorder the candidates as needed.}
Frequency matrices are defined analogously: For a candidate $c$ and a position $i \in [m]$, let
$\#\relvot_E(c,i)$ be~$\frac{\#\pos_E(c,i)}{n}$, let vector
$\#\relvot_E(c)$ be~$ ( \#\relvot_E(c,1), \ldots, \#\relvot_E(c,m) )$,
and let each \emph{frequency matrix} of election $E$,
denoted~$\#\relvot(E)$, consist of
columns~$\#\relvot_E(c_1), \ldots, \#\relvot_E(c_m)$.
In other words, a frequency matrix is a position matrix normalized by the number of voters.

Note that in each position matrix, each row and each column
sums up to the number of voters in the election. Similarly, in each
frequency matrix, the rows and columns sum up to one (such matrices
are called bistochastic).
For a positive integer $m$, we write $\calF(m)$ [$\calP(m)$] to denote the set of
all $m \times m$ frequency [position] matrices.

\begin{example}\label{ex:matrices}
  Consider an  election $E = (C,V)$, where $C = \{a,b,c,d\}$,
  $V = (v_1,v_2,v_3)$, and the voters have
  the following preference orders:  
  \begin{align*}
    v_1 \colon a \pref b \pref c \pref d, && 
    v_2 \colon a \pref c \pref b \pref d, && 
    v_3 \colon d \pref b \pref c \pref a.
  \end{align*}
  The position and frequency matrices of this election are:
  \begin{align*}
    \#\pos(E) = 
     \kbordermatrix{ & a & b & c & d \\
    1 &                2 & 0 & 0 & 1 \\
    2 &                0 & 2 & 1 & 0 \\
    3 &                0 & 1 & 2 & 0 \\
    4 &                1 & 0 & 0 & 2       
   }
                               \text{\quad and\quad}
                               \#\relvot(E) = 
     \kbordermatrix{ & a & b & c & d \\
    1 &                \nicefrac{2}{3} & 0               & 0                   & \nicefrac{1}{3} \\
    2 &                0               & \nicefrac{2}{3} & \nicefrac{1}{3}     & 0 \\
    3 &                0               & \nicefrac{1}{3} & \nicefrac{2}{3}     & 0 \\
    4 &                \nicefrac{1}{3} & 0               & 0                   & \nicefrac{2}{3}                                                                    }.
  \end{align*}  
\end{example}
Naturally, two distinct elections can have identical position matrices
(this issue was recently studied in detail by
\citet{boe-cai-fal-fan-jan-kac:c:position-matrices}, who has shown
that counting the number of non-isomorphic elections with a given position matrix is
$\sharpp$-complete, but it is possible to deduce some nontrivial
properties of the elections with a given matrix).
\begin{example}\label{ex:non-unique-matrix}
  Consider
  election 
  $E' = (C,U)$, where $C = \{a,b,c,d\}$,
  $U = (u_1, u_2, u_3)$, and the voters have
  preference orders:  
  \begin{align*}
    u_1 \colon a \pref b \pref c \pref d, &&
    u_2 \colon a \pref b \pref c \pref d, &&
    u_3 \colon d \pref c \pref b \pref a. 
  \end{align*}
  This election has the same position (and frequency) matrix as $E$ in
  Example~\ref{ex:matrices}.  However, $E$ and $E'$ are not isomorphic
  because all votes in $E$ are different, whereas $E'$ contains two
  identical ones.
\end{example}

\subsection{Distances Among Vectors}

For our new distance, we will need distances among vectors. Given two
vectors, $x = (x_1, \ldots, x_t)$ and $y = (y_1, \ldots, y_t)$, their
$\ell_1$-distance is:
\begin{align*}
  \ell_1(x,y) = |x_1-y_1| + |x_2-y_2| + \cdots +|x_t - y_t|.
\end{align*}
To define the \emph{earth mover's distance} (EMD)~\citep{rub-tom-gui:j:emd},
we additionally require that the entries of these vectors are
nonnegative and sum up to the same value. Then their earth mover's
distance, denoted $\EMD(x,y)$, is defined as the lowest total cost of
operations that transform vector~$x$ into vector $y$, where each
operation is of the form ``\emph{subtract $\delta$ from position~$i$
  and add $\delta$ to position~$j$}'' and costs $\delta \cdot
|i-j|$. Such an operation is legal if the current value at
position~$i$ is at least $\delta$.

It is well-known that $\EMD(x,y)$ can be computed in polynomial time
using a simple greedy algorithm. For each vector
$z = (z_1, \ldots, z_t)$, let $\hat{z}$ be its prefix-sum vector,
i.e., $\hat{z} = (z_1, z_1+z_2, \ldots, z_1 + z_2 + \cdots + z_t)$.
Then, we have that
$\EMD(x,y) = \ell_1(\hat{x}, \hat{y})$~\citep{rub-tom-gui:j:emd}.

\subsection{Positionwise Distance}

We are now ready to define the positionwise distance. Its main
underlying principle is that the most valuable information about each
candidate can be extracted from the positions that this candidate
occupies in the voters' preference rankings.

\begin{definition}\label{def:poswise} 
  Let $E = (C,V)$ and $F = (D,U)$ be two elections with $m$ candidates
  each (we do not require that $|V| = |U|$).  Let $\delta$ be a bijection
  in $\Pi(C,D)$. We define:
  \[
    \textstyle
    \delta\hbox{-}\POS(E,F) := \sum_{c \in C}  \EMD(\#\relvot_E(c), \#\relvot_F(\delta(c))).
  \]
  The positionwise
  distance between~$E$ and~$F$, denoted $\POS(E,F)$, is defined as:
  \[
    \textstyle \POS(E,F) := \min_{\delta \in \Pi(C,D)} \delta\hbox{-}\POS(E,F).
  \]
  We refer to $\delta$ as the \emph{candidate matching} that implements
  (or, witnesses) the distance between $E$ and $F$.
\end{definition}

\noindent In other words, the positionwise distance is the sum of the
earth mover's distances between the frequency vectors of the
candidates from the two elections, with candidates/columns matched
optimally according to $\sigma$.  The positionwise distance is
invariant to renaming the candidates and reordering the voters.

\begin{example}\label{ex:poswise}
  Consider election $E$ from \Cref{ex:matrices} and election
  $F = (D,U)$ with candidate set $D = \{x,y,z,w\}$ and voter
  collection $U = (u_1, u_2, u_3)$, where the voters have preference
  orders:
  \begin{align*}
    u_1 \colon x \pref y \pref z \pref w, &&
    u_2 \colon y \pref x \pref w \pref z, &&
    u_3 \colon w \pref z \pref x \pref y.
  \end{align*}
  These elections have the following frequency matrices:
  \begin{align*}
    \#\relvot(E) = 
     \kbordermatrix{ & a & b & c & d \\
    1 &                \nicefrac{2}{3} & 0 & 0 & \nicefrac{1}{3} \\
    2 &                0 & \nicefrac{2}{3} & \nicefrac{1}{3} & 0 \\
    3 &                0 & \nicefrac{1}{3} & \nicefrac{2}{3} & 0 \\
    4 &                \nicefrac{1}{3} & 0 & 0 & \nicefrac{2}{3}       
    },
      &&
    \#\relvot(F) = 
     \kbordermatrix{ & x & y & z & w \\
    1 &                \nicefrac{1}{3} & \nicefrac{1}{3} & 0               & \nicefrac{1}{3} \\
    2 &                \nicefrac{1}{3} & \nicefrac{1}{3} & \nicefrac{1}{3} & 0 \\
    3 &                \nicefrac{1}{3} & 0               & \nicefrac{1}{3} & \nicefrac{1}{3} \\
    4 &                0               & \nicefrac{1}{3} & \nicefrac{1}{3} & \nicefrac{1}{3}       
    }.
  \end{align*}
  Let us consider a mapping $\sigma$ such that $\sigma(a) = y$,
  $\sigma(b) = x$, $\sigma(c) = z$, and $\sigma(d) = w$. We see
  that:
  \begin{align*}
    \EMD( \#\relvot_E(a), \#\relvot_F(\sigma(a))) =
    \EMD( (\nicefrac{2}{3},0,0,\nicefrac{1}{3}), (\nicefrac{1}{3},\nicefrac{1}{3},0,\nicefrac{1}{3}) = \nicefrac{1}{3}, \\
    \EMD( \#\relvot_E(b), \#\relvot_F(\sigma(b))) =
    \EMD( (0, \nicefrac{2}{3},\nicefrac{1}{3},0), (\nicefrac{1}{3},\nicefrac{1}{3},\nicefrac{1}{3},0) = \nicefrac{1}{3}, \\    
    \EMD( \#\relvot_E(c), \#\relvot_F(\sigma(c))) =
    \EMD( (0, \nicefrac{1}{3},\nicefrac{2}{3},0), (0,\nicefrac{1}{3},\nicefrac{1}{3},\nicefrac{1}{3}) = \nicefrac{1}{3}, \\    
    \EMD( \#\relvot_E(d), \#\relvot_F(\sigma(d))) =
    \EMD( (\nicefrac{1}{3},0,0,\nicefrac{2}{3}), (\nicefrac{1}{3},0,\nicefrac{1}{3},\nicefrac{1}{3}) = \nicefrac{1}{3}. 
  \end{align*}
  Hence, for this mapping, the positionwise distance is
  $\nicefrac{1}{3} + \nicefrac{1}{3} + \nicefrac{1}{3} +
  \nicefrac{1}{3} = \nicefrac{4}{3}$. This mapping is, indeed,
  optimal. To see this, note that each column of $\#\relvot(E)$
  includes the value $\nicefrac{2}{3}$, whereas the largest value in each
  column of $\#\relvot(F)$ is $\nicefrac{1}{3}$. Hence, the smallest
  possible EMD distance between a column of $\#\relvot(E)$ and a
  column of $\#\relvot(F)$ is at least
  $\nicefrac{1}{3}$. Consequently, we have
  $\POS(E,F) = \nicefrac{4}{3}$.
\end{example}

\begin{remark}\label{rem:norm-pos}
  To obtain the \emph{normalized positionwise distance} between two
  elections with $m$ candidates each (where $m$ is divisible by $4$),
  we divide their positionwise distance by $\frac{1}{3}(m^2-1)$. This
  normalization factor is the largest positionwise distance between
  two elections with $m$
  candidates~\citep{boe-fal-nie-szu-was:c:metrics}. \Cref{pr:calc} and
  the discussion below will make this normalization more intuitive.
\end{remark}

\Cref{def:poswise} calls for some explanations. First, let us mention
that we could have used any other distance between vectors instead of
EMD and, indeed, \citet{boe-fal-nie-szu-was:c:metrics} evaluate the
effects of using~$\ell_1$ (and conclude that the practical differences
are not huge, but still EMD is preferable). We use EMD because
it captures the idea that being ranked on the top position is more
similar to being ranked on the second position than to being ranked on
the bottom one.

The second issue is that we could have used position matrices instead
of the frequency ones. The difference would only be technical. Indeed,
if we used position matrices then we would have to assume that the
elections have the same numbers of voters (in our maps this holds
anyway) and the obtained distances would be multiplied by this number
of voters. We use frequency matrices because of their increased
flexibility (this will become clear, e.g., in \Cref{sec:four}).

The third issue is that we should formally argue that the positionwise
distance is, indeed, a pseudometric. We do so in the next proposition.

\begin{proposition}
  The positionwise distance is a pseudometric.
\end{proposition}

\begin{proof}
  We show that the positionwise distance satisfies the triangle inequality
  (the other requirements for being a pseudometric are easy to verify).
  Consider three elections with candidate sets of equal size,
  $E_1 = (C_1,V_1)$, $E_2 = (C_2,V_2)$, and $E_3 = (C_3,V_3)$. Let
  $\delta$ and $\sigma$ be the candidate matchings that minimize the
  EMD distances between $E_1$ and $E_2$ and between $E_2$ and $E_3$,
  respectively. We have that:
  \begin{align*}
    \POS(E_1, E_3) &\leq \textstyle \sum_{c \in C_1} \EMD( \#\relvot_{E_1}(c), \#\relvot_{E_3}(\sigma(\delta(c)))) \\
    &\textstyle \leq \sum_{c \in C_1} \EMD( \#\relvot_{E_1}(c), \#\relvot_{E_2}(\delta(c))) \\
    &\textstyle + \sum_{c \in C_2} \EMD( \#\relvot_{E_2}(\delta(c)), \#\relvot_{E_3}(\sigma(\delta(c)))) \\
    &\textstyle =\POS(E_1, E_2) + \POS(E_2, E_3) \text{.}
  \end{align*}
  The first inequality follows from the definition of the positionwise
  distance, the second one from the fact that EMD is a metric.
\end{proof}

Last but not least, a crucial feature of the positionwise distance is
that we can compute it in polynomial time. In short, doing so reduces
to finding a minimum-cost matching in a certain weighted bipartite graph.

\begin{proposition}\label{pro:pos-poly}
  There is a polynomial-time algorithm for computing the positionwise
  distance.
\end{proposition}
\begin{proof}
  Let $E = (C,V)$ and $F = (D,U)$ be two elections where $|C| = |D|$.
  The value of $\POS(E,F)$ is equal to the minimum-cost matching in
  the bipartite graph whose vertex set is $C \cup D$ and which has the
  following edges: For each $c \in C$ and each $d \in D$ there is an
  edge with the cost equal to the EMD distance between
  $\#\relvot_E(c)$ and $\#\relvot_F(d)$; these weights can be computed
  independently for each pair of candidates. Such minimum-cost
  matchings can be computed in polynomial time (see, e.g., the
  overview of \citet{ahu-mag-orl:b:flows}).
\end{proof}

\subsection{Recovering Elections from Matrices}

Since the positionwise distance works on frequency matrices, it will
sometimes be convenient to operate directly on such matrices rather
than on the elections that generate them. To this end, we now argue
that given a position or a frequency matrix, we can always compute in
polynomial time some election that either generates it (for position
matrices) or is very close to generating it (for frequency matrices).

First, we  observe that each $m\times m$ position matrix has a corresponding
$m$-candidate election with at most $m^2-2m+2$ distinct preference
orders. This was shown by \citet[Theorem 7]{leep1999marriage} (they
speak of ``semi-magic squares'' and not ``position matrices'' and
show a decomposition of a matrix into permutation matrices, which
correspond to votes in our setting).
Their proof
lacks some algorithmic details which we provide in the appendix.
\begin{restatable}[\appsymb]{proposition}{polyalgo}\label{pro:poly-algo}
  Given a position matrix $X\in \calP(m)$, one can compute
  in $O(m^{4.5})$ time  an election~$E$ that contains at most $m^2-2m+2$
  different votes such that $\#\pos(E) =
  X$. 
\end{restatable}

Next, we consider the issue of recovering elections based on frequency
matrices. Given an $m \times m$ bistochastic matrix $X$ and a number
$n$ of voters, we would like to find an election~$E$ with position
matrix $nX$. This may be impossible as $nX$ may have fractional
entries, but we can get very close to this~goal.
The next proposition shows how to achieve it, and justifies speaking
of elections and frequency matrices interchangeably.
By $\lfloor nX \rfloor$, we mean the matrix whose each entry is equal to the
floor of the corresponding entry of $nX$.

\begin{restatable}{proposition}{freqtopos}\label{pro:frequency}
  Given an $m \times m$ bistochastic matrix $X$ and an integer~$n$,
  one can compute in polynomial time an election~$E$ with $n$ voters
  whose position matrix~$P$ satisfies $|nx_{i,j} - p_{i,j}| < 1$
  for each $i, j \in [m]$ and, under this condition, minimizes the
  value $\sum_{1 \leq i,j \leq m} |nx_{i,j} - p_{i,j}|$.
\end{restatable}
\begin{proof}
  We first design a randomized algorithm that uses dependent
  rounding. Then we derive a deterministic algorithm, based on
  computing min-cost flows, which also performs the minimization step
  (however, we need the former algorithm to explain why the latter
  always produces a correct result).
  We start by computing matrix $Y$
  where each entry $y_{i,j}$ is equal to:
  \[ y_{i,j} = nx_{i,j} - \lfloor nx_{i,j}\rfloor. \]
  All entries of $Y$ are between~$0$
  and~$1$, and each row and each column
  of~$Y$ sums up to an integer because the latter property holds both
  for $nX$ and $\lfloor
  nX\rfloor$ (although different rows and columns may sum up to
  different integers).  
  We construct an edge-weighted bipartite
  graph~$G$ with vertex sets~$A = \{a_1, \ldots, a_m\}$ and~$B =
  \{b_1, \ldots, b_m\}$. For each two vertices~$a_i$
  and~$b_j$, we have a connecting edge of
  weight~$y_{i,j}$. For each vertex $c \in A \cup
  B$, we let its fractional
  degree~$\delta_G(c)$ be the sum of the weights of the edges incident
  to it.  Then, we invoke the dependent rounding procedure of
  \citet[Theorem~2.3]{gan-khu-par-sri:j:dependent-rounding} on this
  graph: Consequently, in polynomial time we obtain an unweighted
  bipartite graph
  $G'$ with the same two vertex sets, such that the (standard) degree
  of each vertex $c \in A \cup B$ in $G'$ is equal to $\delta_G(c)$
  (this is property P2 in the paper of
  \citet{gan-khu-par-sri:j:dependent-rounding}; note that dependent
  rounding is computed via a randomized algorithm, but this condition
  on the degrees is always satisfied, independently of the random bits
  selected).\footnote{Dependent rounding accepts input exactly in the
    format that we have and implements exactly the effect that we
    describe (plus additional guarantees that we do not need).}
  Using $G'$, we form an $m \times m$ matrix
  $D$ such that for each $i,j \in [m]$,~$d_{i,j}$ is $1$ if
  $G'$ contains an edge between $a_i$ and $b_j$, and $d_{i,j} =
  0$ otherwise.  Finally, we compute matrix $P = \lfloor nX \rfloor +
  D$.

  The entries of~$P$ differ from those of~$nX$ by less than one, and
  the rows and columns of $P$ sum up to~$n$ (to see it, consider the
  degrees of the vertices in~$G'$).  So, we obtain the desired
  election by invoking Proposition~\ref{pro:poly-algo} on matrix $P$.

  Next, we provide a fully deterministic algorithm for computing the
  matrix $D$ (and, consequently, matrix $P$), which does not invoke
  dependent rounding and which minimizes the value
  $\sum_{1 \leq i,j \leq m} |nx_{i,j}-p_{i,j}|$.  Consider matrix~$Y$
  from the first part of the proof and let
  $Z = \sum_{1 \leq i,j \leq m} y_{i,j}$.

  We form a flow network with source~$s$, nodes $v_{i,j}$ for each
  $i,j \in [m]$, ``pre-sink'' nodes $t_1, \ldots, t_m$, and sink
  node~$t$. We have the following edges (unless specified otherwise,
  all of them have cost $0$):
  \begin{enumerate}
  \item For each $i \in [m]$, we have a directed path which starts at
    the source node~$s$, then goes to $v_{i,1}$, next to $v_{i,2}$,
    and so on, until $v_{i,m}$. Each edge on this path has capacity
    equal to $Y_i = \sum_{j=1}^m y_{i,j}$. Recall that this value is
    an integer and note that $Y_i$ is exactly the number of entries in
    the $i$-th row of matrix $D$ that need to have value $1$ (the
    other entries in this row will have value $0$). 

  \item For each $i, j \in [m]$, we have an edge from $v_{i,j}$ to
    $t_j$, with capacity one and with cost $1-2y_{i,j}$. The intuition
    is that if a unit of flow goes from $v_{i,j}$ to $t_j$, then
    $d_{i,j}=1$, and if the entire flow that enters $v_{i,j}$ goes to
    $v_{i,j+1}$, then $d_{i,j} = 0$. The intuition for the cost of the
    edge is that if $d_{i,j} = 0$ then
    $p_{i,j} = \lfloor nx_{i,j} \rfloor$ and, consequently,
    $|nx_{i,j} - p_{i,j}| = |nx_{i,j} - \lfloor nx_{i,j}\rfloor| =
    y_{i,j}$, which we treat as the ``default.'' If we send a unit of
    flow from $v_{i,j}$ to $t_j$ then $d_{i,j} = 1$ and
    $p_{i,j} = \lfloor nx_{i,j} \rfloor + 1$.  Consequently,
    $|nx_{i,j} - p_{i,j}| = |nx_{i,j} - 1 - \lfloor nx_{i,j}\rfloor| =
    1 + \lfloor nx_{i,j}\rfloor - nx_{i,j} = 1-y_{i,j}$.  So the
    difference between the default value $y_{i,j}$ of
    $|nx_{i,j} - p_{i,j}|$ that we get for $d_{i,j}=0$ and its value
    $1-y_{i,j}$ that we get for $d_{i,j} = 1$ is $1-2y_{i,j}$.  Note
    that $1-2y_{i,j}$ may be negative, but this is not an issue for the flow computation as our
    flow network will not have cycles.
     
  \item Finally, for each $j \in [m]$, we have an edge from $t_j$ to
    $t$ with capacity $\sum_{i=1}^m y_{i,j}$. This value is an integer
    and it is the number of entries in the $j$-th column of matrix $D$
    that need to become $1$ (instead of being $0$).
  \end{enumerate}
  Next, we compute in polynomial time an integral flow that moves
  $\sum_{i,j\in[m]} y_{i,j}$ units of flow from $s$ to $t$ at the
  lowest possible cost (which we denote $Z_f$; we use an arbitrary
  polynomial-time algorithm that accepts negative costs, such as one
  of those available in the overview of
  \citet{ahu-mag-orl:b:flows}). If this flow existed, then we could
  compute matrix~$D$ by setting, for each $i, j \in [m]$, $d_{i,j}$ to
  be the amount of flow ($0$ or~$1$) going from $v_{i,j}$
  to~$t_j$. Indeed, by definition of our flow network, for each~$i$ we
  would have that $\sum_{j=1}^m d_{i,j} = Y_i= \sum_{j=1}^m y_{i,j}$
  (because a flow of $Y_i$ needs to enter $v_{i,1}$ and when a flow enters some node $v_{i,j}$, then it can either
  go to $t_j$ or to $v_{i,j+1}$). Due to the capacities on the edges
  from the pre-sink nodes to the sink, for each $j \in [m]$ we would
  also have that $\sum_{i=1}^m d_{i,j} = \sum_{i=1}^m y_{i,j}$.

  We will now argue
  that our flow problem has a solution. For this, we reexamine the
  graphs $G$ and $G'$ constructed in the first part of the proof. Note
  that the two above-discussed properties of the flow network are
  equivalent to ensuring that the degrees of the vertices in~$G'$ are
  equal to the fractional degrees in~$G$, as is done by dependent
  rounding. This means that given graph $G'$ computed by our first
  algorithm and the matrix $D$ induced by $G'$, we can compute a
  capacity-respecting flow for our network by sending one unit of flow
  from $v_{i,j}$ to $t_j$ if $d_{i,j}=1$. This flow moves
  $\sum_{i,j\in[m]} y_{i,j}$ units of flow and accordingly, our flow
  problem has a solution.
 
  Further, for a flow $f$ implying values $d_{i,j}$, we have that:
  \[
    \textstyle \sum_{1 \leq i,j \leq m}|nx_{i,j} - p_{i,j}| = \sum_{1
        \leq i,j \leq m} |nx_{i,j} - (\lfloor nx_{i,j} \rfloor +
      d_{i,j})| = Z + Z_f.
  \]
  To see why this is the case, note that if all the values $d_{i,j}$,
  as well as $Z_f$, were $0$, then the equality would hold by the definition of $Z$. Now, for
  each $d_{i,j}=1$, the left-hand side increases by $1-2y_{i,j}$ and,
  consequently, the whole sum is equal to $Z + Z_f$.  Thus minimizing
  $Z_f$ is equivalent to minimizing
  $\sum_{1 \leq i,j \leq m}|nx_{i,j} - p_{i,j}|$, which is our goal.
\end{proof}

Together with the fact that computing a frequency matrix for a given
election is straightforward, Proposition~\ref{pro:frequency} gives a
two-way interface between elections and frequency matrices.  In the
remainder of the paper, we will also express many (idealized families
of) elections by focusing on their frequency matrices and we will
often speak of elections and their frequency matrices interchangeably.

\subsection{Four Compass Elections/Matrices}\label{sec:four}
In this section we prepare our ``compass,'' i.e., we identify four
special matrices (or, families of elections) that correspond to
different types of (dis)agreement among the voters.
Importantly, we want our compass to form a backbone of the maps and,
consequently, we want them to be at large and close-to-constant
distances from each other as we consider different numbers of
candidates (so that irrespective of the sizes of the considered
elections, the compass matrices could be viewed as fixed points on the
maps).
We argue as to why our matrices indeed are quite extreme, occupy very
different areas in the space of elections, and stay at close-to-fixed
distances from each other.
As explained in the previous section, we focus on
$m \times m$ frequency matrices (or, equivalently, on elections with $m$
candidates).

\paragraph{Identity and Uniformity.}
The first two matrices are the \emph{identity} matrix, $\ID_m$, with
ones on the diagonal and zeros elsewhere, and the \emph{uniformity}
matrix, $\UN_m$, with each entry equal to $\nicefrac{1}{m}$:
\begin{align*}
 \ID_m = \begin{bmatrix}
   1 & 0 & \cdots & 0\\
   0 & 1 & \cdots & 0\\
   \vdots & \vdots & \ddots & \vdots \\
   0 & 0 & \cdots & 1
 \end{bmatrix}, &&
 \UN_m = \begin{bmatrix}
   \nicefrac{1}{m} & \nicefrac{1}{m} & \cdots & \nicefrac{1}{m}\\
   \nicefrac{1}{m} & \nicefrac{1}{m} & \cdots & \nicefrac{1}{m}\\
   \vdots & \vdots & \ddots & \vdots \\
   \nicefrac{1}{m} & \nicefrac{1}{m} & \cdots & \nicefrac{1}{m}
 \end{bmatrix}.                
\end{align*}
The identity matrix corresponds to elections where each voter has the
same preference order, i.e., they capture perfect agreement
among the voters. In contrast, the uniformity matrix captures
elections where each candidate is ranked on each position equally
often, i.e., where, in aggregate, all the candidates are viewed as
equally good. Hence, uniformity elections can be seen as capturing
perfect lack of agreement regarding the relative qualities of the
candidates.

Uniformity elections are quite similar to the IC ones
and, in the limit, indistinguishable from them.  Indeed, if we choose
each preference order uniformly at random then, in expectation, each
candidate would be ranked on each position the same number of times.
Yet, for a fixed number of voters, typically IC elections are at some
(small) positionwise distance from uniformity.

\paragraph{Stratification.}
The next matrix, \emph{stratification}, is defined as
follows (we assume that $m$ is even):
\[
  \ST_m = \begin{bmatrix}
    \UN_{\nicefrac{m}{2}} & 0 \\
    0 & \UN_{\nicefrac{m}{2}}
  \end{bmatrix}.
\]
Stratification matrices correspond to elections where the voters agree
that half of the candidates are 
more desirable than the other half, but, in aggregate, are unable to
distinguish between the qualities of the candidates in each group.

 \paragraph{Antagonism.}
 For the next matrix, we need one more piece of notation.  Let
 $\rID_m$ be the matrix obtained by reversing the order of the columns
 of the identity matrix $\ID_m$. We define the \emph{antagonism}
 matrix, $\AN_m$, to be
 $ \textstyle \nicefrac{1}{2} \ID_m+\nicefrac{1}{2} \rID_m:$
\[
 \textstyle
 \AN_m = \frac{1}{2}\begin{bmatrix}
   1 & 0 & \cdots & 0 & 0\\
   0 & 1 & \cdots & 0 & 0\\
   \vdots & \vdots & \ddots & \vdots \\
   0 & 0 & \cdots & 1 & 0 \\
   0 & 0 & \cdots & 0 & 1
 \end{bmatrix}
 +
 \frac{1}{2}\begin{bmatrix}
   0 & 0 & \cdots & 0 & 1\\
   0 & 0 & \cdots & 1 & 0 & \\
   \vdots & \vdots & \iddots & \vdots \\
   0 & 1 & \cdots & 0 & 0\\
   1 & 0 & \cdots & 0 & 0
 \end{bmatrix}
 =
 \begin{bmatrix}
   \nicefrac{1}{2} & 0 & \cdots & 0 & \nicefrac{1}{2}\\
   0 &\nicefrac{1}{2} &  \cdots &  \nicefrac{1}{2} & 0\\
   \vdots & \vdots & \iddots & \vdots \\
   0 & \nicefrac{1}{2} & \cdots & \nicefrac{1}{2} & 0\\
   \nicefrac{1}{2} & 0 & \cdots & 0 & \nicefrac{1}{2}
 \end{bmatrix}.
\]
Such matrices are generated, e.g., by
elections where half of the voters rank the candidates in one order,
and half of the voters rank them in the opposite one, so there is a
clear conflict (however, there are also many other elections that
generate this matrix). 
In some sense, stratification and antagonism are based on similar
premises. Under stratification, the set of candidates is partitioned
into halves with different properties, whereas in antagonism the
voters are partitioned into such halves. However, the nature of the partitioning is
quite different.

\paragraph{Distances Between the Matrices.}
We chose the above matrices
because they capture natural, intuitive phenomena, 
seem to occupy very different areas of the space of elections, and
their relative distances do not depend strongly on the numbers of
candidates.  To see that the latter two points hold, let us calculate
their positionwise distances (for the calculations, see
\Cref{se:recov_app}).

\begin{restatable}[\appsymb]{proposition}{dist}\label{pr:calc}
    If $m$ is divisible by $4$, then it holds that:
  \begin{enumerate}
  \item $\POS(\ID_m,\UN_m) = \frac{1}{3}(m^2-1)$,
  \item $\POS(\ID_m,\AN_m) = \POS(\UN_m,\ST_m) = \frac{m^2}{4}$,
  \item
    $\POS(\ID_m,\ST_m) = \POS(\UN_m,\AN_m)  = 
    \frac{2}{3}(\frac{m^2}{4}-1)$,
  \item $\POS(\AN_m,\ST_m) = \frac{13}{48} m^2 - \frac{1}{3}$.
  \end{enumerate}
\end{restatable}

\noindent
To normalize these distances, we divide them by
$D(m) = \POS(\ID_m, \UN_m)$, which is the largest positionwise
distance between two matrices from $\calF(m)$, as shown by
\citet{boe-fal-nie-szu-was:c:metrics}; recall
\Cref{rem:norm-pos}.\footnote{The distance between $\ID_m$ and $\UN_m$
  is also the largest possible under the isomorphic swap distance.
  Interestingly, the distance between $\ID_m$ and $\AN_m$ is also the
  largest possible in this setting, and the distance between $\AN_m$
  and $\UN_m$ seems to approach this largest possible value as the
  number of candidates
  grows~\citep{boe-fal-nie-szu-was:c:metrics}. This is another reason
  for including $\AN_m$ in our compass.}  For each two matrices $X$
and $Y$ among our four, we let
$d(X,Y) := \lim_{m \rightarrow
  \infty}\nicefrac{\POS(X_{4m},Y_{4m})}{D(4m)}$. A simple computation
shows the following (see also the drawing on the right; we sometimes
omit the subscript $m$ for simplicity):

\begin{minipage}[b]{0.45\textwidth}
  \centering
  \begin{align*}
  &d(\ID,\UN) = 1,\\
  &d(\ID,\AN) = d(\UN, \ST) = \nicefrac{3}{4},\\
  &d(\ID,\ST) = d(\UN,\AN) = \nicefrac{1}{2},\\
  &d(\AN,\ST) = \nicefrac{13}{16}.\\
\end{align*}
\end{minipage}
\begin{minipage}[b]{0.5\textwidth}
  \centering
        \newcommand{\drawun}[2]{
    \draw (#1+0.5, #2+1) node[anchor=south] {UN};
    \fill[black!25!white] (#1+0,#2+0)  rectangle (#1+1,#2+1);
    \draw (#1+0,#2+0) rectangle (#1+1,#2+1);
  }

  \newcommand{\drawan}[2]{
    \draw (#1+1.75, #2+0.5) node[anchor=south] {AN};
    \fill[black!25!white] (#1+0,#2+0)  -- (#1+0.2, #2+0) -- (#1+1, #2+1-0.2) -- (#1+1, #2+1) -- (#1+1-0.2, #2+1) -- (#1, #2+0.2) -- cycle;
    \fill[black!25!white] (#1+0,#2+1)  -- (#1+0.2, #2+1) -- (#1+1, #2+0.2) -- (#1+1, #2) -- (#1+1-0.2, #2) -- (#1, #2+1-0.2) -- cycle;
    \draw (#1+0,#2+0) rectangle (#1+1,#2+1);
  }

  \newcommand{\drawid}[2]{
    \draw (#1+0.5, #2+1) node[anchor=south] {ID};
    \fill[black!25!white] (#1+0,#2+1)  -- (#1+0.2, #2+1) -- (#1+1, #2+0.2) -- (#1+1, #2) -- (#1+1-0.2, #2) -- (#1, #2+1-0.2) -- cycle;
    \draw (#1+0,#2+0) rectangle (#1+1,#2+1);
  }

  \newcommand{\drawst}[2]{
    \draw (#1-0.75, #2-0.5) node[anchor=south] {ST};
    \fill[black!25!white] (#1+0,#2+1)  rectangle (#1+0.5, #2+0.5);
    \fill[black!25!white] (#1+0.5,#2+0.5)  rectangle (#1+1, #2+0);
    \draw (#1+0,#2+0) rectangle (#1+1,#2+1);
  }

    \begin{tikzpicture}[xscale=0.5, yscale=0.5]
    \clip (-0.1, -3) rectangle (9, 3.5);
    \drawun{0}{0}
    \drawid{8}{0}
    \drawan{3}{2}
    \drawst{5}{-2}
    \draw (1,0.5) -- (8,0.5);
    \draw (3,0.5) node[anchor=south] {$1$};
    \draw (1,1) -- (3,2.5);
    \draw (2,1.75) node[anchor=south] {$\frac{1}{2}$};
    \draw (4,2.5) -- (8,1);
    \draw (6,1.75) node[anchor=south] {$\frac{3}{4}$};
    \draw (4,2) -- (5,-1);
    \draw (4.7,0.75) node[anchor=south] {$\frac{13}{16}$};
    \draw (1,0) -- (5,-1.5);
    \draw (2,-0.5) node[anchor=north] {$\frac{3}{4}$};
    \draw (6,-1.5) -- (8,0);
    \draw (7.2,-0.5) node[anchor=north] {$\frac{1}{2}$};
  \end{tikzpicture}
\end{minipage}

Importantly, already for fairly small numbers of candidates the
normalized distances between the compass matrices become very similar
to the above-computed limit values. For example, for $m=12$ they only
differ by a few percent. Hence, the compass matrices indeed can be
used to form the backbones of our maps.

\begin{remark}
  Using a combination of brute-force search and ILP solving, for very
  small values of $m$, such as $4$ or $5$, we found sets of
  matrices that are even further away from each other than the compass
  ones. Unfortunately, we did not see ways of generalizing these
  matrices to arbitrary numbers of candidates.
\end{remark}

\subsection{Paths Between Election Matrices}

Next, we consider convex combinations of frequency matrices.  We do so
to derive matrices that would be located between the compass ones in
the space of elections (as defined by the positionwise distance) and
would support our maps' backbones.

Given two frequency matrices,~$X$ and $Y$, and $\alpha \in [0,1]$, one
might expect that matrix $Z = \alpha X + (1-\alpha)Y$ would lie at
distance~$(1-\alpha) \POS(X,Y)$ from $X$ and at distance
$\alpha \POS(X,Y)$ from~$Y$, so that we would~have:
\begin{align*}
     \POS(X,Y) = \POS(X, Z) 
                     + \POS(Z, Y).
\end{align*}
However, without further assumptions this is not necessarily the
case. Indeed, if we take~$X = \ID_m$ and $Y = \textrm{rID}_m$, then
$\POS(X,Y) = 0$, but $Z = 0.5X+0.5Y = \AN_m$, so
$\POS(X,Z) = \POS(\ID,\AN) > 0$ and
$\POS(X,Y) \neq \POS(X,Z) + \POS(Z,Y)$.  Yet, if we arrange the
two 
matrices $X$ and $Y$ so that their positionwise distance is witnessed
by the identity permutation of their column vectors (i.e., by a
trivial candidate matching), then their convex combination lies
exactly between them. 
\begin{proposition}\label{pro:paths}
     Let $X = (x_1, \ldots, x_m)$ and $Y = (y_1, \ldots, y_m)$ be two
  $m \times m$ frequency matrices such that
  $
     \POS(X,Y) = \textstyle \sum_{i=1}^m \EMD(x_i,y_i).
  $
  Then, for each $\alpha \in [0,1]$ it holds that~$\POS(X,Y) = \POS(X, \alpha X + (1-\alpha)Y) + \POS( \alpha X +
  (1-\alpha)Y, Y)$.
\end{proposition}
\begin{proof}
  Let $Z = (z_1, \ldots, z_m) = \alpha X + (1-\alpha) Y$ be our convex
  combination of $X$ and $Y$.  We note two properties of the earth
  mover's distance. Let $a$, $b$, and $c$ by three vectors that
  consist of nonnegative numbers, where the entries in $b$ and $c$ sum up to the same
  value. Then, it holds that $\EMD(a+b,a+c) = \EMD(b,c)$. Further, for a
  nonnegative number~$\lambda$, we have that
  $\EMD(\lambda b, \lambda c) = \lambda\EMD(b,c)$.  Using these
  observations and the definition of the earth mover's distance, we
  note that:
  \begin{align*}
    \textstyle
    \POS(X,Z)  & \textstyle\leq \sum_{i=1}^m \EMD(x_i,z_i) 
     =   \textstyle\sum_{i=1}^m \EMD(x_i,\alpha x_i + (1-\alpha)y_i) \\
     &=  \textstyle\sum_{i=1}^m \EMD((1-\alpha)x_i, (1-\alpha)y_i) 
     =  \textstyle(1-\alpha) \sum_{i=1}^m \EMD(x_i,y_i) = (1-\alpha)\POS(X,Y).
  \end{align*}
  The last equality follows by our assumption regarding $X$ and
  $Y$. By an analogous reasoning we also have that
  $\POS(Z,Y) \leq \alpha \POS(X,Y)$. By putting these two inequalities
  together, we have that:
  \[
    \POS(X,Z) + \POS(Z,Y) \leq \POS(X,Y).
  \]
  By the triangle inequality, we have that $\POS(X,Y) \leq \POS(X,Z) + \POS(Z,Y)$
  and, so, we conclude that
  $\POS(X,Z) + \POS(Z,Y) = \POS(X,Y)$.
\end{proof}

Using  Proposition~\ref{pro:paths}, for each two compass matrices,
we can generate a sequence of matrices that form a path between them.
For example, matrix $0.5\ID + 0.5\UN$ is exactly at the same distance
from $\ID$ and from $\UN$.  Note that by the proof of
Proposition~\ref{pr:calc}, it holds that the positionwise distance
between any two of our four compass matrices is witnessed by the identity
mapping, as required by Proposition~\ref{pro:paths}.

\begin{figure}[t]
    \centering

        \includegraphics[width=6.0cm, trim={0.2cm 0.2cm 0.2cm 0.2cm}, clip]{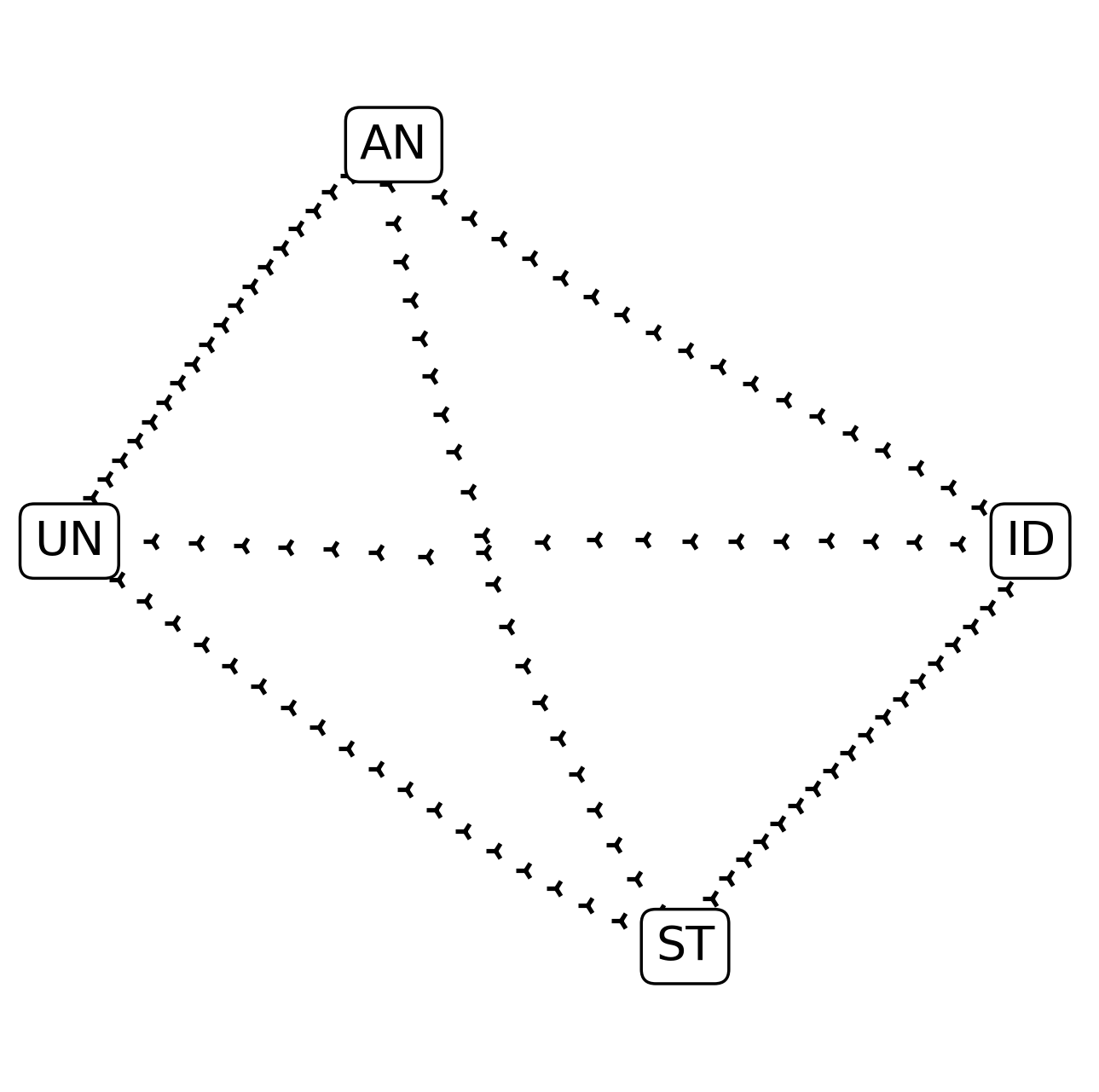}
    
        \caption{The backbone map, consisting of the four compass
          matrices and their connecting paths.  Each compass matrix is
          represented by a rectangle with a label. For each two
          compass matrices we have 20 path matrices, formed as their
          convex combinations.}
    \label{fig:paths}
\end{figure}

\subsection{The Backbone Map}\label{sec:basic-map}

Finally, we gather all the tools introduced in this section and
form our first map consisting of the four compass matrices and their convex combinations. While it
will be very simplistic, it will illustrate how such maps are made.

To form a map, we first need to select a dataset of elections (or,
frequency matrices). In our case, we take the four compass matrices
and for each two of them, we take $20$ path matrices formed as their
convex combinations (specifically, for each two compass matrices $X$
and $Y$ and each
$\alpha \in \{\nicefrac{1}{21}, \nicefrac{2}{21}, \ldots,
\nicefrac{20}{21}\}$, we have matrix $\alpha X + (1-\alpha)Y$).
Altogether, we have $124$ matrices. We generate the matrices for the
case of $10$ candidates, but for other numbers of candidates the map
would look similarly (except for very small candidate sets).  Since we
are working with frequency matrices only, the number of voters is
irrelevant. For each pair of our matrices, we compute their
positionwise distance (in the current setting, we could even derive
these distances analytically, but generally we invoke the algorithm
from \Cref{pro:pos-poly}).  Finally, given these distances, we invoke
an embedding algorithm (in this case we use Kamada-Kawai algorithm,
KK~\citep{kam-kaw:j:embedding,mt:sapala}). This algorithm associates
each election from the dataset with a point on a 2D plane in such a
way that the Euclidean distance between two points resembles the
positionwise distance between the corresponding elections as accurate
as possible.

We show the result in \Cref{fig:paths}. We can see that the picture is
very similar to the fully analytical diagram shown at the end of
\Cref{sec:four}. Indeed, the compass matrices are located at similar
distances as in the diagram. Further, the matrices from each path are
equally spaced (since the distance between $\UN$ and $\ID$ is larger
than between, say, $\UN$ and $\AN$, the matrices from the former path
are more spread, whereas those from the latter are located closer to
each other). This reinforces our hope that when there is a good
embedding, then KK finds it. In \Cref{sec:map-cultures} we will build a
full-fledged map of elections, with a much more diverse dataset, and
we will provide stronger arguments for its credibility.

\subsection{Why Is Positionwise Distance a Good Choice?}

Before we move on to the next section, where we will create our
full-fledged maps, let us take a pause to consider why positionwise
distance is a good choice for the map framework. To this end, we will
largely summarize the results of
\citet{boe-fal-nie-szu-was:c:metrics}, where the authors compared
several possible choices of distances.

Ideally, we would prefer to use the isomorphic swap or Spearman
distances instead of the positionwise one (recall
\Cref{sec:distances}, where these distances are briefly discussed, and
the work of \citet{fal-sko-sli-szu-tal:j:isomorphism}). The reasons
for this are as follows. First, these two distances distinguish as
many elections as possible while being invariant to renaming the
candidates and voters; in other words, under these two distances two
elections are at distance zero if and only if they are
isomorphic. Second, if we have two elections $E$ and $F$, then their
swap and Spearman distances change at most very mildly if we introduce
a small local change to one of them (such as swapping two candidates
in some vote). In other words, these distance measures are robust---a
small change in the input cannot lead to a large change in the
output. The third argument is that both distance measures treat all
positions in the votes identically: A swap of two top candidates in
some vote affects the distance in qualitatively the same way as
swapping two bottom candidates. While there certainly are applications
where it would be beneficial to pay attention to tops of the votes
more closely than to their bottoms, such an agnostic approach is
natural when developing a generic framework that can be used for many
different settings and for general comparison of elections. On the
negative side, the huge (practical) downside of the isomorphic swap
and Spearman distances is that they are NP-hard to compute. Indeed,
\citet{boe-fal-nie-szu-was:c:metrics} report that computing a map of
elections with $10$ candidates and $50$ voters based on these two
distances took several weeks on a powerful computing cluster.
Consequently, generating maps with $100$ candidates and $100$ voters
(or even more), as we do in the following sections, is completely
infeasible using these distances.

Let us now compare how the positionwise distance fares against the
isomorphic swap and Spearman ones. Foremost, there is a fast,
polynomial-time algorithm for computing the positionwise
distance. Hence, it is possible to use it for much larger elections
than the two isomorphic distances; computing the maps shown in this
paper takes at most a few hours on a modern desktop computer and not
weeks on a powerful computing cluster. Further, as the positionwise
distance operates on frequency matrices, it can directly work on
elections with different numbers of voters (although we never use this
feature in this paper). On the other hand, applying isomorphic swap
and Spearman distances to elections with different numbers of voters
would require additional work.  The positionwise distance also has a
similar robustness property as the swap and Spearman ones: Introducing
a small change to one of the elections in its input can only mildly
affect the output. Just like the two isomorphic distances, it is also
agnostic to which parts of the votes are changed; it treats all the
positions in the votes similarly.

The big difference between the positionwise distance and the swap and
Spearman ones lies in the equivalence classes they induce. Given an
election $E$, we say that its equivalence class under a distance $d$
consists of all elections that are at distance zero from $E$ under
$d$. Since the positionwise distance is defined over frequency
matrices, by definition it has fewer equivalence classes (hence, it is
in some sense less precise) than the swap and Spearman
distances. However, \citet{boe-fal-nie-szu-was:c:metrics} have shown
that among all distances that they tried (that are invariant to
renaming the candidates and voters), the positionwise distance still
has the largest number of equivalence classes, with a fairly large
advantage over the other distances, including one defined on top of
weighted majority relations of input elections. Hence, while the
positionwise distance loses precision as compared to the swap and
Spearman distances, the loss is the smallest among distances that were
considered for the map framework to date.\footnote{Note that
  \citet{boe-fal-nie-szu-was:c:metrics} only looked at equivalence
  classes of very small elections, with a few candidates and voters.}

Finally, we mention that \citet{boe-fal-nie-szu-was:c:metrics} also
computed Pearson correlation coefficient (PCC, see \Cref{sec:pcc} for
a formal introduction of this measure) between distances computed
using the isomorphic swap distance and several other ones. They found
that the positionwise distance had the highest correlation (e.g., on
a dataset of $340$ diverse elections, formed in a similar spirit as
the dataset that we describe in the next section, the PCC between the
positionwise distance and isomorphic swap distance was 0.745, which
means a strong correlation).

Overall, we conclude that the positionwise distance is a good choice
for the map framework, for regimes where neither the isomorphic
swap nor the isomorphic Spearman ones can be used.

\section{Maps of Statistical Cultures}\label{sec:map-cultures}

The main goal of this section is to form a map of elections generated
using various statistical cultures, to analyze its contents, and to
argue that it is accurate and, hence, credible.  We consider all the
statistical cultures from Section~\ref{sec:cultures}, but, naturally,
we use a limited set of parameters.  Specifically, in
\Cref{sec:datasets}, we describe the precise composition of our
diverse synthetic dataset and present its visualization as a map.
Subsequently, in \Cref{sec:general_observations}, we analyze where
elections generated by different statistical cultures land on the map
and which elections end up being placed close to each other and why.
Next, in \Cref{sec:robustness}, we verify whether the produced maps
are credible in the sense that distances between points on the map
accurately reflect the positionwise distances between the
corresponding elections.

Overall, this section presents the first use case of the map
(\emph{Finding Relations Between Elections}), as presented
in \Cref{sec:usecases}.

\begin{table}[]
  \centering
{
  \begin{tabular}{lcc}
    \toprule
    Model & Number of Elections \\
    
    \midrule
    Impartial Culture           & 20 \\
    \midrule
    Single-Peaked (Conitzer)  & 20 \\
    Single-Peaked (Walsh)      & 20 \\
    SPOC                       & 20 \\
    Single-Crossing           & 20 \\
    \midrule
    1D        & 20 \\
    2D         & 20 \\
    3D         & 20 \\
    5-Cube         & 20 \\
    10-Cube         & 20 \\
    20-Cube         & 20 \\
    \midrule
    Circle         & 20 \\
    Sphere         & 20 \\
    4-Sphere   & 20 \\
    \midrule
    Group-Separable (Balanced)     & 20 \\
    Group-Separable (Caterpillar)    & 20 \\
    \midrule
    Urn       & 80 \\  
    Mallows    & 80 \\
    \midrule
    Compass ($\ID$,~$\AN$,~$\UN$,~$\ST$) & 4 \\
    \midrule
    Paths & 20$\times$4 \\
    \bottomrule
  \end{tabular} }
    \caption{\label{tab:embed_setup}Composition of the datasets used in Section~\ref{sec:map-cultures}.}
\end{table}

\subsection{Choosing the Dataset and Drawing the Map}\label{sec:datasets}

All datasets that we use in this section consist of $480$ elections
from various statistical cultures---their exact composition is given in
\Cref{tab:embed_setup}---and we describe the parameters for the urn
and Mallows model a bit later. Additionally, we include the four
compass matrices and $80$ path matrices ($20$ matrices per path; we
omit the paths between $\UN$ and $\ID$ and between $\AN$ and $\ST$,
which would cross in the middle of the map and would make it more
cluttered).  All the elections that we consider include $100$ voters
and, typically, $100$ candidates, but occasionally we consider
candidate sets with different cardinalities.  We refer to our main
dataset as the $100 \times 100$ one, but we also consider, e.g.,
$4 \times 100$, $10 \times 100$ or $20 \times 100$ datasets, whose
elections have, respectively, $4$, $10$, or $20$ candidates and $100$
voters. All these datasets consist of elections coming from the same
cultures, as described in~\Cref{tab:embed_setup}. For each number of
candidates, we only generate a single dataset. Thus, for example, all
maps that depict the $100 \times 100$ dataset regard the same $480$
elections (plus the $84$ matrices for the compass and the
paths).\footnote{The only exception to this rule appears in
  \Cref{sec:robustness} where we generate several different datasets
  of each size (but with the same composition), to analyze the
  robustness of the maps that we obtain.  Nonetheless, even in that
  section the figures depict our main $100 \times 100$ dataset, that
  also appears in other sections.}

Most of the statistical cultures that we use either do not have
parameters, as in the case of, e.g., impartial culture or
single-crossing models, or have fairly simple ones, as in the case of,
e.g., the Euclidean models for which we only select the
dimension.\footnote{Naturally, one could argue that for the Euclidean
  model we should also view the distribution of the points as a
  parameter.  This certainly is a valid point of view and it would be
  interesting to consider what types of elections one
  obtains by varying the distributions of the candidate and voter
  points, but it goes beyond the scope of this paper. We focus on the
  simple distributions described in \Cref{sec:cultures}.}  Yet, the
choice of the parameters for the urn and Mallows models requires some care:
\begin{description}
\item[Urn model.]  Recall that the urn model has the parameter of
  contagion $\alpha$, which takes values between~$0$ and $\infty$. The
  larger is this parameter, the more similar are the generated
  elections to the identity one.  To generate an urn election, we
  choose $\alpha$ according to the Gamma distribution with shape
  parameter $k=0.8$ and scale parameter $\theta=1$. This ensures that
  about half of the urn elections are closer to $\UN$ than to $\ID$,
  and about half are closer to $\ID$ than to $\UN$ (we have verified
  this experimentally).

\item[Mallows model.] Regarding the Mallows model, we have the central
  vote $v^*$ and the dispersion parameter $\phi$, which takes values
  between $0$ and~$1$. Using $\phi = 0$ leads to generating $\ID$
  elections (where all the votes are equal to $v^*$) and using
  $\phi = 1$ leads to generating IC ones.  Hence, on the surface, it
  seems natural to choose $\phi$ uniformly at random from the interval
  $[0,1]$.  We do not take this approach and, instead, we use the
  normalization introduced by
  \citet{boe-bre-fal-nie-szu:c:compass}.\footnote{Note that the work
    of \citet{boe-bre-fal-nie-szu:c:compass} is one of the two
    conference papers on which the current work is based. We omit a
    detailed discussion of the proposed normalization of the Mallows model here because we consider
    it to be sufficiently important that it deserves its own paper,
    focused explicitly on its features and properties (which we intend
    to write later, and whose conference version already
    exists~\citep{boe-fal-kra:c:mallows-normalization}). Consequently,
    here we only refer to the initial definition of the normalization
    and we briefly explain why we find it useful.}  Specifically, they
  introduced a new parameter $\normphi$, which is converted to the original
  dispersion parameter as follows: For an election with $m$ candidates
  and a value of $\normphi \in [0,1]$, we compute $\phi$ such that the
  expected swap distance between a generated vote and the central
  one (normalized by ${m(m-1)}/{2}$, i.e., the largest possible swap
  distance between two votes) is equal to $\normphi/2$.
  Consequently, using $\normphi=0$ still leads to generating identity
  elections, using $\normphi=1$ still gives the Impartial Culture
  model.  However, now using $\normphi = x\in [0,1]$ leads to
  generating elections whose expected swap distance to $\ID$ is an~$x$
  fraction of the distance between $\ID$ and $\UN$, irrespective of
  the number of candidates (\citet{boe-bre-fal-nie-szu:c:compass}
  observed that the same holds if we use the positionwise distance
  instead of the isomorphic swap one).\footnote{The reader may worry
    that UN is a matrix rather than election. Indeed, for the swap
    distance, to compute the distance between $\ID$ and $\UN$, we
    would measure a distance between an election where all votes are
    identical (ID) and an election where each possible vote appears
    the same number of times (UN).}  The problem with the classic
  dispersion parameter is that if we fixed it as $\phi = x$,
  $0 < x < 1$, and generated elections with more and more candidates,
  their swap distance would decrease relative to the distance between
  $\ID$ and $\UN$~\citep{boe-bre-fal-nie-szu:c:compass}.
  Consequently, using the Mallows model with the classic $\phi$ parameter
  selected uniformly at random (for a large number of candidates, such
  as $100$) would lead to a large cluster of elections near $\ID$ and
  far fewer elections near $\UN$. Using values of the $\normphi$ parameter chosen
  uniformly at random generates Mallows elections that, roughly
  speaking, are spread uniformly between $\ID$ and $\UN$. Hence, we
  use the latter approach (see also the arguments for the normalized
  model provided by \citet{boe-fal-kra:c:mallows-normalization}).
  In the following, whenever we discuss the Mallows model, we use
  the $\normphi$ parameter.
\end{description}

\begin{figure}
\centering
  \includegraphics[width=7cm]{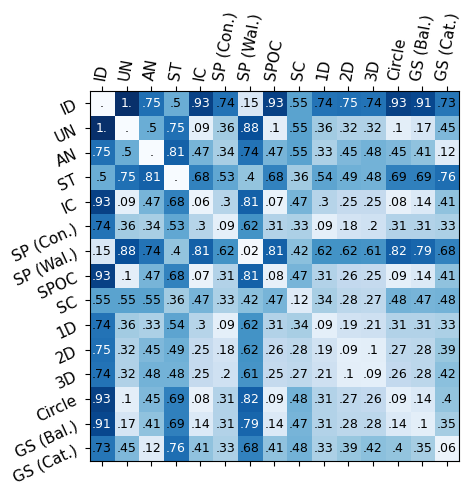}
  \caption{\label{fig:distances-100x100}Matrix of normalized positionwise distances between elections from the
    $100 \times 100$ dataset.  The number in each cell gives the average
    distance between elections from respective cultures (cells on the
    diagonal give average distances between elections from the given
    culture).}
\end{figure}

\begin{figure}[]
  \centering

  \includegraphics[width=10cm, trim={0.2cm 0.2cm 0.2cm
      0.2cm}, clip]{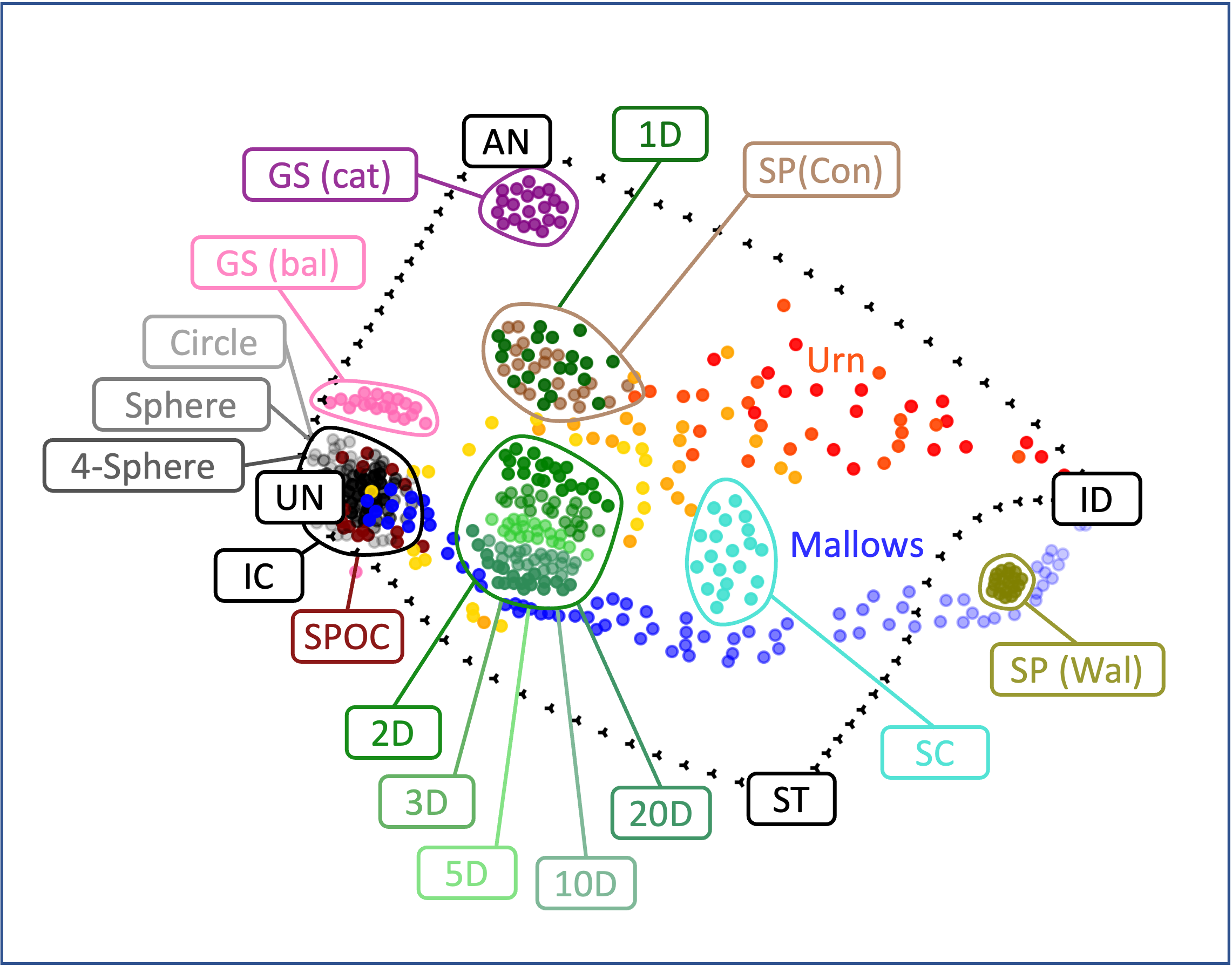}
      
  \caption{Map of elections obtained using the Kamada-Kawai algorithm (KK). The
    colors of the dots correspond to the statistical culture from
    which the elections are generated. Elections generated from the Mallows
    model use the blue color, and the more pale they are, the closer
    is the $\normphi$ parameter to $0$. The urn elections use 
    yellow-red colors, and the more red they are, the larger is the
    contagion parameter $\alpha$.}
    \label{fig:maps100x100}   
\end{figure}

We have computed the positionwise distance between each pair of
elections from the $100 \times 100$ dataset (see
\Cref{fig:distances-100x100}). Next, we have prepared the map of
this dataset using the same methodology as in \Cref{sec:basic-map}. In
particular, this means that we used the Kamada-Kawai embedding
algorithm~\citep{kam-kaw:j:embedding}.  In \Cref{app:embeddings} we
compare this map with two other ones, formed using the
Fruchterman-Reingold (FR) embedding
algorithm~\citep{fru-rei:j:graph-drawing} and the MDS
algorithm~\citep{kru:j:mds,lee:j:mds}.

\subsection{General Observations}\label{sec:general_observations}

In this section, we describe and analyze the map from
\Cref{fig:maps100x100}. We start with some more general observations
before we discuss different areas of the map in more detail.  One of
the most glaring observations is that elections coming from the same
statistical culture are typically clustered together and do not have
elections from other cultures in between (two main exceptions include
elections that are located close to the $\UN$ matrix, which are quite
mixed, and the 1D-Interval and single-peaked elections generated
according to the Conitzer model; see our discussion below).  Indeed,
this clustering seems appropriate given the (underlying) positionwise
distances between the elections, as shown in
\Cref{fig:distances-100x100}: Elections from the same statistical
cultures are, on average, closer to each other than to the other
elections.
Another general observation is that the compass matrices and the paths
between them (mostly) encircle the other elections (with the exception
of the area close to the Walsh single-peaked elections). This suggests
that, indeed, the compass matrices are quite extreme in the space of
elections induced by the positionwise distance (or that our set of
statistical cultures is missing some distinguished type of elections).

Throughout the rest of this section we discuss specific statistical
cultures and areas of the map. We stress that the maps themselves are
visualization tools and only provide intuitions. Accordingly, we verified  all the comments
and observations that we made below against the
actual (underlying) positionwise distances.

\subsubsection{Impartial Culture and Elections Around UN}
The area near the $\UN$ matrix contains a number of very different
elections, ranging from the completely unstructured ones, generated
using the impartial culture model, to highly structured SPOC ones (and
closely related Sphere ones). The reason why these very different
elections take this position is an inherent flaw in the
positionwise distance and frequency matrices. Indeed, both under IC
and under SPOC, each candidate appears at each position in a
generated vote with the same probability. The difference between IC and
SPOC is that under IC there are almost no correlations between the
positions of two given candidates in a single vote (except that they cannot
be equal), whereas under SPOC the correlations can be quite strong
(e.g., if two candidates are next to each other on the societal axis,
then they are likely located nearby in the vote). Yet, the
positionwise distance is blind to such correlations as they are not
captured in election's frequency matrices. This is the price that we need to pay for
using a fast-to-compute distance (using, e.g., the isomorphic swap
distance would avoid this issue, but would prevent us from considering
large candidate sets).
We point to the work of
\citet{boe-bre-elk-fal-szu:c:frequency-matrices} for a detailed
discussion of expected frequency matrices of various statistical
cultures, and to the work of \citet{boe-fal-nie-szu-was:c:metrics} for
a comparison of several distances that one could use in the maps
(including the positionwise one and the isomorphic swap one).

Another interesting issue is that balanced group-separable elections
are near the $\UN$ matrix, but still at some distance from it. If
the number of candidates were a power of two, then each candidate in a
vote generated for a balanced group-separable election would be
located at each possible position with the same probability and group-separable elections would be mixed
together with the IC and SPOC ones. However, since we consider $100$
candidates, the symmetry is broken to some extent and balanced
group-separable elections are a bit further away.

Finally, we note that elections generated from the impartial culture
model cover a relatively small area of the map. This confirms the
well-accepted intuition that limiting one's experiments to this model
is likely to produce biased results, as we would ignore a huge number
of possible elections that are not similar to the IC ones.

\subsubsection{Urn and Mallows Elections}

The elections generated using the urn and Mallows models form
``paths'' that link the $\UN$ and $\ID$ matrices. Importantly, these
paths are, roughly, equally dense. For example, if we consider the urn
``path'' then irrespective of whether we are looking at the area near the
$\UN$ matrix, or the area near the $\ID$ matrix, or anywhere in
between, there is a similar number of elections in the vicinity. The
same applies to the Mallows elections. This confirms that our way of
generating the $\alpha$ and $\normphi$ parameters is reasonable (in
\Cref{sec:scalability} we will see that this property is maintained for
different numbers of candidates).

\subsubsection{Single-Peaked and 1D-Interval Elections}
Single-peaked elections generated using the Conitzer model are nearly
indistinguishable from those generated using the 1D-Interval model (as
far as their positionwise distances go). 
This is because both models generate elections in a somewhat similar manner.
Indeed, generating 1D-Interval elections, if we draw the ideal
points of the candidates and voters uniformly at random from an
interval, then the process of forming the preference orders of the
voters is similar to the one used in the Conitzer model.
The difference is that 
in the 1D-Interval model there are more correlations: A given voter
$v$ ranks the candidate with the closest ideal point on top, the next
ranked candidate is either to the left or to the right of the first
one, and this happens with probability close to $\nicefrac{1}{2}$, and
so on (as in the Conitzer model). The correlations occur because the decisions regarding which
candidate should be ranked next are not made independently for each
voter, but are derived from the positions of the ideal points; still,
apparently, these correlations are sufficiently small not to be easily detectable
using the positionwise distances.

Elections generated using the Walsh model are very different from
those generated using the Conitzer one. To get a feeling as to why
this is a natural result, let us consider a candidate set
$C = \{l_m, \ldots, l_1, c, r_1, \ldots, r_m\}$ and the corresponding
societal axis
$l_m \lhd \cdots \lhd l_1 \lhd c \lhd r_1 \lhd \cdots \lhd r_m$. Under
the Conitzer model, the probability of generating a vote with a given
candidate on top is $\frac{1}{2m+1}$ (by definition of the model), but
under the Walsh's model these probabilities differ drastically.  The
probability of a vote with $l_m$ on top (or, with $r_m$ on top) is
$\frac{1}{2^{2m}}$ (because there are $2^{2m}$ different single-peaked
votes for this axis, each of them is drawn uniformly at random, and
only one of them starts with $l_m$), whereas the probability of
generating a vote with $c$ on top is $\Theta(\frac{1}{\sqrt{m}})$ (we
omit the calculations).  Generally, Walsh's model generates votes that
are similar to
$c \pref \{l_1,r_1\} \pref \{l_2, r_2\} \pref \cdots \pref \{l_m,
r_m\}$; i.e., they are close to having a center order, as in the
Mallows model. This explains both the distance of the Walsh elections
from the Conitzer ones, and their proximity to some Mallows elections.
\citet{boe-bre-elk-fal-szu:c:frequency-matrices} give a detailed
analysis of the frequency matrices of Conitzer, Walsh, and Mallows
distributions, which supports this observation. The differences between the Conitzer and Walsh
elections are also very clearly visible in ``microscope'' maps
provided by \citet{fal-kac-sor-szu-was:c:microscope}, which depict
inner structures of single elections.

\subsubsection{Euclidean Elections} 
We observe that, for a given $x$, the $x$-dimensional hypercube
elections are quite similar to each other and also fairly similar to
hypercube elections generated for other dimensions. One exception is
that the 1D-Interval elections are more different from the other
hypercube ones. This is understandable since 1D-Interval elections are
single-peaked and single-crossing, whereas the other hypercube
elections are not. Interestingly, the computed positionwise distances
were sufficient to recognize this difference.  Generally, hypercube
elections seem to form a large and diverse class, which means that
they may be useful in experiments. Indeed, in \Cref{sec:real-life} we
observe that many real-life elections land in the same area of the map
as high-dimensional hypercube ones.

Similar observations as for the hypercube elections apply to
hypersphere ones, except that---as discussed before---they are quite similar to the impartial culture
elections under the
positionwise distance. This shows that the distribution of the ideal points in the
Euclidean models has a strong impact on the generated
elections. Further studies are needed to recommend distributions that
should be used in experiments (as some may lead to particularly
appealing classes of elections, or, perhaps, to elections that are
close to those appearing in reality).
The similarity between SPOC  and hypersphere elections is
reassuring. Indeed, Circle elections are, by their nature, a subset
of SPOC ones, and for higher dimensions we would not expect
big changes. 

\subsubsection{Single-Crossing Elections}
Elections generated using our single-crossing model are fairly close
to both the Mallows ones and the urn ones, but are at quite some distance
from the 1D-Interval elections, which also are
single-crossing. Indeed, the map certainly does not show all kinds of
similarities between elections, but only those captured by the
positionwise distance.

\subsection{Accuracy of the Maps}\label{sec:robustness}

Now that we have seen that the maps can provide  some
interesting insights regarding the nature of the depicted elections, let us
analyze how accurate is the map from
Figure~\ref{fig:maps100x100},
in the sense of capturing the positionwise distances between the
elections most precisely.  To this end, we will evaluate it using
three measures. First, in \Cref{sec:pcc} we will evaluate the
correlation between the positionwise distance between two elections
and the Euclidean distance between their points in the embeddings.
Second, in \Cref{sec:distortion} we will have a closer look at the
extent to which distances between elections are distorted.  Finally,
in \Cref{app:monotonicity} we will consider whether relations between
positionwise distances are preserved.

In this section, whenever we write \textit{original distance} we refer
to the positionwise distance between two elections, and whenever we
write \textit{embedded distance} we refer to the Euclidean distance
between the points on a given map (which correspond to these
elections). Whenever we write \emph{normalized embedded (original) distance}, we refer to
the embedded (original) distance divided by the embedded (original) distance between the $\ID$ and $\UN$
matrices (as it is the largest possible one).  Further, we introduce
the notion of an \emph{embedding summary}: By an
embedding summary~$Q=(\mathcal{E}, d_\mathcal{M}, d_\Euc)$ we refer to a
triple that consists of a set of elections~$\mathcal{E}$, original
distances~$d_\mathcal{M}$ between these elections according to
metric~$\mathcal{M}$, and Euclidean distances~$d_\Euc$ between these
elections after the embedding.

In the analysis below, we will be interested in two kinds of results.
Foremost, we will analyze the accuracy of the embedding shown in
\Cref{fig:maps100x100}.  However, to also evaluate if the results
are not due to random chance, for each
dataset size (i.e., $4 \times 100$, $10 \times 100$, $20 \times 100$,
and $100 \times 100$) we generate $10$ datasets (as described in
\Cref{sec:datasets}) and report averages and standard deviations for
them (given the low variance of the results we observed, we find using
10 datasets sufficient).

\newcommand{\numberbarPCC}[2]{\tikz{
    \fill[purple!17] (0,0) rectangle (#1*200mm-167mm,10pt);
    \node[inner sep=0pt, anchor=south west] at (0,0) {#1 $\pm$ #2};}
}

\begin{table}
  \centering
  \begin{tabular}{c|l}
    \toprule
                     & \multicolumn{1}{c}{average PCC values} \\
    dataset          &  Kamada-Kawai \\
    \midrule
    $4   \times 100$ & \numberbarPCC{0.9661}{0.0044} \\
    $10  \times 100$ & \numberbarPCC{0.9686}{0.0061} \\
    $20  \times 100$ & \numberbarPCC{0.9745}{0.0015} \\
    $100 \times 100$ & \numberbarPCC{0.9735}{0.0115} \\
    \bottomrule
  \end{tabular}
  \caption{\label{tab:pcc}Average values of the PCC between the
    original and embedded distances for collections of datasets of
    various sizes. After the $\pm$ signs we report the standard
    deviations.}
\end{table}

\subsubsection{Correlation Between the Positionwise and Euclidean Distances}\label{sec:pcc}

Given an embedding summary $Q=(\mathcal{E}, d_\mathcal{M}, d_\Euc)$
associated with some embedding, we would like to evaluate the
correlation between the original distances and the Euclidean ones. To
this end, we use the Pearson correlation coefficient (PCC), a
classical measure of linear correlation used in statistics.

For two vectors $x = (x_1, \ldots, x_t)$ and $y = (y_1, \ldots, y_t)$,
their PCC is defined as follows ($\overline{x}$ is the arithmetic
average of the values from $x$; $\overline{y}$ is defined
analogously):
\[
  \mathrm{PCC}(x,y) =
  \frac{(\sum_{i=1}^t(x_i-\overline{x})(y_i-\overline{y}))}{\sqrt{\sum_{i=1}^t(x_i-\overline{x})^2
      \sum_{i=1}^t(y_i-\overline{y})^2}}.
\]
PCC measures the level of linear correlation between two random
variables and takes values between $-1$ and $1$ (its absolute value
gives the level of correlation and the sign indicates positive or
negative correlation; in our case, the closer a value to $1$, the more
correlation there is).

For the embedding from Figure~\ref{fig:maps100x100} we get very high
PCC, equal to~$0.9805$. In \Cref{tab:pcc} we report average PCC values
and standard deviations for our datasets of various sizes (note that
for our map in Figure~\ref{fig:maps100x100} we chose a particularly
fortunate embedding, so its PCC value is notably better than the
average reported in \Cref{tab:pcc}).  All in all, we see that the PCC
values are very high for all datasets. In \Cref{app:pcc} we compare
the performance of Kamada-Kawai algorithm with that of FR and MDS and
show that it has strong advantage over the former and slight over the
latter.

\subsubsection{Distortion}\label{sec:distortion}

\begin{figure}[t]
    \centering
    
        \includegraphics[width=8.cm, trim={0.2cm 0.2cm 0.2cm 0.2cm}, clip]{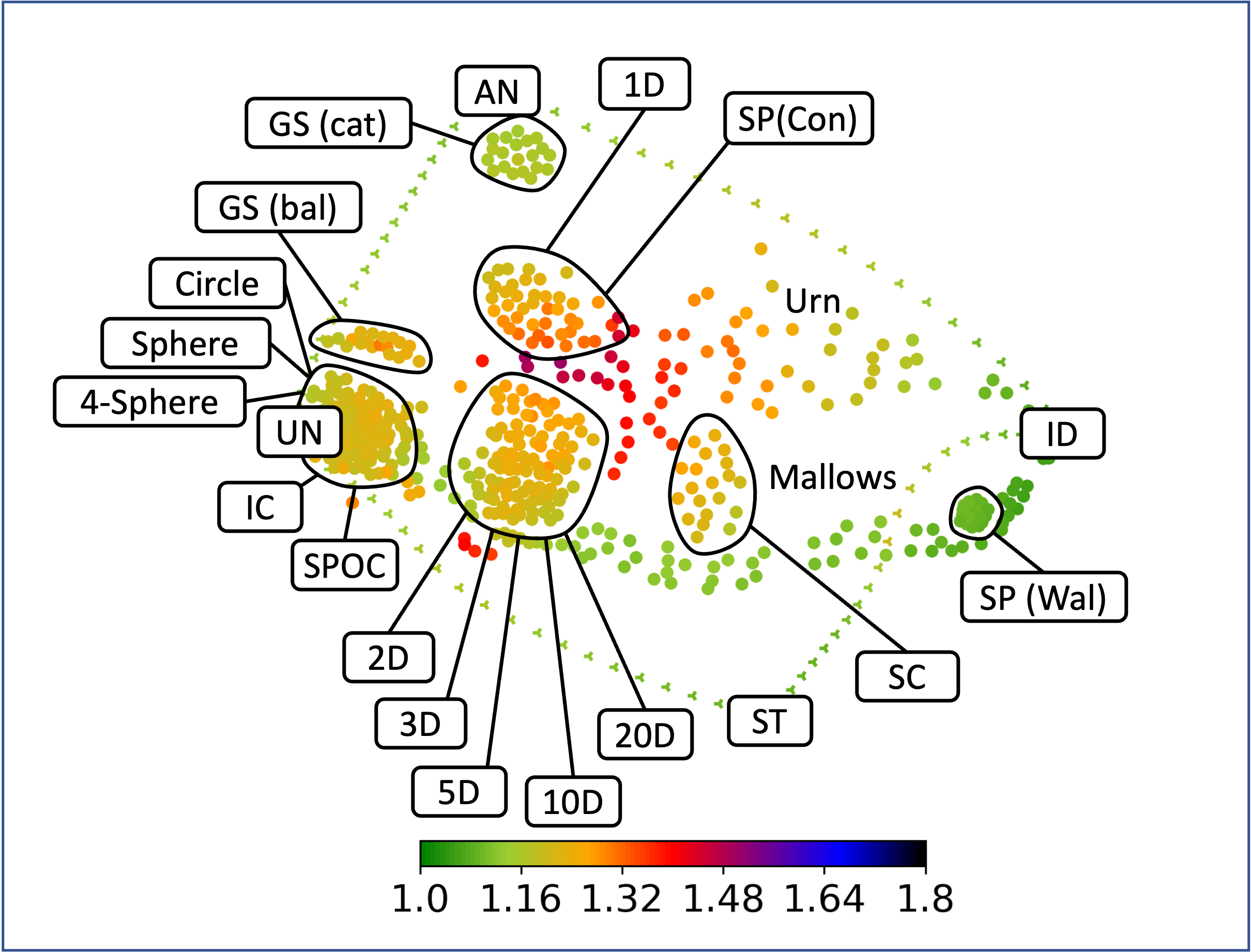}

        \caption{\label{fig:distortion} Distortion coloring for the
          emdedding from \Cref{fig:maps100x100}. Each election $X$
          (i.e., each point $X$) has a color reflecting its value
          $\AMR(X)$ for the respective embedding.}
\end{figure}

We analyze distortion of the distances introduced by the embedding.
The intuition is that the normalized embedded distance between two points
should be similar to the normalized original one. Formally, for a
given pair of elections~$X$ and~$Y$ the \emph{distortion} is defined
as:
\[
\MR(X,Y) = \frac{\max(\bar{d}_\Euc(X,Y), \bar{d}_\mathcal{M}(X,Y))}{\min(\bar{d}_\Euc(X,Y), \bar{d}_\mathcal{M}(X,Y))},  
\]
where~$\bar{d}(X,Y)$ means that the distance between~$X$ and~$Y$ is
normalized by the distance between~$\ID$ and~$\UN$.  Note that
$\MR(X,Y) \geq 1$ for all elections $X$ and $Y$. Its value can be
interpreted as the relative error between the original and the embedded
distance. For example, if $\MR(X,Y) = 1.2$ then it means that the larger
among the original and the embedded distance between $X$ and $Y$ is $20\%$
larger than the other one.
For a given embedding summary~$Q=(\mathcal{E}, d_\mathcal{M}, d_\Euc)$ and a given election~$X$, we define the
average
\emph{distortion} of this election in this embedding summary as:
\[
\AMR_{Q}(X) = \frac{1}{|\mathcal{E}|-1}\sum_{Y \in \mathcal{E} \setminus \{X\} } \MR(X,Y).
\]
The closer is the~$\mathrm{AMR}$ value to one, the better---this means
that the embedded distances are proportional to the original ones.

In \Cref{fig:distortion}, we present the map from
\Cref{fig:maps100x100} with elections colored according to their
distortion (i.e., according to their $\AMR$ vaules).  In
\Cref{app:distortion} we compare these distortion values to those
achieved by FR and MDS and find that our KK embedding is superior.
Wee see that urn elections, as well as 1D-Interval and Conitzer ones,
have the highest average distortion. Another way of phrasing this
observation would be that elections in the center of the map tend to
have the highest distortion.

\newcommand{\numberbarDistortion}[2]{\tikz{
    \fill[red!17] (0,0) rectangle (#1*60mm-60mm,10pt);
    \node[inner sep=0pt, anchor=south west] at (0,0) {#1 $\pm$ #2};}
}

\begin{table}
  \centering
  \begin{tabular}{c|l}
    \toprule
                     & \multicolumn{1}{c}{total average distortion values} \\
    dataset          &  Kamada-Kawai \\
    \midrule
    $4   \times 100$ & \numberbarDistortion{1.2612}{0.0158} \\
    $10  \times 100$ & \numberbarDistortion{1.2625}{0.0125} \\
    $20  \times 100$ & \numberbarDistortion{1.2406}{0.0060} \\
    $100 \times 100$ & \numberbarDistortion{1.2119}{0.0123} \\
    \bottomrule
  \end{tabular}
  \caption{\label{tab:distortion}Total average distortion values
    (i.e., $\AMR$ values) for collections of datasets of various
    sizes. After the $\pm$ signs we report the standard deviations.}
\end{table}

In \Cref{tab:distortion} we present the average values of the average
distortions for the KK embedding and datasets with different numbers
of candidates. For each dataset, we compute the average $\AMR$ value
for its elections, as well as the standard deviation.  We observe that
the more candidates there are, the lower is the distortion
value. Overall, the distances in the embedding are, on average, off by
$20-26\%$.

One reason why the distortion values may be high is that it is
particularly difficult to represent distances between elections within
a single statistical culture. For example, all elections in the
$100 \times 100$ dataset generated from the impartial culture model
are at roughly the same distance from each other, but it is impossible
to reflect this accurately in an embedding. Hence, in
\Cref{tab:distortion-omitting-own-culture} we report average total
distortion values for elections from each of our statistical cultures,
but computed only against elections that are sufficiently far away from them.
Specifically, for each election $X$ from a given statistical culture we consider only
those other election $Y$ from the dataset for which:
\[
  \POS(X,Y) \geq 0.1\cdot \POS(\UN,\ID).
\]
This way we disregard embedding errors that occur close to a given
election and focus on the broader picture.  We see that the
distortion values in \Cref{tab:distortion-omitting-own-culture} are
much lower than those reported in \Cref{tab:distortion}. For example,
the overall result for the KK embedding is $1.159$ and for some
cultures (such as, e.g., the Mallows model or the high-dimensional
Euclidean models) the results are notably lower. We interpret this as
saying that, all in all, our maps are sufficiently accurate, but
certainly not perfect.

\newcommand{\numberbarRange}[1]{\tikz{
    \fill[orange!17] (0,0) rectangle (#1*60mm-60mm,10pt);
    \node[inner sep=0pt, anchor=south west] at (0,0) {#1};}
}

\begin{table}
  \centering
  \begin{tabular}{c|lll}
    \toprule
                     & \multicolumn{1}{c}{average distortion values} \\
    Model          &  Kamada-Kawai \\
    \midrule
    Impartial Culture         &  \numberbarRange{1.07} \\
    \midrule
    Single-Peaked (Conitzer)  &\numberbarRange{1.244} \\
    Single-Peaked (Walsh)     &\numberbarRange{1.071} \\
    SPOC                      &\numberbarRange{1.081} \\
    Single-Crossing           &\numberbarRange{1.225} \\
    \midrule
    1D         &\numberbarRange{1.233} \\
    2D         &\numberbarRange{1.203} \\
    3D         &\numberbarRange{1.146} \\
    5-Cube          &\numberbarRange{1.114} \\
    10-Cube         &\numberbarRange{1.094} \\
    20-Cube         &\numberbarRange{1.097} \\
    \midrule
    Circle          &\numberbarRange{1.101} \\
    Sphere          &\numberbarRange{1.077} \\
    4-Sphere        &\numberbarRange{1.072} \\
    \midrule
    Group-Separable (Balanced)     &\numberbarRange{1.204} \\
    Group-Separable (Caterpillar)  &\numberbarRange{1.14} \\
    \midrule
    Urn        &\numberbarRange{1.285} \\  
    Mallows    &\numberbarRange{1.094} \\
    \midrule
    All &\numberbarRange{1.159} \\
    \bottomrule
  \end{tabular}
  \caption{\label{tab:distortion-omitting-own-culture} Total average
    distortion values for elections from given cultures in the
    $100 \times 100$ dataset, where for each election $X$ we compute
    the average distortion only with respect to elections whose
    distance from $X$ is at least $0.1 \cdot \POS(\ID,\UN)$.}

\end{table}

\subsubsection{Varying the Number of
  Candidates}\label{sec:scalability}

Next, we check the robustness of the map with respect
to varying the number of candidates in the depicted elections.  In
\Cref{fig:scale} we compare the maps of the $4 \times 100$,
$10 \times 100$, $20 \times 100$, and $100 \times 100$ datasets,
created using the KK algorithm.  As we can see, the maps
for~$10$,~$20$, and~$100$ candidates are surprisingly similar. The
largest difference we can observe is for the balanced group-separable
elections---for the case of~$100$ candidates they are clearly
separated, while for~$10$ and~$20$ candidates they are mingling with
the Circle elections (one possible explanation for this could be that,
relatively speaking, for $100$ candidates the underlying tree is
closer to being balanced than for $10$ and $20$ candidates). In
addition, Euclidean elections of different dimensions are better
separated from each other for~$100$ candidates than for~$10$ or~$20$
candidates.  Another interesting observation is that the Walsh
elections are shifting toward the right side of the map as we increase
the number of candidates. At the same time, the caterpillar
group-separable elections are shifting toward~$\AN$ (these two effects
are symmetric as the frequency matrices of Walsh elections are
transposed frequency matrices of the caterpillar group-separable ones;
for details, see the work of
\citet{boe-bre-elk-fal-szu:c:frequency-matrices}).

Only the map with four candidates is different, even if the main shape is
still maintained. Note that for~$4$ candidates there are only~$24$
possible different votes, which likely explains why the map is not as
meaningful. For example, \citet{fal-sor-szu:j:subelection-isomorphism}
have shown that if we generate two elections with 4 candidates and 50
voters using two, possibly different, statistical cultures (including
many of those that we consider), it often suffices to delete between
20\% and 50\% of the voters (from each of the elections) to make them
isomorphic. For the case of 10 candidates, one would typically have to
delete between 85\% and 97\% of the voters to achieve isomorphism. Hence it
is not surprising that our map for~4 candidates is not as similar to
the other ones (even if we use 100 voters and not 50).

\begin{figure}[t]
    \centering
    
    \begin{subfigure}[b]{0.49\textwidth}
        \centering
        \includegraphics[width=6.cm, trim={0.2cm 0.2cm 0.2cm 0.2cm}, clip]
        {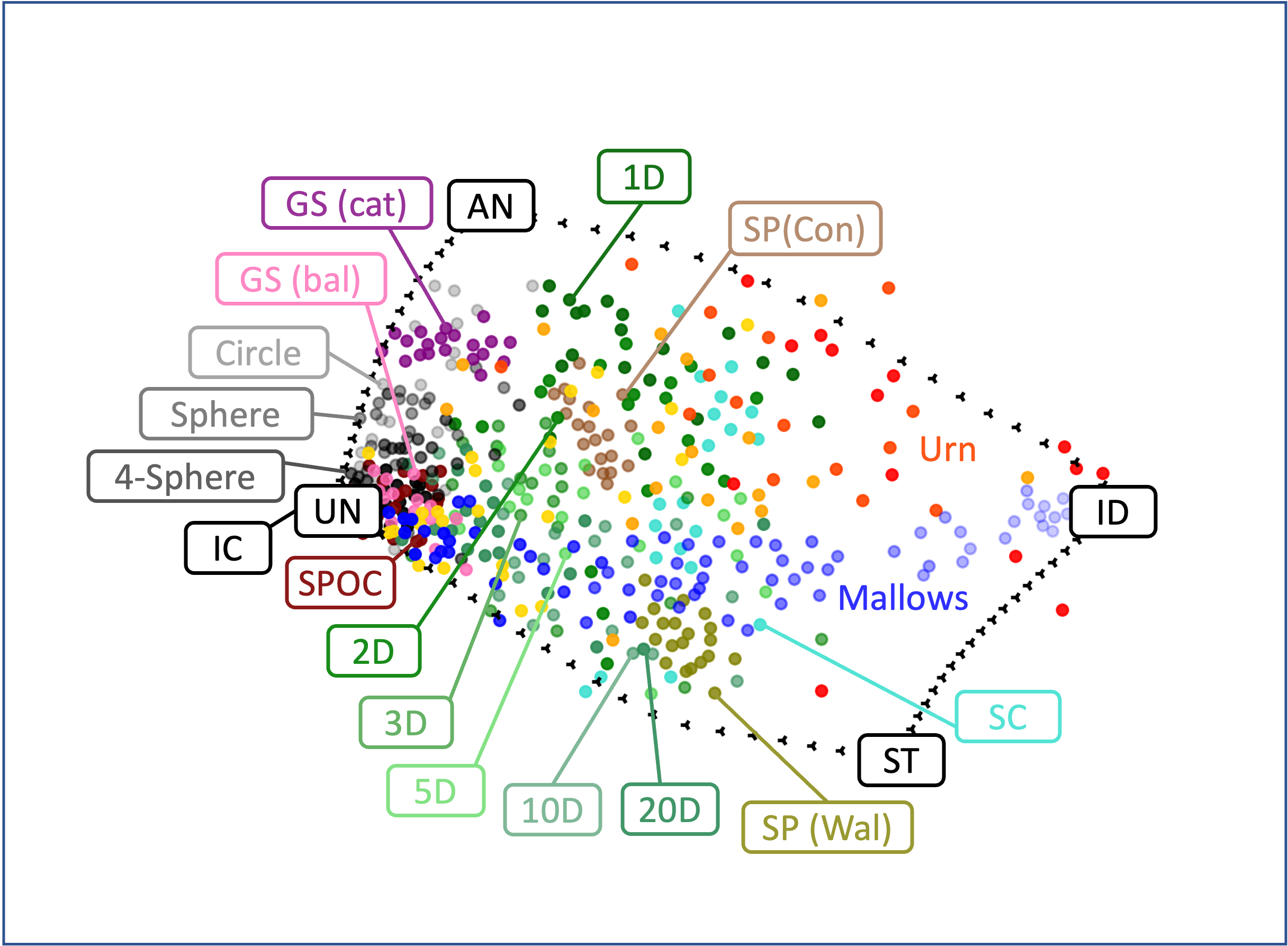}
        \caption{4 candidates and 100 voters}
    \end{subfigure}
    \begin{subfigure}[b]{0.49\textwidth}
        \centering
        \includegraphics[width=6.cm, trim={0.2cm 0.2cm 0.2cm 0.2cm}, clip]
        {images/scale/thesis_10x100.png}
        \caption{10 candidates and 100 voters}
    \end{subfigure}

    \vspace{1em}

    \begin{subfigure}[b]{0.49\textwidth}
        \centering
        \includegraphics[width=6.cm, trim={0.2cm 0.2cm 0.2cm 0.2cm}, clip]
        {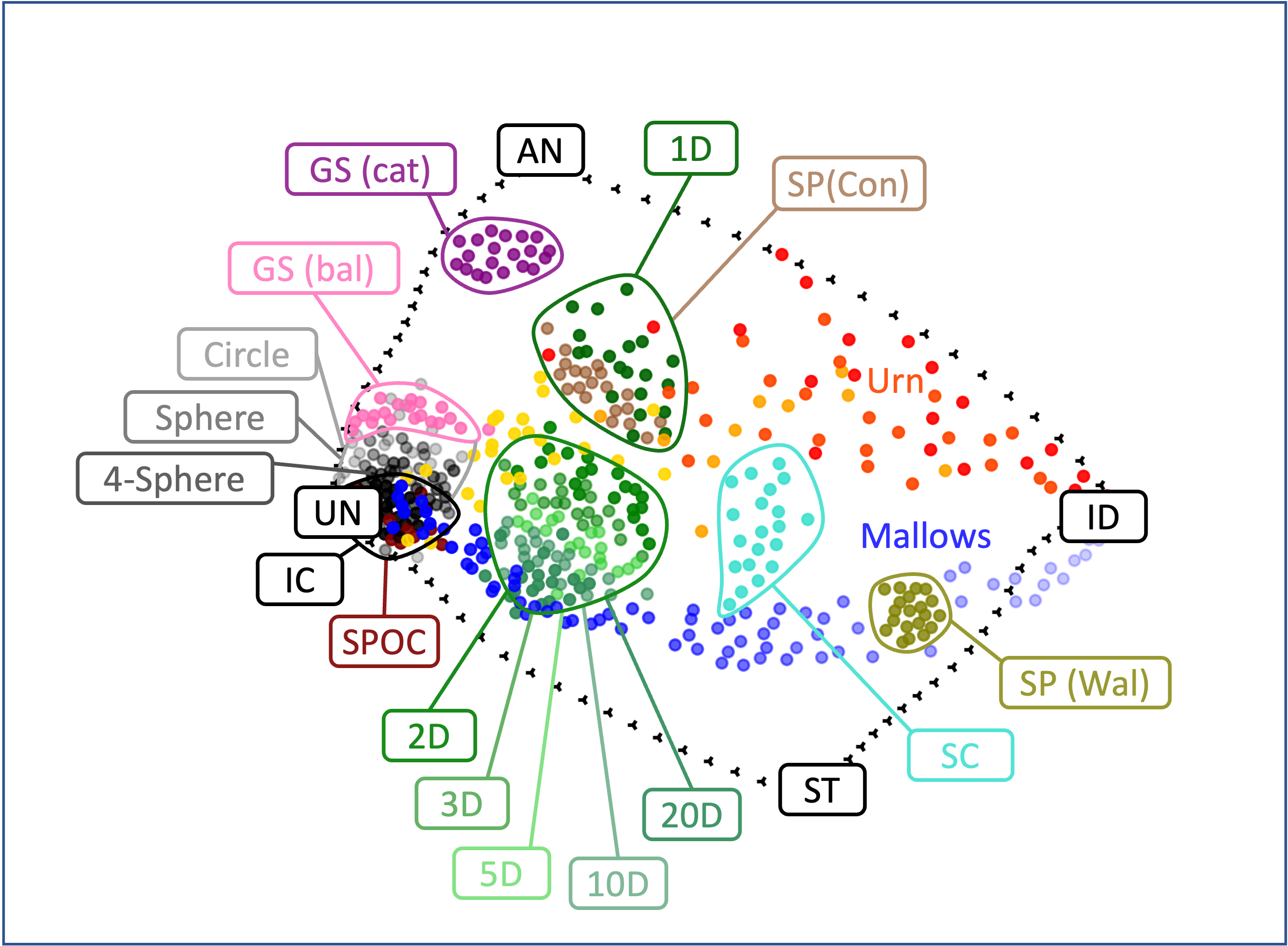}
        \caption{20 candidates and 100 voters}
    \end{subfigure}
    \begin{subfigure}[b]{0.49\textwidth}
        \centering
        \includegraphics[width=6.cm, trim={0.2cm 0.2cm 0.2cm 0.2cm}, clip]
        {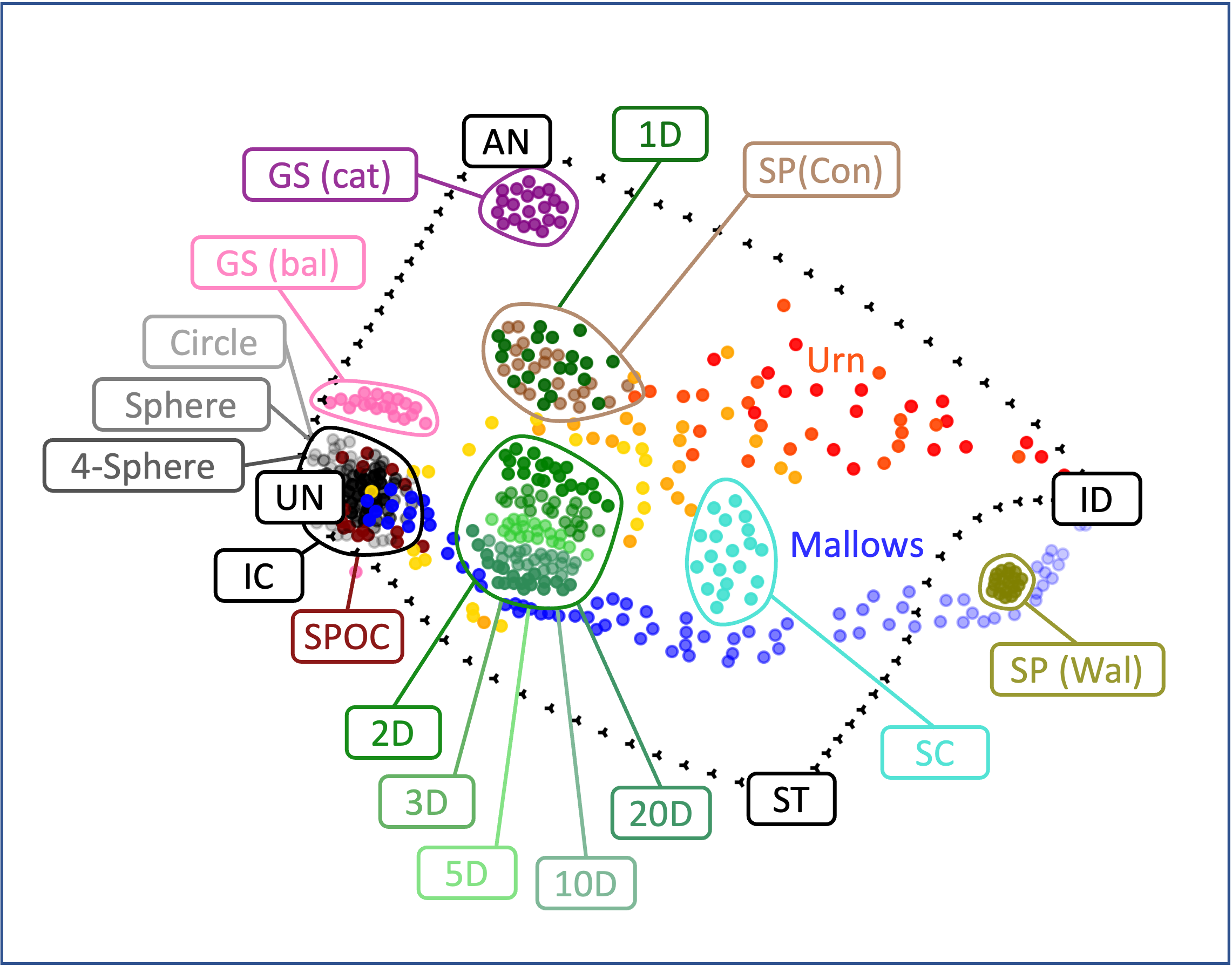}
        \caption{100 candidates and 100 voters}
    \end{subfigure}
    
    \caption{Maps of elections for datasets with different numbers of candidates.}
    \label{fig:scale}
\end{figure}

\begin{figure}[t]
    \centering
    
    \begin{subfigure}[b]{0.31\textwidth}
        \centering
        \includegraphics[width=5.cm, trim={0.2cm 0.2cm 0.2cm 0.2cm}, clip]
        {images/scale/thesis_100x100.png}
        \caption{100 candidates and 100 voters}
    \end{subfigure}
    \begin{subfigure}[b]{0.31\textwidth}
        \centering
        \includegraphics[width=5.cm, trim={0.2cm 0.2cm 0.2cm 0.2cm}, clip]
        {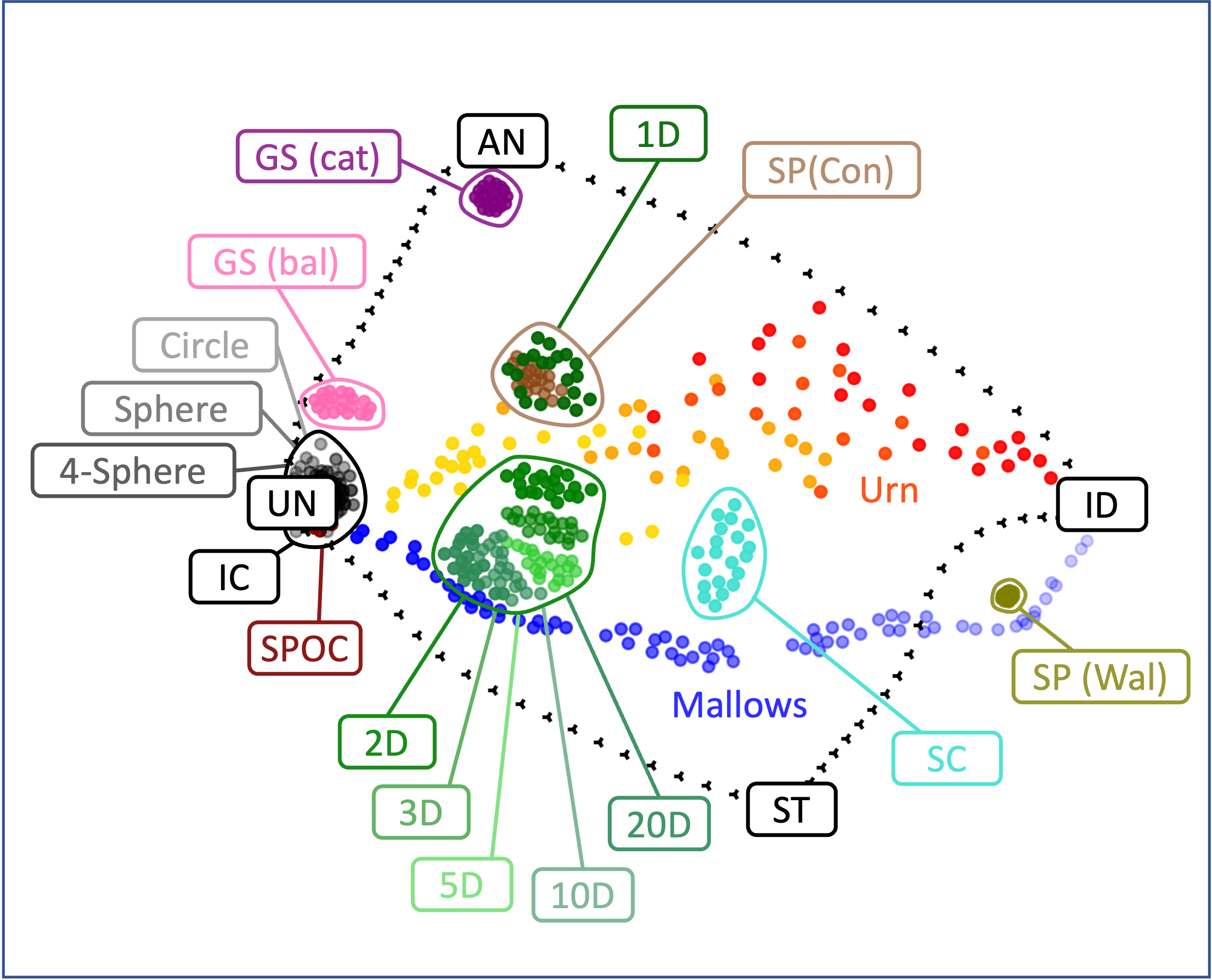}
        \caption{100 candidates and 1000 voters}
    \end{subfigure}
    \begin{subfigure}[b]{0.31\textwidth}
        \centering
        \includegraphics[width=5.cm, trim={0.2cm 0.2cm 0.2cm 0.2cm}, clip]
        {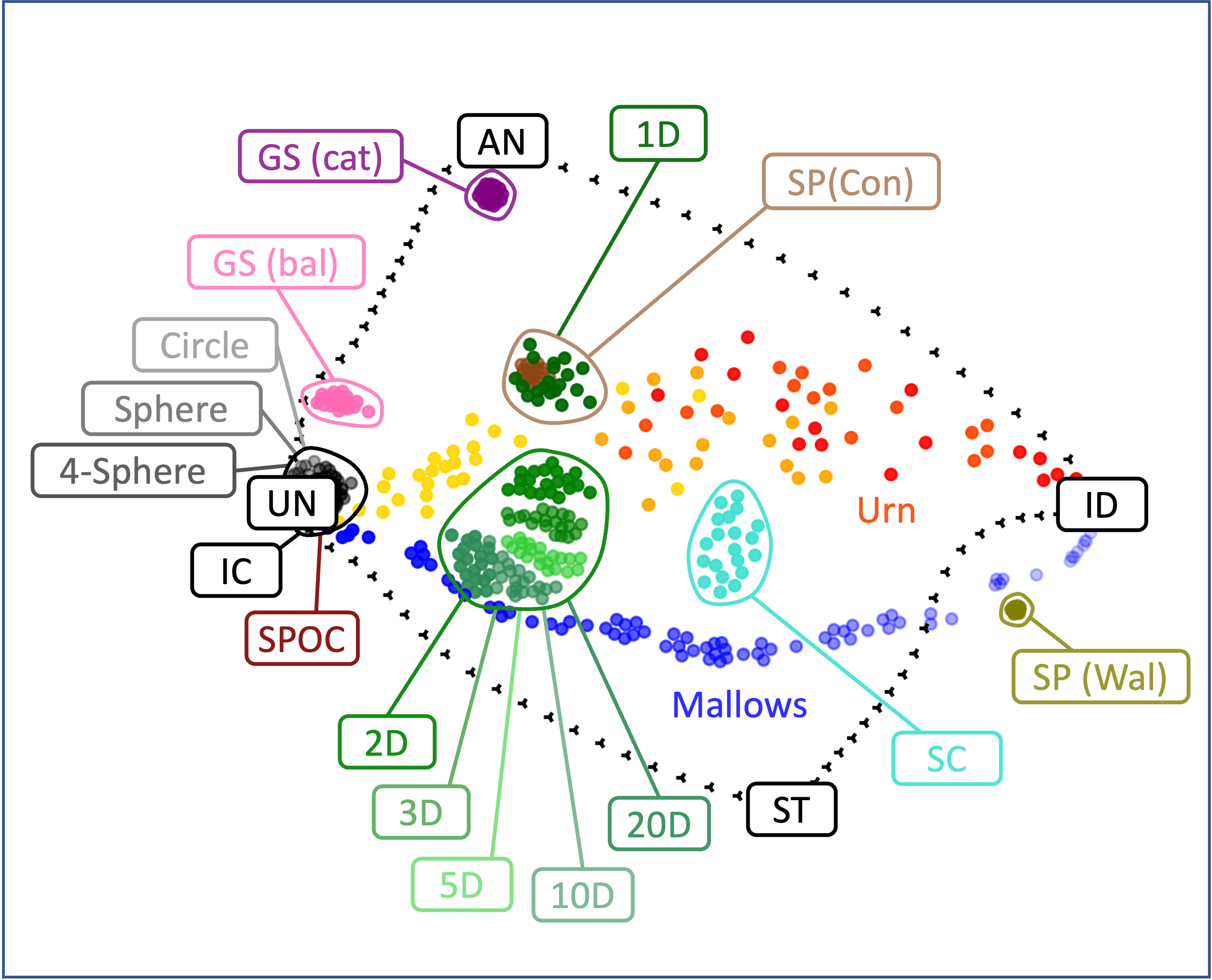}
        \caption{100 candidates and 2000 voters}
    \end{subfigure}
    
    \caption{Maps of elections for datasets with different numbers of voters.}
    \label{fig:scale_voters}
\end{figure}

\subsubsection{Varying the Number of
  Voters}\label{sec:voter-scalability}
One may also ask about the influence of the number of voters on the
shapes of our maps, especially in relation to the number of
candidates. The answer to this question lies in the nature of the
positionwise distance and the work of
\citet{boe-bre-elk-fal-szu:c:frequency-matrices}: As the
positionwise distance operates on elections' frequency matrices, the
more voters there are, the closer are the frequency matrices of elections to the expected
frequency matrices of the culture from which the elections were generated.  As a result, the more voters
there are, the more will our maps resemble the map of these expected
matrices, provided by
\citet{boe-bre-elk-fal-szu:c:frequency-matrices}. This map is very
similar to those we show in \Cref{fig:scale} (except for the one for
four candidates). In particular, this implies that if we were to
consider more and more voters, the location of the elections from the
respective statistical cultures would converge to the locations of
their expected frequency matrices. While this would result in lower
spread of elections from one culture on the map, the general relations
between the statistical cultures would be maintained.  We confirm this
in \Cref{fig:scale_voters}, where we compare the map of the
$100 \times 100$ dataset with maps of elections that have $100$
candidates, but where the numbers of voters are $1000$ and $2000$,
respectively. We see that, indeed, elections generated from each
culture (except for the urn one) are placed more compactly on maps
for higher numbers of voters (this does not apply to the urn
model because urn elections do not have an expected frequency matrix).

\section{Examples of Experiments Using the Map}\label{sec:experiments}

To demonstrate the usefulness of our map framework, we show its
application in several experiments related to the analysis of voting
rules. To this end, we use the $100 \times 100$ dataset from the
previous section. Indeed, some of our experiments involve measuring
running times and, hence, we need a dataset with large enough votes
for the results to be meaningful.  First, we focus on the scores
obtained by winning candidates for two prominent single-winner
rules. Then, we look more closely at the running times of ILP-based
algorithms for selected $\np$-hard rules. Finally, we analyze several
approximation algorithms for finding high-scoring committees.
Overall, the first two groups of experiments focus on the second use
case described in \Cref{sec:usecases} (\emph{Visualizing Election
  Properties}), and the final group focuses on the third one
(\emph{Comparing Algorithm Performance}).

As each of our experiments regards a different collection of voting
rules and/or algorithms, each of them includes a short description of
the necessary preliminaries, where the preliminary part for the first
experiment is also partially applicable to the other experiments.

\subsection{Maps of Winner Scores: Borda and Copeland}\label{sec:map:borda-copeland}

In this section we visualize the scores of the winning candidates
under the Borda and Copeland single-winner voting rules for the
$100 \times 100$ dataset.  Overall, we find that the positions on the
map give relatively good information about the winner score that one
might expect (however, as typical in experimental evaluation, nothing
is perfect).  Knowing the winner scores is useful as it confirms many
natural intuitions (such as that the closer we are to the identity
election, the higher scores we should expect) and, hence, provides
additional arguments for the credibility of our map approach.  We
stress that the conclusions we get here about the scores are not
necessarily universal: For other voting rules the results and
conclusions might be quite different.  Nonetheless, the presented
experiments---especially for the Borda winners---show some of the most
successful applications of the map (the running time experiment for
Harmonic Borda in the next section is similar in this respect).

\subsubsection{Borda and Copeland Rules}\label{sec:borda-copeland}

A single-winner voting rule is a function $f$ that given an election
$E = (C,V)$ outputs a set $W \subseteq C$ of the winners of
$E$. Typically, we expect $W$ to be a singleton, but ties may happen
and then one needs some tie-breaking rule. Fortunately, in our
analysis such tie-breaking will not be necessary. Below we describe
Borda and Copeland voting rules, on which we focus in this experiment.
Both rules are similar in the sense that we compute a score for each
candidate and output those with the highest one:
\begin{description}
\item[Borda.] Each voter assigns~$m-1$ points to the candidate that he
  or she ranks on the top position, $m-2$ points to the second one,
  $m-3$ to the third one, and so on, until the bottom-ranked candidate
  who receives $0$ points from the voter.

\item[Copeland.] We examine all pairs of candidates. In each pair, the
  candidate who is preferred by more than half of the voters gets a
  point. If neither of the candidates is preferred to the other one by
  more than half of the voters, then both of them receive half a point.

\end{description}
Candidate $c$ is a Condorcet winner if for every other candidate $d$,
more than half of the voters rank $c$ ahead of $d$.  Note that if a
Condorcet winner exists, then he or she will always be selected by the
Copeland rule (such rules are called Condorcet extensions).  The
scores of the winning candidates can be computed in polynomial time
for both our rules.

\subsubsection{Score Maps}\label{sec:score-map}

For each of the elections in our dataset, we computed the score of the
winning candidate(s) for both Borda and Copeland.  We present the
results in \Cref{fig:score_single_winner}, where the color of each
point corresponds to the score obtained by the winning candidate(s) in
a given election.
Note that for the Borda map, the paths, which consist of frequency
matrices, are colored as well.  This is so, because for Borda
frequency matrices (together with the number of voters) contain enough
information to compute the candidate scores. For Copeland, we left the
paths' uncolored because the frequency matrices are insufficient to
determine the scores.

\begin{figure}[]
    \centering
    
    \begin{subfigure}[b]{0.49\textwidth}
        \centering
        \includegraphics[width=6.5cm, trim={0.2cm 0.2cm 0.2cm 0.2cm}, clip]
        {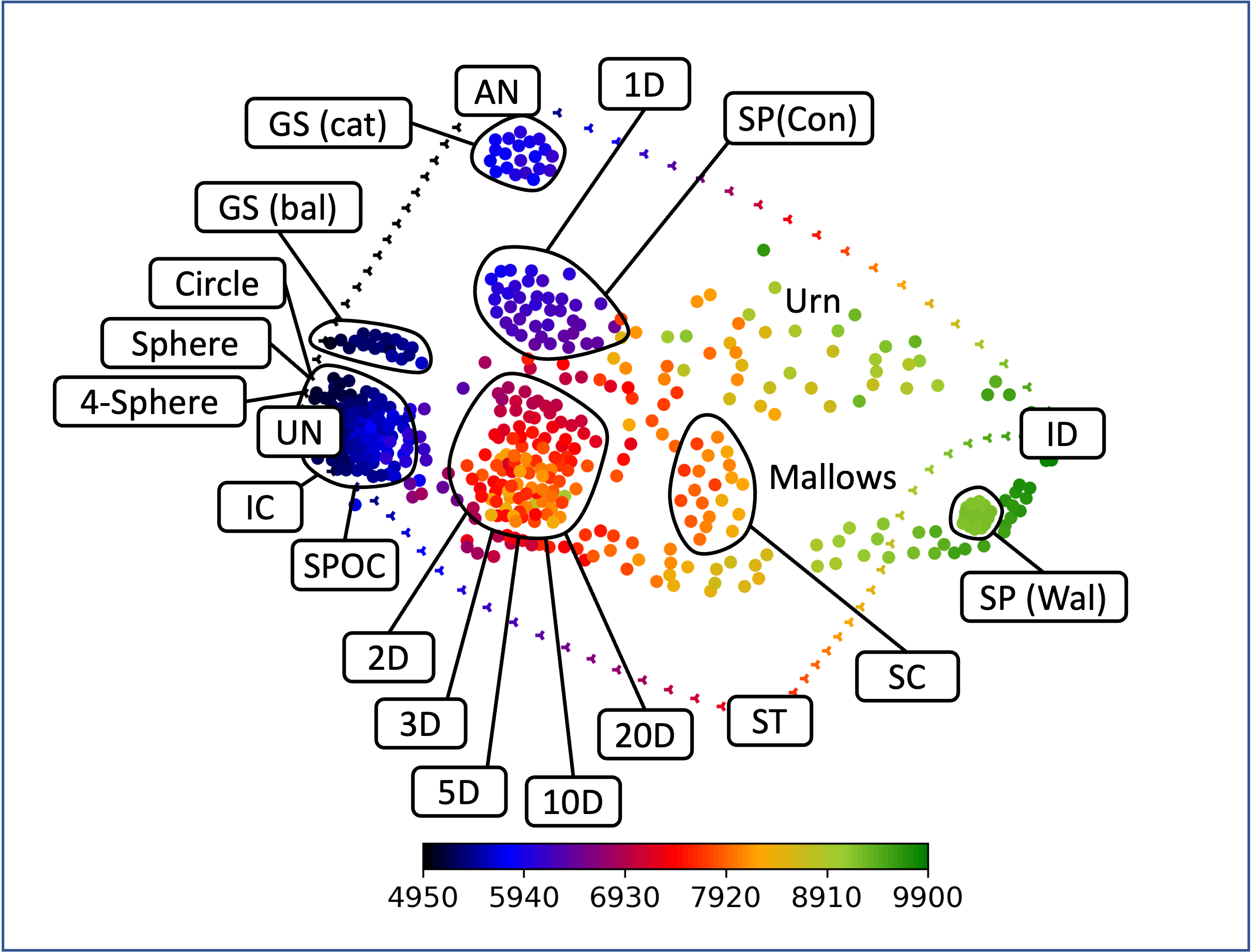}
        \caption{Highest Borda score}
    \end{subfigure}
    \begin{subfigure}[b]{0.49\textwidth}
        \centering
        \includegraphics[width=6.5cm, trim={0.2cm 0.2cm 0.2cm 0.2cm}, clip]
        {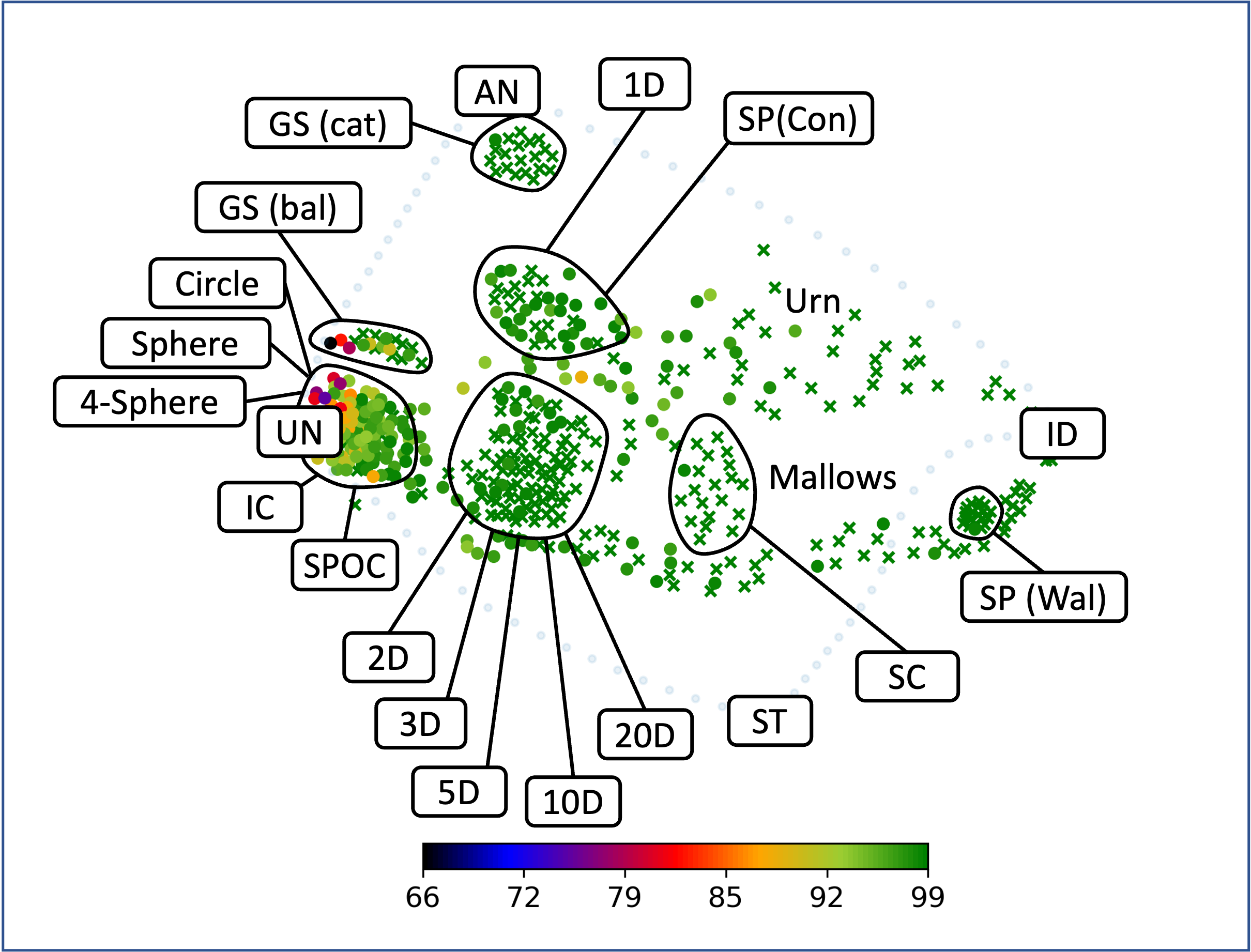}
        \caption{Highest Copeland score}
    \end{subfigure}
    
    \caption{Maps of $100 \times 100$ elections colored according to the score obtained by the winner.}
    \label{fig:score_single_winner}
\end{figure}

As for the Borda coloring, we observe that elections on the line
between UN and AN have the same color. As we move from this line
toward ID, we see a smooth shading.  Consequently, Borda scores of the
winners nicely correlate with the positions of our elections on the
map. As we move closer to ID, the score of the Borda winner is
increasing, and as we move toward UN/AN path, the score is decreasing.
Indeed, the AN-UN path is interesting, as Borda winners have nearly
the lowest possible scores for elections located there.  This means
that elections placed in this part of the map are very close to a tie
according to the Borda rule.

Overall, the smooth shading and interpretability of the map for Borda
is very appealing, but hardly surprising, as the Borda rule---and the
score system it uses---seems particularly well aligned with the nature
of the positionwise distance.

For the Copeland map, with crosses we mark elections that have a
Condorcet winner. In our map~$55\%$ (264 out of 480) of the generated
elections have a Condorcet winner.\footnote{Elections that are
  single-peaked, single-crossing, or group-separable, always have a
  Condorcet winner for an odd number of voters, and always have a weak
  Condorcet winner (or winners) for an even number of voters. A
  candidate $c$ is a weak Condorcet winner if for every candidate $d$,
  at least half of the voters prefer $c$ to $d$.} In particular,
almost all the single-crossing, Walsh, caterpillar group-separable,
3-Cube, 5-Cube, 10-Cube, 20-Cube, and around half of 1D-Interval,
2D-Square, Conitzer, and balanced group-separable elections have a
Condorcet winner. Even one IC election has a Condorcet winner. For the
urn and Mallows elections, having a Condorcet winner is strongly
related to their parameters. However, there are also elections without
Condorcet winners that are located near those that have them. This is
quite natural as relations between pairs of candidates are not
reflected directly in frequency matrices. Yet, it is reassuring that
for such elections Copeland scores of the winners still tend to be
high (except for a few urn elections), so---in a certain formal
sense---they are close to having Condorcet winners.

It is striking that for most elections the Copeland score of the
winner is very high. Indeed, it is greater or equal to 90 in 449
elections out of the 480 present in the $100 \times 100$ dataset.  The lowest
Copeland score is witnessed by a balanced group-separable election. In
general, the vast majority of lowest values are obtained by Circle,
Sphere and~$4$-Sphere, as well as some SPOC and balanced
group-separable elections. Note that all values are much larger than
the lowest possible score, which is~$49.5$. 
Keeping all that in mind, the position on the map is still a good
indication for the score of the Copeland winner: The score is very
high, unless the election is very close to the UN compass matrix
(however, a few of the urn elections have lower Copeland scores while
being located at some distance from UN). Further, the elections in the
vicinity of ID tend to have Condorcet winners (which is hardly
surprising) and even those as far as 3/4 of the UN-ID distance
from ID are likely to have Condorcet winners (which is far more
interesting).

\subsection{Maps of ILP Running Times: Dodgson and Harmonic-Borda}\label{sec:map:ilp-times}

Next, we consider two rules that are $\np$-hard to compute, namely the
single-winner Dodgson rule and the multiwinner Harmonic Borda rule
that we mentioned in our motivating example in \Cref{sec:mot-example}.
For both of these rules there are natural ILP formulations and we will
evaluate how much time an ILP solver needs to compute Dodgson winners
and Harmonic Borda winning committees on our elections.  As throughout
this section and most of the paper, we focus on the 100x100 dataset.

All the times that we report were obtained using the CPLEX ILP solver
on a single thread (Intel(R) Xeon(R) Platinum 8280 CPU @ 2.70GH) of a
448 thread machine with 6TB of RAM.

\subsubsection{Dodgson and Harmonic Borda}\label{sec:dec-dod-hb}
The score of a candidate under the Dodgson rule is the smallest number
of swaps of adjacent candidates that need to be performed in the votes
to make him or her the Condorcet winner. The candidate(s) for whom
this value is lowest win(s).
It is well-known that deciding if a given candidate is a Dodgson
winner is $\np$-hard~\citep{bar-tov-tri:j:who-won} and even
$\Theta_2^p$-complete~\citep{hem-hem-rot:j:dodgson}. However, one can
compute Dodgson winners using an ILP solver and an ILP
formulation~\citep{bar-tov-tri:j:who-won,car-cov-fel-hom-kak-kar-pro-ros:j:dodgson},
or use one of the available $\fpt$
algorithms~\citep{bet-guo-nie:j:dodgson-parametrized}, or---if perfect
precision is not required---use approximation
algorithms~\citep{car-cov-fel-hom-kak-kar-pro-ros:j:dodgson,car-kak-kar-pro:c:dodgson-acceptable}.

Harmonic Borda (HB) is a multiwinner voting rule, which means that given an
election $E$ and committee size $k$ it outputs size-$k$ candidate
subsets as the (tied) winning committees. The rule proceeds as
follows.  Let $E = (C,V)$ be an election and let $k$ be the committee
size.  For a committee $S$ of $k$ candidates and a vote $v$, we let
the HB score that $v$ assigns to $S$ be the result of the following
procedure:
\begin{enumerate}
\item[] For each candidate $c \in S$, we compute this candidate's
  position $\pos_v(c)$ in $v$ and sort these values from lowest to highest,
  obtaining vector $(p_1, \ldots, p_k)$. The HB score assigned to $S$
  by $v$ is
  $(p_1-1) + \frac{1}{2}(p_2-1) + \frac{1}{3}(p_3-1) + \cdots +
  \frac{1}{k}(p_k-1)$.
\end{enumerate}
The HB score of committee $S$ is the sum of the HB scores that each
voter assigns to it. HB outputs all the size-$k$ committees with the
lowest HB score (the HB score we define here is sometimes called a
\emph{dissatisfaction} score, and some authors consider a
\emph{satisfaction} version instead, where we replace candidates'
positions in the above computations with Borda scores and choose the
committee with the highest score; both definitions are equivalent). HB
is a variant of the Proportional Approval Voting rule (PAV) adapted to
the setting of ordinal
voting~\citep{Thie95a,kil:chapter:approval-multiwinner,lac-sko:b:multiwinner-approval}.

\begin{example}\label{ex:hb}
  Consider an election~$E = (C,V)$ with committee size $k=2$, where~$C = \{a,b,c,d\}$, $V = (v_1,v_2,v_3,v_4)$, and the votes are:
  \begin{align*}
    \small
    v_1\colon&  a \pref b \pref c \pref d \\
    v_2\colon&  a \pref b \pref c \pref d, \\
    v_3\colon&  a \pref b \pref c \pref d, \\
    v_4\colon&  d \pref c \pref b \pref a.
 \end{align*}
 Under HB rule,~$\{a, b\}$ is the unique winning committee with the score
 of~$5$: Each of the first three voters contributes score
 $\nicefrac{1}{2}$, and the last one contributes $2+\nicefrac{3}{2}$.
\end{example}

Unfortunately, identifying a winning committee under HB is
NP-hard~\citep{fal-sko-sli-tal:c:paths,sko-fal-lan:j:collective}.  We
can try to overcome this issue, for example, by formulating the
problem as an integer linear program (ILP) and solving it with an
off-the-shelf ILP solver, or by using $\fpt$
algorithms~\citep{bet-sli-uhl:j:mon-cc,bre-fal-kac-kno-nie:c:fpt-multiwinner},
or by designing efficient approximation algorithms (see, e.g., the
works of \citet{sko-fal-lan:j:collective},
\citet{fal-lac-pet-tal:c:csr-heuristics} and
\citet{mun-she-wan:c:approx-cc}).  For the case of HB, it is also
possible to use algorithms (both approximation or exact ones) designed
for the PAV
rule~\citep{fal-sko-sli-tal:c:top-k-counting,byr-sko-sor:c:pav-approx,dud-man-mar-sor:c:pav-tight-approx},
which can be adapted to the case of HB using a trick described by
\citet{bre-fal-kac-kno-nie:c:fpt-multiwinner}.
For a general overview of multiwinner rules we point to the work of
\citet{fal-sko-sli-tal:b:multiwinner-voting}.

\subsubsection{ILP Running Time Maps}\label{sec:running-time}
\begin{figure}[]
    \centering
    
    \begin{subfigure}[b]{0.99\textwidth}
        \centering
        \includegraphics[width=6.5cm, trim={0.2cm 0.2cm 0.2cm 0.2cm}, clip]{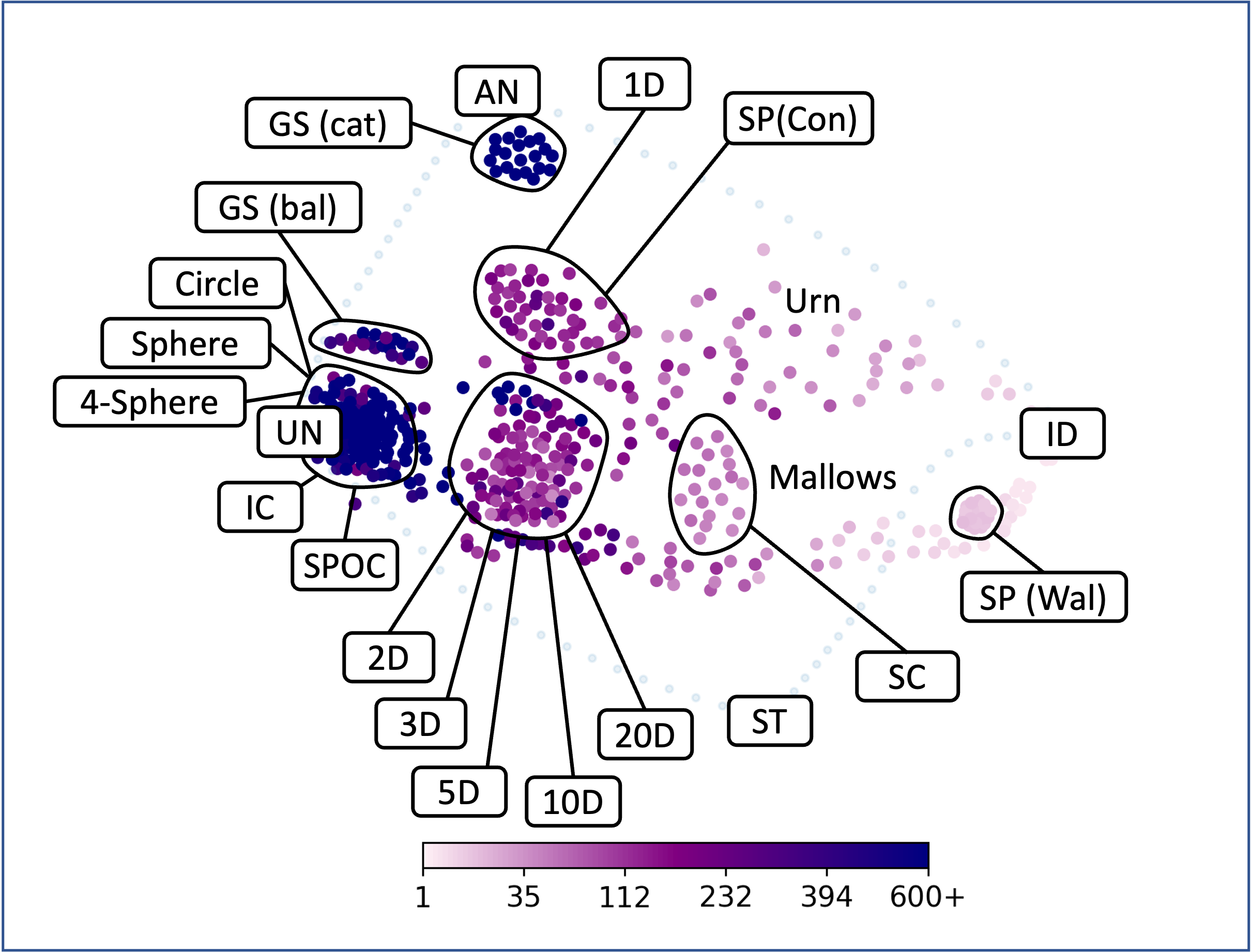}
        \includegraphics[width=6.5cm, trim={0.2cm 0.2cm 0.2cm 0.2cm}, clip]{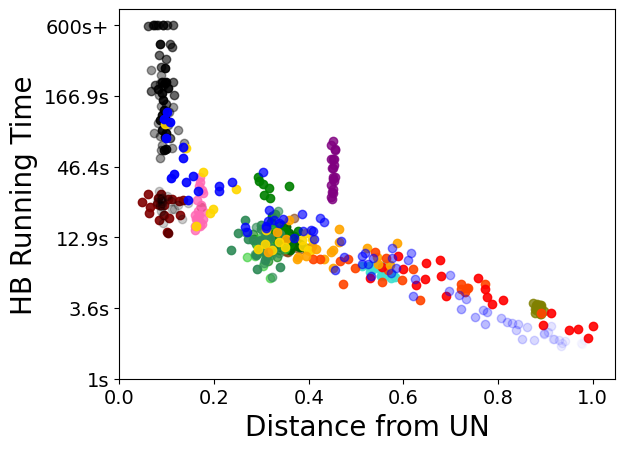}
        \caption{HB running time (the color scale on the left is quadratic; the $y$-axis on the right has logarithmic scale).}
    \end{subfigure}
    \vspace{4mm}

    \begin{subfigure}[b]{0.99\textwidth}
        \centering
        \includegraphics[width=6.5cm, trim={0.2cm 0.2cm 0.2cm 0.2cm}, clip]{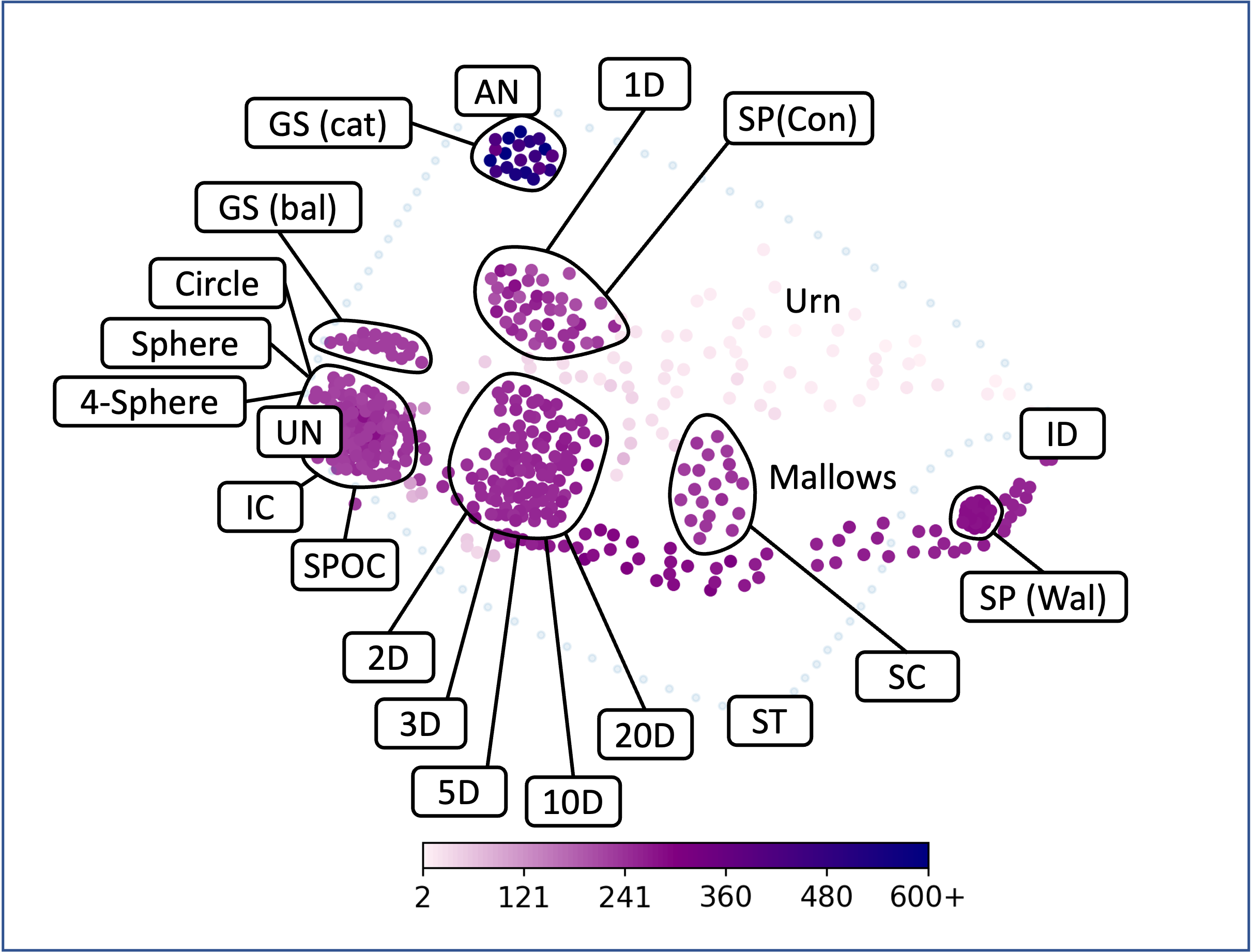}
        \includegraphics[width=6.5cm, trim={0.2cm 0.2cm 0.2cm 0.2cm}, clip]{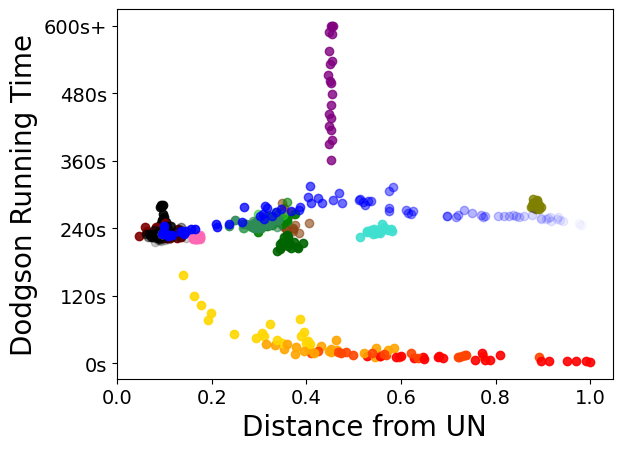}
        \caption{Dodgson running time (the color scale on the left is linear; the $y$-axis on the right has linear scale).}
    \end{subfigure}
    
    \caption{ILP running times (in seconds) for Dodgson and HB for the
      $100 \times 100$ dataset. On the left, we show maps where each
      election is colored according to the time required by the ILP
      solver. On the right, each point is an election (colored
      according to its statistical culture; the color scheme is as in
      \Cref{fig:maps100x100}) and its $x/y$ coordinates indicate its
      normalized positionwise distance from $\ID$ and the time taken
      by the ILP solver.}
    \label{fig:voting_rules_time}
\end{figure}

We now analyze the time needed to compute the outcomes of the Dodgson
and HB rules using ILP formulations and the CPLEX solver. We report
the achieved running times in Figure~\ref{fig:voting_rules_time}. On
the left we show maps, where the colors give the running times (the
darker the color, the longer was the computation time). On the right,
we plot the relation between the running times needed by particular
elections (on the $y$ axis) and their normalized positionwise distance
from UN (on the $x$ axis). The points in these plots are elections,
colored according to the statistical culture from which they were
generated, using the same color scheme as in
\Cref{fig:maps100x100}. Note that for HB in these plots the $y$-axis
has logarithmic scale, but for Dodgson it is linear.
We use the ``presentation limit'' of $600$ seconds.  All instances
that required more time than the limit have the same color in our
maps. Moreover, the coloring scale for the HB map is quadratic,
whereas for Dodgson it is linear.\footnote{We would prefer to have the
  same scale for both pictures, but then it would have been hard to
  see anything interesting either on one or the other.}
In~\Cref{table:time_hb_dodgson} we present average running times and
their standard deviations for each statistical culture and each rule
that we considered. 

We first consider the HB rule.  The most challenging instance took
five hours to solve (and, hence, exceeded the ``presentation limit''
and is reported as ``600+'' on the figure). The simplest instance
required less than two seconds. Ten worst cases were witnessed
by~$4$-Sphere elections, which suggests that it is particularly hard
to find optimal winning committees under the HB rule for elections
from this model.  Perhaps the most visible phenomenon on the map is
that the ILP solver tends to take most time on elections similar to
those generated by the impartial culture, and the farther elections
are away from $\UN$, the less time is needed.  Consequently, we see
that in this experiment the position on the map is highly meaningful.

\begin{table}
    \centering
    \footnotesize
    \begin{tabular}{c | c | c | c | c }
     \toprule
        Culture &  \multicolumn{2}{c|}{HB} & \multicolumn{2}{c}{Dodgson} \\
          & avg. time & std. dev.  & avg. time & std. dev. \\
    \midrule
    Impartial Culture  & 155.7s & 85.5 & 251.4s & 17.2 \\
    \midrule
    Conitzer SP  & 12.2s & 2.2 & 252.7s & 14.9 \\
    Walsh SP  & 3.6s & 0.2 & 281.1s & 5.1 \\
    SPOC  & 22.3s & 4.1 & 230.0s & 8.0 \\
    Single-Crossing  & 7.1s & 0.6 & 235.6s & 4.4 \\
    \midrule
    1D & 11.8s & 1.4 & 213.2s & 7.8 \\
    2D & 20.6s & 9.1 & 249.9s & 7.3 \\
    3D & 11.4s & 1.8 & 254.1s & 10.6 \\
    5-Cube  & 10.8s & 2.0 & 251.0s & 6.7 \\
    10-Cube & 12.0s & 3.4 & 249.3s & 5.2 \\
    20-Cube & 12.1s & 3.0 & 249.7s & 6.0 \\
    \midrule
    Circle & 22.9s & 3.7 & 222.3s & 2.9 \\
    Sphere & 126.8s & 64.9 & 230.5s & 5.0 \\
    4-Sphere & 1614.1s & 3912.1 & 232.3s & 6.0 \\
    \midrule
    Balanced GS & 22.2s & 6.9 & 223.0s & 2.1 \\
    Caterpillar GS & 44.7s & 14.3 & 513.0s & 106.1 \\
    \midrule
    Urn & 11.1s & 13.1 & 30.0s & 35.1 \\
    Mallows & 17.3s & 24.1 & 264.2s & 19.7 \\
    	\bottomrule

        \end{tabular}
    \caption{\label{table:time_hb_dodgson} Analysis of ILPs running time for HB and Dodgson rules.}
    
\end{table}

For the Dodgson rule the situation is quite different. Most instances
need the same amount of time (i.e., around four minutes on
average). Two exceptions are urn elections (which needed only half a
minute per election on average) and caterpillar group-separable
elections (which needed eight and a half minutes per election on
average). 

As we can see on the plots on the right side of
Figure~\ref{fig:voting_rules_time}, for HB all of the hardest
instances were very close to $\UN$, while for the Dodgson rule this is
not the case. Indeed, the Dodgson required most time on caterpillar
group-separable elections. This is very interesting as these are
highly structured elections and one would expect that finding winners
for them would be easy. This is a clear indication that there is value
in considering models producing elections with special structure that
at first may seem quite arbitrary (indeed, it is unlikely that
caterpillar group-separable preferences would ever arise in practice)
and looking at maps containing elections that are as diverse as
possible.  On the downside, we have no guarantee that Dodgson (or any
other of our rules) would not take even more time on yet different
kinds of elections, not present in our dataset.
Altogether, the Dodgson ILP running-time experiment illustrates a
situation where elections' positions on the map do not necessarily
correspond to the analyzed feature, but the fact that maps encourage
the use of diverse datasets helps in spotting interesting phenomena.

\subsection{Maps of Approximation Ratios: Chamberlin--Courant}\label{sec:map:approx}

In our final example of an experiment in this section, we compare two
approximation algorithms for the Chamberlin--Courant multiwinner rule.
Here we will not see any particularly strong connection between
experiment results and the positions of the elections on the map, but
still the maps will be helpful in getting intuitions and observing the
results.

\subsubsection{Chamberlin--Courant}
Chamberlin--Courant (CC) is a multiwinner voting rule that proceeds
similarly to HB, but with a different notion of a score. Specifically,
given an election $E = (C,V)$ and a committee $S$ of size $k$, the CC
score of $S$ is defined as:
\begin{align}
    \sum_{v \in V} \min_{c \in S}\big(\pos_v(c)-1\big).
\end{align}
The intuition here is that each voter picks the member of the
committee that he or she ranks highest (i.e., who has the lowest
position) and views this candidate as his or her ``representative.''
The voter adds this candidate's position (minus one) to the committee
score. The rule selects the committee with the lowest score, i.e.,
each voter prefers to be represented by the candidate that he or she
ranks as highly as possible. As in the case of HB, this is the
dissatisfaction variant of the score and there is also a satisfaction
variant, which leads to the same rule.

\begin{example}
  Consider the same election as in \Cref{ex:hb}.  According to the CC
  rule,~$\{a,d\}$ is the winning committee with a score of~$0$
  (which is the smallest possible in any election).
\end{example}

Computing a CC winning committee is
NP-hard~\citep{pro-ros-zoh:j:proportional-representation,
  bou-lu:c:chamberlin-courant, bet-sli-uhl:j:mon-cc}, but as for HB,
there are ILP formulations, $\fpt$
algorithms~\citep{bet-sli-uhl:j:mon-cc,bre-fal-kac-kno-nie:c:fpt-multiwinner},
and efficient approximation algorithms (see, e.g., the works of
\citet{sko-fal-lan:j:collective},
\citet{fal-lac-pet-tal:c:csr-heuristics} and
\citet{mun-she-wan:c:approx-cc}). In particular, we are interested in
the following two algorithms:
\begin{description}
\item[Sequential~CC.] Sequential~CC starts with an empty
  committee 
  and works in~$k$ iterations, where in each of them it adds to the
  committee a single candidate, so that the resulting committee has as
  low total (dissatisfaction) score as possible.
  \citet{mun-she-wan:c:approx-cc} presents a detailed analysis of this
  algorithm, and \citet{nem-wol-fis:j:submodular} give a general
  discussion of such algorithms (in particular, if we are interested
  in voter satisfaction rather than dissatisfaction, then the former
  paper shows that Sequential~CC is a polynomial-time approximation
  scheme, PTAS, and the latter shows a general approximation ratio of
  $1-\nicefrac{1}{e}$ of such greedy algorithms for a whole class of
  related problems). The first explicit mention of this algorithm in
  the context of CC was in the work of
  \citet{bou-lu:c:chamberlin-courant}.

\item[Removal~CC.] Removal~CC proceeds similarly to Sequential~CC, but
  it starts with a committee containing all the candidates and works
  in~$m-k$ iterations, in each of them removing a single candidate, so that
  the resulting committee has as low total dissatisfaction as possible.
  This algorithm belongs to the folklore.
\end{description}

Sequential~CC and Removal~CC, as well as their variants for other
multiwinner rules, are well-known in the literature and are often used
to compute winning
committees~\citep{bou-lu:c:chamberlin-courant,sko-fal-lan:j:collective,sko-lac-bri-pet-elk:c:proportional-rankings,fal-lac-pet-tal:c:csr-heuristics}.

\subsubsection{Approximation Quality Comparison Maps}\label{sec:approx-map}

\begin{table}
  \centering
  \begin{tabular}{c|cc}
    \toprule
    Model              & \multicolumn{2}{c}{average approximation ratio} \\
            &   Sequential~CC & Removal~CC \\
    \midrule
    Impartial Culture           &  1.142 & 1.17 \\
    \midrule
    Single-Peaked (Conitzer)  & 1.108 & 1.061 \\
    Single-Peaked (Walsh)      &  1.084 & 1.028 \\
    SPOC                       & 1.115 & 1.060 \\
    Single-Crossing           & 1.120 & 1.000 \\
    \midrule
    1D         & 1.121 & 1.056 \\
    2D         & 1.188 & 1.131 \\
    3D         & 1.228 & 1.126 \\
    5-Cube         & 1.150 & 1.135  \\
    10-Cube         & 1.048 & 1.101 \\
    20-Cube         & 1.052 & 1.084 \\
    \midrule
    Circle         & 1.108 & 1.081  \\
    Sphere         &  1.265 & 1.148 \\
    4-Sphere   & 1.263 & 1.174 \\
    \midrule
    Group-Separable (Balanced)       &  1.045 & 1.030 \\
    Group-Separable (Caterpillar)    & -- & -- \\
    \midrule
    Urn        &  2.375 & 1.192 \\
    Mallows    & 1.037 & 1.053  \\
    \bottomrule
  \end{tabular}
  \caption{\label{tab:mw-approximation} Average approximation ratios
    achieved by two algorithms for the CC rule on the
    elections from the $100 \times 100$ dataset.}
\end{table}

We evaluate our two approximation algorithms by computing the
approximation ratios that they achieve on each of the elections from
the $100 \times 100$ dataset.  The approximation ratio (that a given
algorithm achieves on a given election for a given committee size) is
the ratio between the (dissatisfaction) score of the committee provided
by the algorithm and the lowest possible one.  The smaller (and closer
to $1$) an approximation ratio is, the better.  One complication is
that for some elections the lowest possible dissatisfaction score is
zero. We omit such elections in this experiment as they are easy to
identify and it is unreasonable to use an approximation algorithm in
their case. Indeed, to form optimal committees for such elections it
suffices to take each voter's top-ranked candidate and, possibly,
arbitrary other candidates to fill-in the committee size.

\begin{figure}[t]
    \centering

    \begin{subfigure}[t]{0.45\textwidth}
      \includegraphics[width=8cm, trim={0.2cm 0.2cm 0.2cm 0.2cm}, clip]{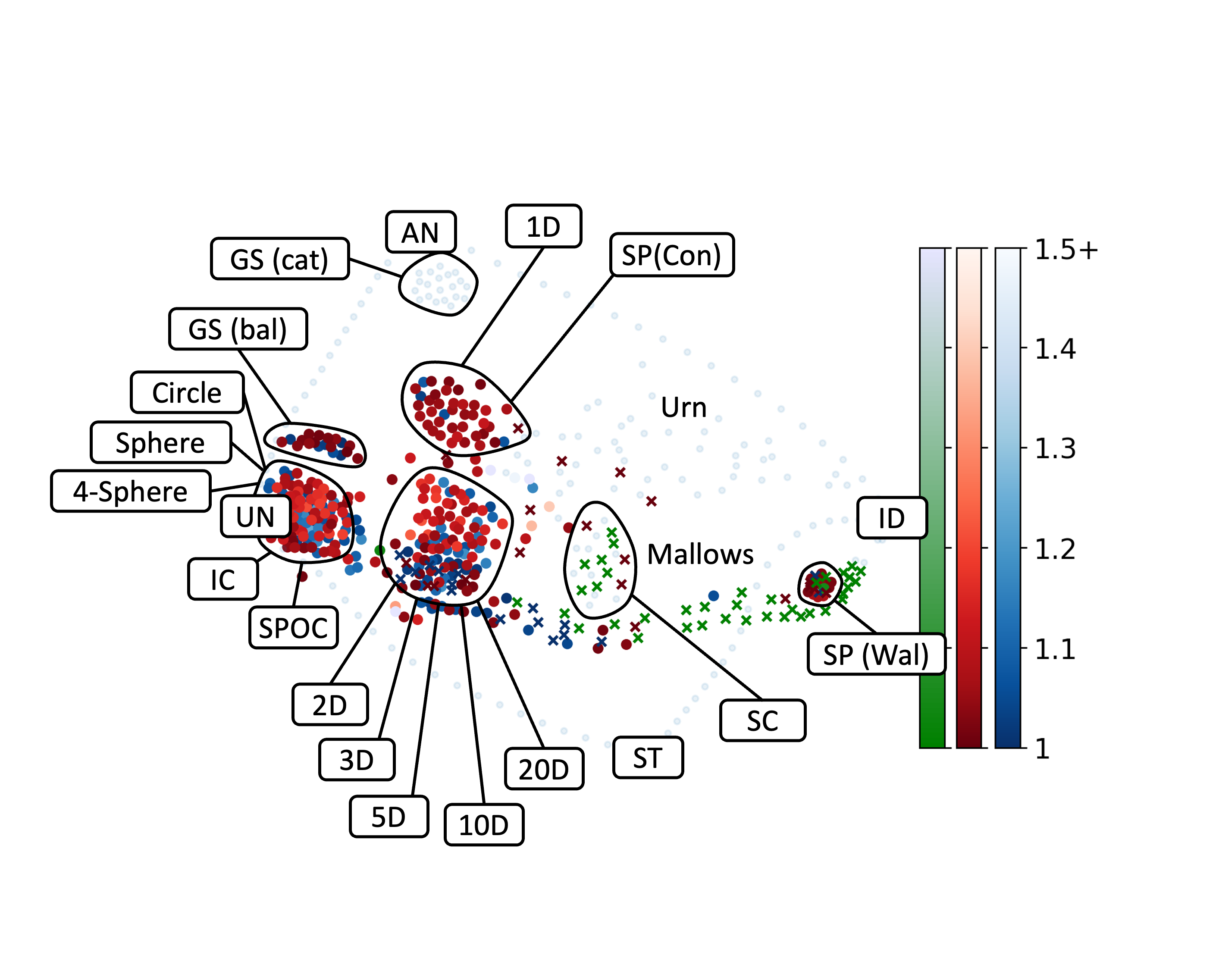}%
      \caption{\label{fig:approx-ratio} Approximation ratios of the
        better among the Sequential CC (blue) and Removal~CC (red)
        algorithms. Green points indicate elections where both
        algorithms achieved the same approximation ratios. Crosses
        indicate elections for which the respective algorithm found an
        optimal solution.}
    \end{subfigure}
    \quad
    \begin{subfigure}[t]{0.45\textwidth}
      \includegraphics[width=8cm, trim={0.2cm 0.2cm 0.2cm 0.2cm}, clip] {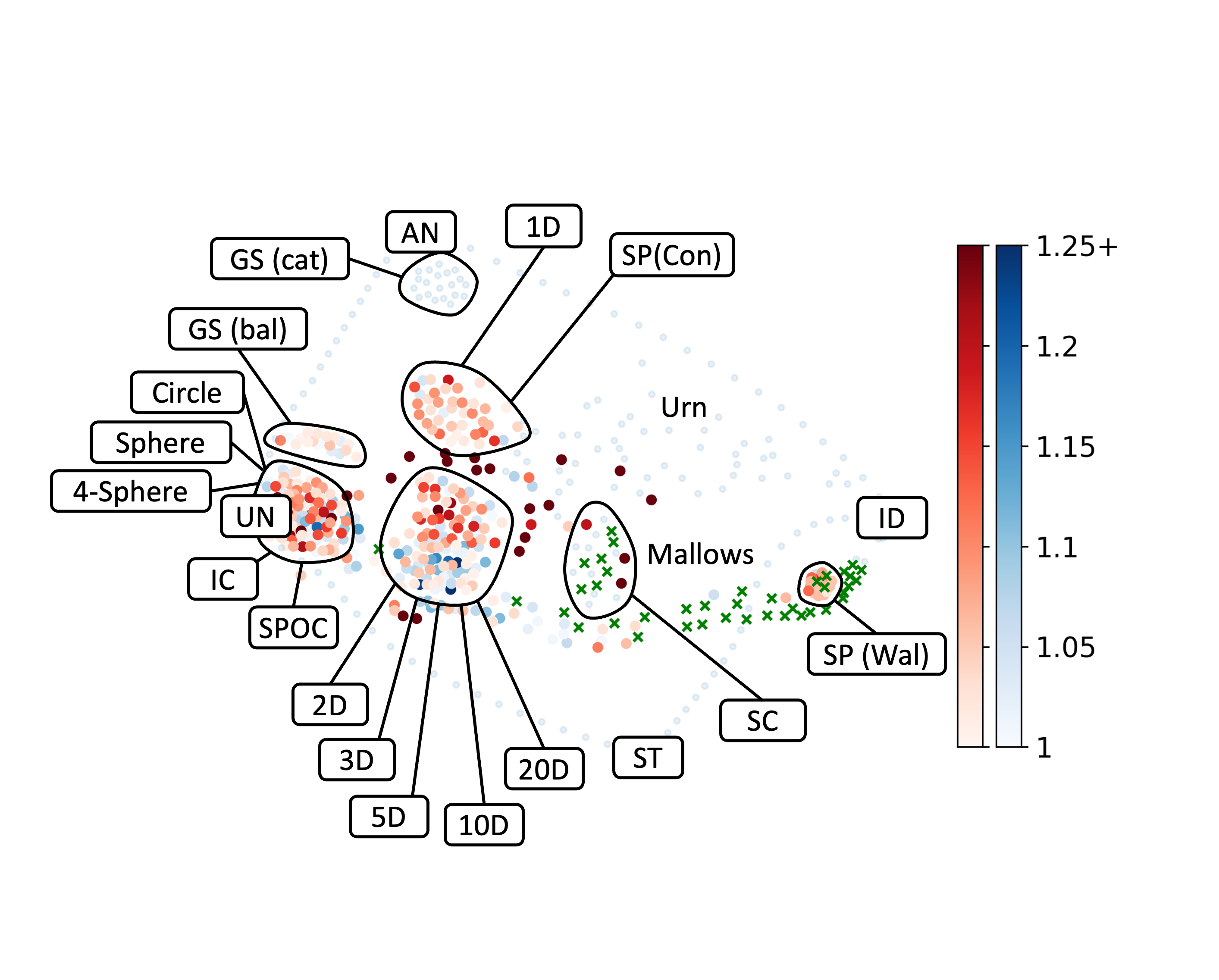}
      \caption{\label{fig:approx-discr} Indication how much better was
        the algorithm that achieved lower (better) approximation ratio
        than the other.  For each election, the intensity of the color
        corresponds to the value $\max(A,B)/\min(A,B)$, where $A$ is
        the approximation ratio achieved by Sequential CC and $B$ is
        the approximation ratio achieved by Removal CC. Colors specify
        which algorithm performed better (blue means Sequential CC,
        red means Removal CC, and green means that both algorithms
        achieved the same result.}
    \end{subfigure}
    
    \caption{Analysis of the approximation algorithms for CC.}

    \label{fig:approx_2in1}
\end{figure}

We present the approximation ratios achieved by our algorithms in an
aggregate form in \Cref{tab:mw-approximation} Generally, both
algorithms perform very well, except for urn elections where
Sequential~CC is notably worse than Removal~CC.  Turning to a
comparison between them, the results in \Cref{tab:mw-approximation}
either suggest that both Sequential~CC and Removal~CC are, more or
less, equally good, or point marginally in the direction of Removal~CC
(e.g., if one were to naively count the number of entries in the table
where its approximation ratio is smaller than that of
Sequential~CC). It turns out that the truth is a bit more
complicated. Indeed, while on many elections both algorithms perform
comparably well, on a number of others their outcomes differ
substantially.  We illustrate this effect in \Cref{fig:approx-ratio},
where the blue points refer to elections where Sequential~CC is better
at approximating CC, the red points mark the elections where
Removal~CC is better, and the green points depict elections where
there is a draw (crosses indicate elections where the algorithms found
optimal solutions). Removal~CC often achieves better approximation
ratios, even though theoretical arguments suggest that Sequential~CC
is a better algorithm (i.e., there are stronger guarantees for its
approximation ratio in the satisfaction-based model). Further, we see
that in the $\UN$ area, both algorithms are sometimes superior, and as
far as hypercube elections are concerned, the higher dimension, the
more successful Sequential~CC becomes; indeed, this is also visible in
\Cref{tab:mw-approximation}. Finally, it is quite striking how often
both algorithms give the same result. This is not as surprising close
to $\ID$, where all the votes are similar, but the phenomenon appears
also at some nontrivial distance from $\ID$.

Next, in \Cref{fig:approx-discr}, we present a different way for
comparing the two algorithms. For each election we compute the value
$\max(A,B)/\min(A,B)$, where $A$ is the approximation ratio achieved
by Sequential~CC and $B$ is the approximation ratio achieved by
Removal~CC, and color the elections according to this value (and the
algorithm that performed better). This value indicates how much better
the performance of the better approximation algorithm was. We see that
while for many elections both algorithms achieve similar approximation
ratios, there also are cases where either of the algorithms performs
notably better than the other one (this is particularly visible for
some urn elections and for Euclidean elections with higher
dimensions).

Overall, we believe that the maps from \Cref{fig:approx_2in1} add
interesting insights that are not possible to observe in
\Cref{tab:mw-approximation} (but, of course, one could spot them by
looking at raw results or by a more careful statistical analysis). For
example, we see that even though Removal~CC is slightly better on
elections from many statistical culture, Sequential~CC still obtains
better results in a nonnegligible number of cases.

\subsection{Conclusions From Using the Map in Experiments}

Let us take a moment to reflect on how the map approach was helpful in
our experiments. While, in principle, each of the experiments could
have been performed without the map, and could have led to the same
high-level conclusions, we believe that using the map it was far
easier to reach them. Further, the sheer fact of being able to look at
the results for all elections, in a non-aggregate way, motivated
particular questions and further lines of inquiry. For example, if we
were analyzing the running time of the ILP formulation for Dodgson in
\Cref{sec:running-time} without the map, it would be quite likely that
we would have simply not thought of looking at caterpillar
group-separable elections. This way, we would have lost the
observation that they are much more challenging for the ILP solver
than other, more typical models. The map aims to bring a varied set of
elections together and, so, helps in finding special cases where
algorithms behave in unusual ways.  In particular, caterpillar
group-separable elections are in the part of the map where no other
elections reside, so if we looked at the map without them, we would be
prompted to look for a way of generating data in this area.
In \Cref{sec:approx-map} looking at the map helped in understanding
and justifying that Sequential~CC and Removal~CC complement each
other.
Finally, the score maps from \Cref{sec:score-map} gave quite some
insight regarding the nature of the elections. For example, it is
striking how many elections have very high Copeland winner scores and
how many elections with Condorcet winners there are in our dataset,
even quite far away from $\ID$.

Overall, we found that if two elections were close on the map, then in
all of our experiments---with the possible exception of the last one,
in \Cref{sec:approx-map}---they also had similar properties. Sometimes
this was because positions on the maps had strong connection to the
analyzed features, and sometimes this seemed more connected to the
fact that the two elections were from the same statistical culture.
We should also stress that this is might be connected to us
considering fairly high-level properties that are robust to small
changes in the elections (for example, swapping two candidates can
only affect their Borda scores in a fairly insignificant way, and may
not change the Borda score of the election's winner). If we considered
more fragile properties, the outcome might be different. Nonetheless,
even in such cases the maps could help in identifying general trends.

\section{Map of Real-Life Elections} \label{sec:real-life} \label{preflib_info}

Our final goal is to demonstrate the usefulness of the map for
analyzing real-life elections. To this end, we consider several
datasets coming from Preflib \citep{mat-wal:c:preflib} and collected
by \citet{boe-sch:c:real-life-elections}.
There are two main problems with real-life election data. One is that
these elections typically involve only a few candidates. The second
one is that in many cases numerous votes are incomplete. Thus we will
first describe our general strategy of preprocessing the available
data (e.g., how we extend incomplete preference orders) and, then, we
will describe our datasets together with the exact preprocessing that
we applied to each of them. Indeed, some datasets needed some special
treatment. Next, we (i)~present where these real-life elections land
on our maps of elections, (ii)~color these maps according to the Borda
and Copeland scores and, finally, (iii)~we find ($\normphi$)
parameters of the Mallows model that best approximate our real-life
elections.
Whenever we speak of real-life elections in this section, we mean
elections from our datasets. Other ``real-life'' elections, from other
sources, may behave differently (even if we hope that our datasets are
at least somewhat representative, we cannot provide any guarantees).

\subsection{Forming the Datasets}
In \Cref{tab:preflib_selected}, we present a detailed description of
the real-life datasets that we decided to use. All of them are
available in Preflib~\citep{mat-wal:c:preflib} (the TDF, GDI, and
Speed Skating data was collected by
\citet{boe-sch:c:real-life-elections}).  We chose eleven real-life
datasets of different types, and we divided them into three
categories.  The first group contains \textit{political} elections:
city council elections in Glasgow and
Aspen~\citep{openstv}, elections from Dublin North and
Meath constituencies (Irish), and elections held by non-profit
organizations, trade unions, and professional organizations (ERS). The
second group consists of \textit{sport} elections: Tour de France
(TDF)~\citep{boe-sch:c:real-life-elections}, Giro d'Italia
(GDI)~\citep{boe-sch:c:real-life-elections}, speed
skating~\citep{boe-bre-fal-nie-szu:c:compass}, and figure skating. The
last group consists of \textit{surveys}: preferences over
Sushi~\citep{kam:c:sushi}, T-Shirt designs, and costs of living and
population in different cities~\citep{car-cha-kri-vou:j:optimizing}.
For TDF and GDI, each race is a vote, and each season is an
election. For speed skating, each lap is a vote, and each competition
is an election. For figure skating, each judge's opinion is a vote,
and each competition is an election.

We are interested in elections that have at least~$10$
candidates.\footnote{The more candidates, the more interesting is the
  data and the resulting maps. On the other hand, if we had required too many candidates, we
  would end up having very few election instances. In this sense,~$10$
  candidates is a tradeoff between the number of candidates and the
  number of instances that include as many candidates.}  As our map
framework requires complete votes without ties, we are interested in
instances where votes are as complete as possible and contain only a
few ties. In some datasets, only parts of the data meets our criteria
(i.e., complete votes without ties over at least~$10$ candidates). For
example, in the dataset containing Irish elections, we have three
different elections, but one of them (an election from the Dublin West
constituency) contains only nine candidates. We remove all 
elections not meeting our criteria from consideration.

We further reduce the number of elections by considering only selected
ones.  We based our decisions on the number of voters and
candidates. That is, for ERS, we only take elections with at
least~$500$ voters, for Speed Skating with at least~$80$ voters, for
TDF with at least~$20$ voters, and for Figure Skating with at
least~$9$ (these numbers are in accordance with the sizes of the
elections in the particular datasets). In addition to that, for TDF,
we only selected elections with no more than~$75$ candidates.
Naturally, these choices are somewhat arbitrary, but some such
decisions are necessary to get a manageable collection of data.
In \Cref{tab:preflib_selected}, we include in the column
\textit{\#~Selected Elections} the number of elections we selected
from each dataset.

\begin{table*}[t]
\centering
\resizebox{\textwidth}{!}{\begin{tabular}{c  c  c  c  c  c}
			\toprule
			Category & Name & \# Selected Elections & Avg.~$m$ & Avg.~$n$ & Description\\	
			\midrule
			Political & Irish &~$2$ &~$13$ &~$\sim$~$54011$ & Elections from Dublin North and Meath\\
			Political & Glasgow	&~$13$ &~$\sim$~$11$ &~$\sim$~$8758$ & City council elections \\
            Political & Aspen &~$1$ &~$11$ &~$2459$	& City council elections\\
			Political & ERS	&~$13$ &~$\sim$~$12$ &~$\sim$~$988$ & Various elections held by non-profit organizations,\\ 
			        & & & & & trade unions, and professional organizations  \\
            \midrule
			Sport & Figure Skating &~$40$ &~$\sim$~$23$ &~$9$ & Figure skating  \\
			Sport & Speed Skating &~$13$ &~$\sim$~$14$ &~$196$ & Speed skating  \\
            Sport & TDF &~$12$ &~$\sim$~$55$ &~$\sim$~$22$ & Tour de France\\	
            Sport & GDI &~$23$ &~$\sim$~$152$ &~$20$  & Giro d’Italia	\\	
            \midrule
            Survey & T-Shirt &~$1$ &~$11$ &~$30$ & Preferences over T-Shirt logo \\	
            Survey & Sushi &~$1$ &~$10$ &~$5000$ & Preferences over Sushi\\	
            Survey & Cities &~$2$ &~$42$ &~$392$ & Preferences over cities	\\			
			\bottomrule
	\end{tabular}}
      \caption{\label{tab:preflib_selected} Each row contains a
        description of one of the real-life datasets we consider. In
        the column \textit{\# Selected Elections}, we denote the
        number of elections from the respective dataset that we
        consider.}
\end{table*}

\subsubsection{Preprocessing of Datasets}
There are two types of problems that we encounter in the selected
datasets. First, some datasets include ties (i.e., pairs or
larger sets of candidates that are reported as equally good by
particular voters). We break any ties that appear by ordering the
involved candidates uniformly at random.  Second, many of the votes
are incomplete in the sense that the reported preference orders do not
rank all the candidates but only some top ones (note that while an
incomplete vote can be seen as if the unranked candidates were tied on
the bottom of the vote, we address the incompleteness differently than
having a tie in a vote). Sometimes votes are both incomplete and
include ties.

For all the elections from our selected datasets that contain
incomplete votes, we need to fill-in all the missing data. For the
decision how to complete each vote, we use the other votes as
references, assuming that voters that rank the same candidates on top
also continue to rank candidates similarly toward the bottom.
For each incomplete vote~$v$, we proceed as follows. Let us assume
that vote~$v$ is over~$m$ candidates. Let~$V_P$ be the set of all
original votes of which~$v$ is a prefix. We uniformly at random select
one vote~$v_p$ from~$V_P$ and then at the end of vote~$v$ we add the
candidate which appears in position $m+1$ in vote $v_p$. We repeat the
procedure until vote~$v$ is complete. If the set~$V_P$ is empty, then
we choose~$c$ uniformly at random (from those candidates that are not
part of~$v$ yet).

Our approach to filling-in missing information is reminiscent of what
\citet{dou:thesis:partial-social-choice} calls \emph{imputation
  plurality}. We point to his work for detailed analysis of various
approaches to filling-in incomplete votes.

After applying these preprocessing steps, we arrive at a collection of
datasets containing elections with ten or more candidates and complete
votes without ties. As we focus on ten candidates, we need to select a
subset of ten candidates for each election. In a given election, we
compute each candidate's average position in the votes and select ten
candidates who, on average, are ranked highest (this is equivalent to
choosing ten candidates with the highest Borda scores). The rationale
is that these are the most relevant candidates in the elections. In
case there is a tie, we break it randomly.  We refer to the datasets
resulting from this preprocessing as \emph{intermediate} datasets.
Each intermediate dataset internally contains one or more elections
(e.g., the Sushi dataset contains only one election, whereas TDF contains
twelve).

\subsubsection{Sampling Elections from the Intermediate Datasets}
We treat each of our intermediate datasets as a separate election
model from which we sample~$15$ elections to create the final dataset
that we use. For each intermediate dataset, we sample elections as
follows. First, we randomly select one of the elections present
internally in it.
Notably, we may select the same election multiple
times. Second, we sample~$100$ votes from this election uniformly at
random with replacement (this implies that for elections with less
than~$100$ votes, we select some votes multiple times, and for
elections with more than~$100$ votes, we do not select some votes at
all). We do so to make full use of elections with far more than~$100$
votes. For instance, our Sushi intermediate dataset contains only one
election consisting of~$5000$ votes. Sampling an election from the
Sushi intermediate dataset thus corresponds to drawing~$100$ votes
uniformly at random from the set of~$5000$ votes. On the other hand,
for intermediate datasets containing a higher number of elections,
e.g., the Tour de France intermediate dataset, most of the sampled
elections come from different original elections.

After executing this procedure, we arrive at eleven sets, each
containing~$15$ elections consisting of~$100$ complete and strict
votes over~$10$ candidates, which we use for our experiments.

\begin{figure}[t]
    \centering
    \includegraphics[width=8.5cm, trim={0.2cm 0.2cm 0.2cm 0.2cm}, clip]{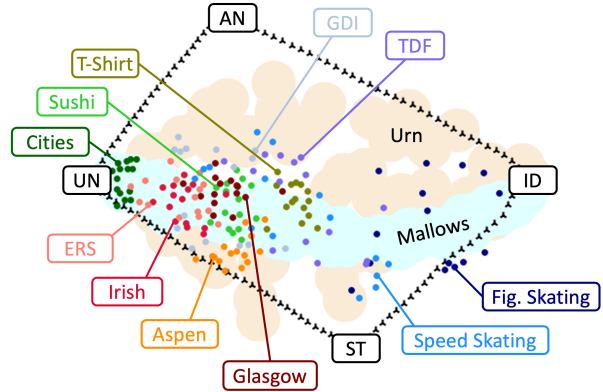}

    \caption{Map of real-life elections.}
    \label{fig:main_preflib_map}
\end{figure}

\begin{figure}[t]
    \centering

        \begin{subfigure}[b]{0.4\textwidth}
        \centering
        \includegraphics[width=5cm, trim={0.2cm 0.2cm 0.2cm 0.2cm}, clip]{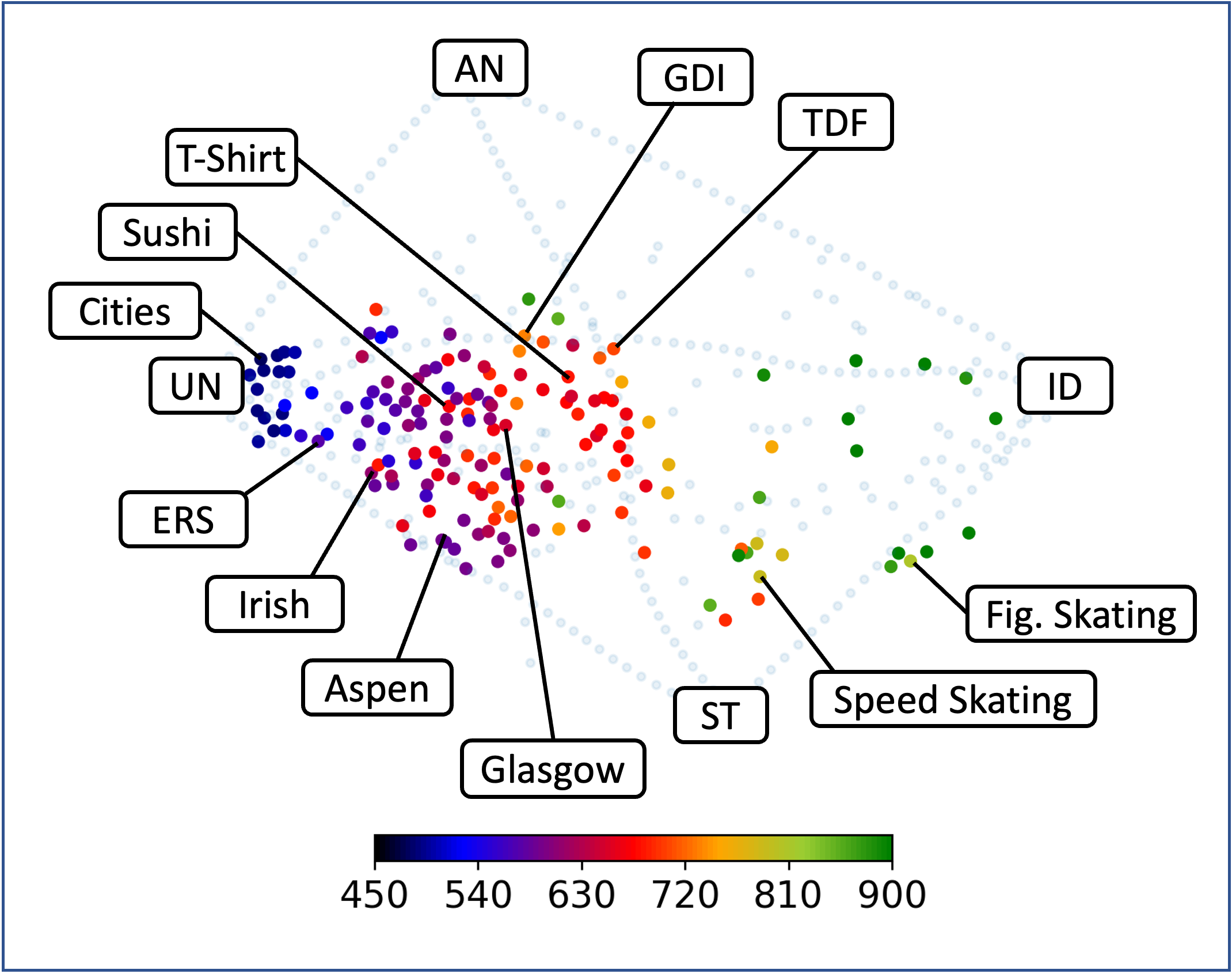}
        \caption{Borda (Real-Life)}
    \end{subfigure}
    \begin{subfigure}[b]{0.4\textwidth}
        \centering
        \includegraphics[width=5cm, trim={0.2cm 0.2cm 0.2cm 0.2cm}, clip]{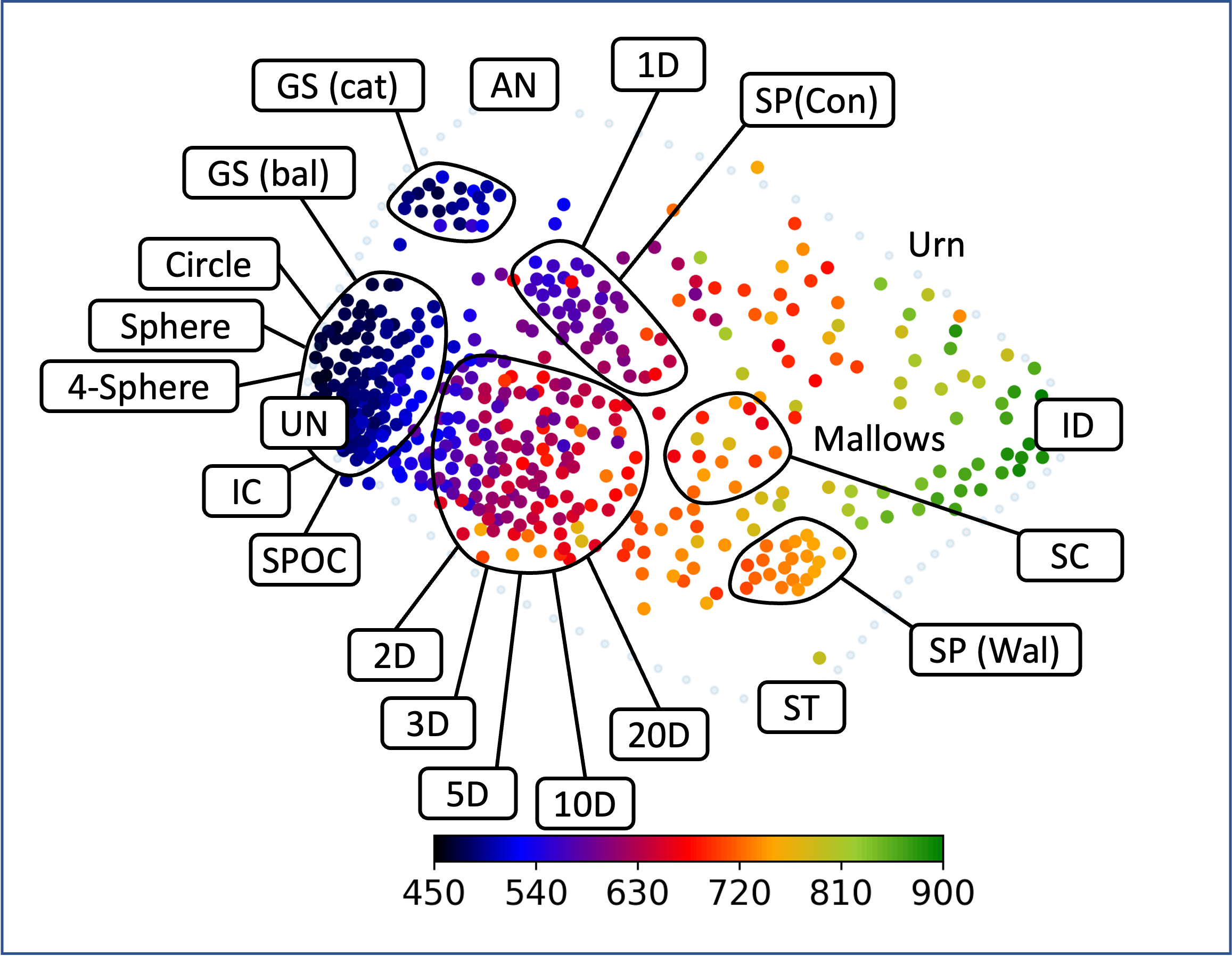}
        \caption{Borda ($10 \times 100$)}
    \end{subfigure}
    
    \vspace{1em}

    \begin{subfigure}[b]{0.4\textwidth}
        \centering
        \includegraphics[width=5cm, trim={0.2cm 0.2cm 0.2cm 0.2cm}, clip]{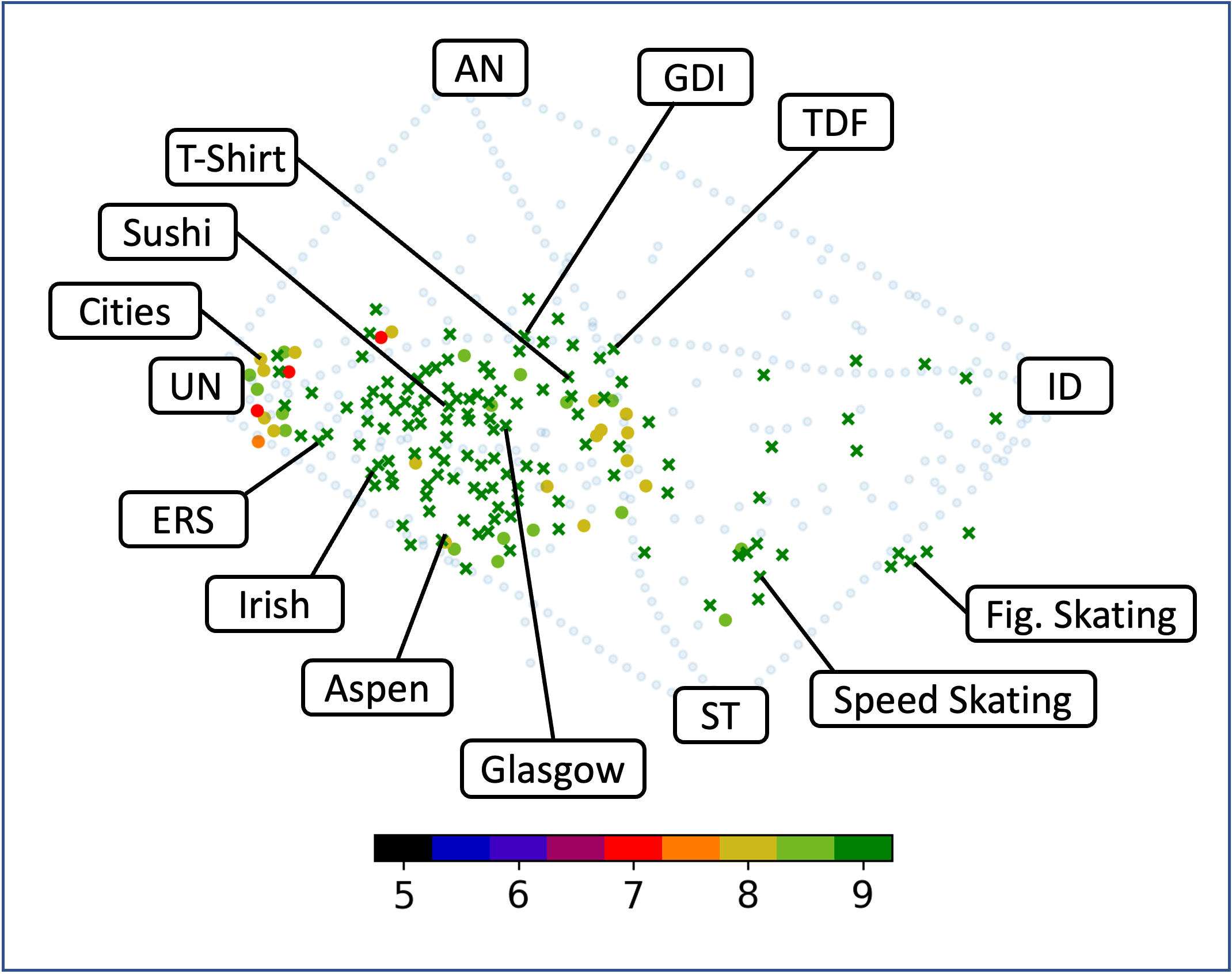}
        \caption{Copeland (Real-Life)}
    \end{subfigure}
    \begin{subfigure}[b]{0.4\textwidth}
        \centering
        \includegraphics[width=5cm, trim={0.2cm 0.2cm 0.2cm 0.2cm}, clip]{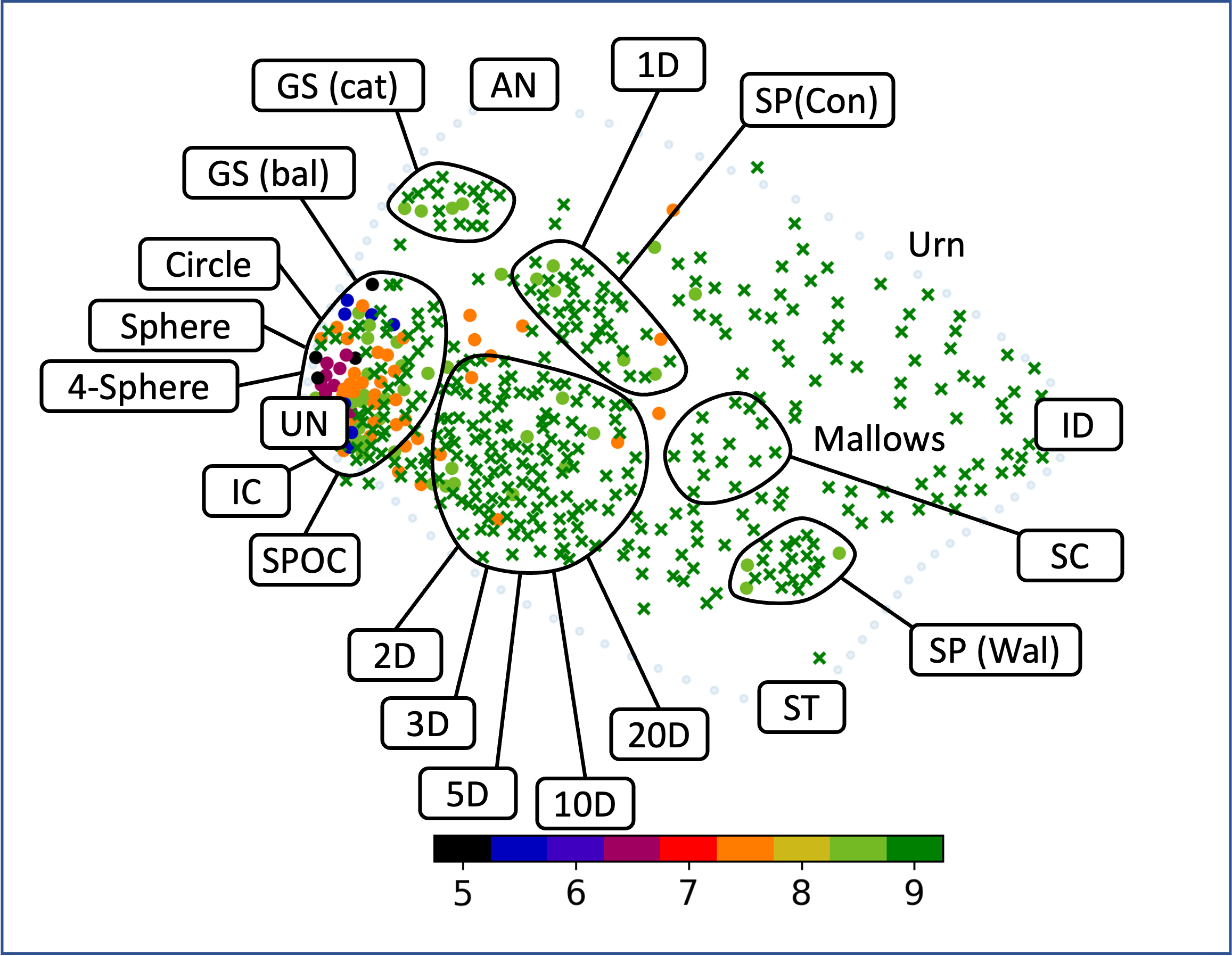}
        \caption{Copeland ($10 \times 100$)}
    \end{subfigure}
    
    \caption{Maps of the real-life elections (on the left) and the
      $10 \times 100$ dataset (on the right), colored according to the
      scores of the winning candidate under Borda and
      Copeland. On the real-life elections map, urn
      and Mallows elections are depicted as pale blue points, in the
      same way as the paths between the compass matrices.}
    \label{fig:preflib_coloring}
\end{figure}

\subsection{Real-Life Elections on the Map} \label{sub:maprel}

In Figure~\ref{fig:main_preflib_map}, we show a map of our real-life elections
along with the compass, Mallows, and urn elections. For readability, we
present the Mallows and urn elections as large, pale-colored areas. Not
all real-life elections form clear clusters, hence the labels refer to
the largest compact groupings.

While the map is not a perfect representation of distances among
elections (see \Cref{sec:robustness}),
analyzing it nevertheless leads to many conclusions.
Most strikingly, real-life elections occupy a very limited area of the
map. This is especially true for political elections and surveys,
which appear in the bottom-left quarter of the map. Except for several
sport elections, all elections are closer to~$\UN$ than to~$\ID$, and
none of the real-life elections falls in the top-right part of the
map. Another observation is that Mallows elections go right through
the real-life elections, while urn elections are on average further
away. This means that for most real-life elections there exists a
parameter~$\phi$ such that elections generated according to the
Mallows model with that parameter are relatively close (see
\Cref{sub:recom} for specific recommendations).
    
Most of the political elections lie close to each other and are
located next to the Mallows elections (and high-dimensional hypercube
ones, had we shown them on the map). At the same time, sport elections
are spread over a larger part of the map and, with the exception of
GDI, are shifted toward~$\ID$. Regarding the surveys, the Cities
survey is very similar to a sample from IC.\footnote{In the survey
  people were casting votes in the form of truncated ballots, ranking
  only their six favorite options. Hence, our completion method
  filled-in the rest of the votes uniformly at random. This is partly
  the reason why it is so similar to IC. Nevertheless, we have not
  observed any particular structure within these votes, hence its
  similarity to IC is not only accidental.}  The Sushi survey is
surprisingly similar to political elections. The T-Shirt surveys are
located in the middle of the Mallows cloud, closer to $\ST$ than to
$\AN$.
We mention that \citet{boe-sch:c:real-life-elections} show maps of
many further survey and sport-based elections. They find that the
former are occasionally closer to $\AN$, whereas the latter often
appear in the vicinity of $\ID$ (they also consider ``ground-truth''
elections, where the votes can be viewed as noisy estimates of the
objectively correct rankings, and these elections are even closer to
$\ID$).

\subsection{Winner Scores Under Borda and Copeland}\label{sec:real-scores}

In \Cref{sec:map:borda-copeland} we have shown maps of the
$100 \times 100$ dataset where each election was colored according to
the score of the winner under Borda and Copeland.  In
\Cref{fig:preflib_coloring} we show analogous maps for the real-life
elections and, for comparison, for the $10 \times 100$ dataset of synthetic elections.  We
make two observations. First, generally the colorings of the maps for the
$10 \times 100$ dataset are very similar to the ones for the
$100 \times 100$ dataset.  Second, the colorings of the real-life maps
are consistent with the colorings of the respective $10 \times 100$
maps, in the sense that similar colors appear in similar areas of the
maps.

All in all, this is yet another argument that the maps provide some
useful information regarding the elections they include and their absolute positions on the map carry some inter-map meaning.

\subsection{Capturing Real-Life Elections} \label{sub:recom}

As we have observed that real-life data is scarce and naturally hard to control, we now make some recommendations on how the Mallows model can be used to generate semi-realistic data. 
For this, we now analyze how to choose the~$\normphi$ parameter so that
elections generated using the Mallows model with our normalization
resemble the real-life ones. In this section we include only a fairly
simple experiment. In a follow-up work, we have explored the
possibility of learning Mallows parameters using the map framework
more deeply~\citep{boe-bre-elk-fal-szu:c:frequency-matrices}.

We consider four different datasets, each
consisting of elections with~$10$ candidates and~$100$ voters (created
as described in \Cref{sub:maprel}): the set of all political
elections, the set of all sport elections, the set of all survey
elections, and the combined set, i.e., the union of the three
preceding ones. For each of these four datasets, to find the value
of~$\normphi$ that produces elections that are as similar as possible
to the respective real-life elections, we conducted the following
experiment.  For each~$\normphi\in \{0,0.001,0.002,...,0.999,1\}$, we
generated~$100$ elections with~$10$ candidates and~$100$ voters from
the Mallows model with the given~$\normphi$ parameter. Subsequently,
we computed the average distance between these elections and the
elections from the respective dataset. Finally, we selected the value
of~$\normphi$ that minimized this distance. We present the results of
this experiment in \Cref{ta:realphiRL}.

\begin{table*}
\centering
\small
\begin{tabular}{l|c|c|c|c}
    \toprule
    Type of elections& Value of & Avg. Norm. & Norm. Std.& Num. of \\
   &~$\normphi$ & Distance & Dev. & Elections \\
    \midrule
    Political elections &~$0.750$ &~$0.15$ &~$0.036$ &~$60$ \\
    Sport elections &~$0.534~$ &~$0.27$ &~$0.080$ &~$60$ \\
    Survey elections &~$0.730$ &~$0.20$ &~$0.034$ &~$45$ \\ 
    All real-life elections &~$0.700$ & ~$0.22$ &~$0.106$ &~$165$ \\
    \bottomrule
  \end{tabular}
  \caption{\label{ta:realphiRL}Values of~$\normphi$ such that elections
    generated with the Mallows model for~$m=10$ are, on average, as close
    as possible to elections from the respective dataset. We include
    the average positionwise distance of the elections generated with the Mallows model for
    this parameter~$\normphi$ from the elections from the dataset as
    well as the standard deviation, both normalized by the positionwise distance between the uniformity and identity.
    The last column gives the number of elections in the respective
    real-life dataset.}
\end{table*}

Recall that in the previous section we have observed that a majority
of real-life elections are close to some elections generated from the
Mallows model with a certain dispersion parameter. However, we have
also seen that the real-life datasets consist of elections that differ
to a certain extent from one another (in particular, this is very
visible for the sports elections). Thus, it is to be expected that
elections drawn from the Mallows model for a fixed dispersion
parameter are at some nonzero (average) distance from the real-life
ones. Indeed, this is the case here. However, the more
homogeneous political elections and survey elections can be captured quite well using the Mallows model with parameter~$\normphi=0.750$ and~$\normphi=0.730$, respectively.  Generally speaking, if one wants to
generate elections that should be particularly close to elections from
the real world, then choosing a~$\normphi$ value between~$0.7$ and~$0.8$ seems like a good strategy. If, however, one wants to capture the full
spectrum of real-life elections, then we recommend using the Mallows
model with different values of~$\normphi$ from the interval~$[0.5,0.8]$. 

\section{Conclusions and Future Work}\label{sec:conclusions}
Our main conclusion is that the map of elections framework is both
credible and useful. In particular, regarding the credibility of the
map, in \Cref{sec:map-cultures} we have shown that the Euclidean
distances between points representing elections on the map are similar to the
positionwise distances between these elections, and that relations
between the distances are often preserved.  This is reinforced by the
fact that elections located close to each other on the map often have
similar high-level properties, such as, e.g., the scores of the
winning candidates and committees (see \Cref{sec:experiments}).
Regarding the usefulness of the map, in \Cref{sec:experiments} and in
\Cref{sec:real-life} we have shown a number of experiments where our
framework proved to be helpful.

Generally speaking, the strategy of using the map of elections in
one's own experiments is as follows: Take a particular dataset of
elections (e.g., our $100 \times 100$ dataset), compute the feature of
interest for each of the elections, and then draw the map, coloring
each election according to its computed feature. This way it is
possible to get a bird's eye view at the space of elections and see
how the considered feature depends on the election's position on the
map. Importantly, this way we see results for all the elections at the
same time, without the need to, e.g., compute average values or other
similar statistics, which often lose important information.  These
results can then guide more focused experiments. For example, one may
realize which statistical cultures lead to the most extreme values of
the feature, or which statistical cultures lead to elections with the
most diverse values.

\subsection{Advantages and Shortcomings of the Maps}
The most important conclusions of our paper relate to the
advantages and disadvantages of using the maps. Below we list a number
of observations that we either made throughout the paper or that
become apparent when looking at our work as a whole. We start with the
advantages of using the map framework:
\begin{enumerate}
\item Maps allow us to see relations between particular elections or
  groups of elections and, maybe even more importantly, with respect
  to compass elections. This helps to understand the nature of the
  elections that we have in our datasets. For example, we were able to
  see that elections generated from particular statistical cultures
  tend to be quite similar to each other. The two exceptions here are the urn and
  Mallows elections that cover large areas, where their placement is
  largely controlled by their parameters.

\item The maps allow us to form general intuitions about the features
  of our elections, by looking at the colorings according to these
  features. Looking at a map, we get a comprehensible bird's eye view
  of individual results for all the elections in the dataset. We saw
  this in essentially all the experiments in \Cref{sec:experiments}
  (we also describe more specific conclusions from these experiments
  in \Cref{sec:conclusions-specific} below).

\item The maps facilitate and motivate experiments conducted on
  elections from multiple, diverse data sources (such as different
  statistical cultures or different datasets from Preflib). While, in
  principle, this is not a feature of the maps---one could simply
  choose to use diverse datasets without knowing about maps at
  all---maps help one realize that given datasets are not diverse. For
  example, if one only considered one or two statistical cultures/data
  sources, a map of such data would make it apparent that there are
  areas in the election space that are not covered. Moreover, by
  observing which parts of the map remain unoccupied, one can get a
  good sense of what type of elections one's dataset is missing.  On a
  practical level, we also offer ready-made diverse datasets.
  We believe this is quite important, especially since
  \citet{boe-fal-jan-kac-lis-pie-rey-sto-szu-was:c:guide} observed
  that papers in computational social choice often use very limited
  data.
\end{enumerate}
The maps also have a number of shortcomings. We do not necessarily
perceive these shortcomings as reasons for not using the maps, but we
stress that one should take them into account. Below we list a few of
the issues:
\begin{enumerate}
\item The maps, at least when using the positionwise distance
  introduced and considered in this paper, are limited to datasets
  where all elections have the same number of
  candidates. 
  Similarly, the current map framework for ordinal elections does not
  handle the case where some preference orders are partial.

\item The maps, by necessity, are approximate. On the one hand, in
  this paper, we use the positionwise distance to measure the distance
  between two elections which, due to focusing on frequency matrices,
  disregards some information about the input elections. On the other
  hand, embedding algorithms introduce additional errors in the
  visualizations. While we put effort into evaluating the credibility
  of our maps, some errors are always present and one has to
  double-check the intuitions one gets from the maps.

\item As the positionwise distance focuses on elections' frequency
  matrices, it is incapable of differentiating elections whose
  frequency matrices are very similar. While in most cases this does
  not lead to significant issues, it should for instance be taken into
  account when analyzing elections close to the UN matrix---i.e.,
  elections where each candidate appears at each position in a vote
  with close to equal probability. For example, this is the case for
  elections generated using IC and SPOC statistical cultures. These
  elections are of very different nature---with the former one being
  quite chaotic and the latter one having a rather rigid
  structure---but generate nearly identical frequency matrices and,
  consequently, take similar positions on the maps based on the
  positionwise distance.
  
\item Due to the embedding algorithms used, the exact arrangement of
  elections in a map may be affected by the composition of the dataset
  used. If we add a large number of elections of a particular type,
  the area where these elections land on the map may be overinflated,
  with these elections taking a larger part of the picture. Hence, we
  repeat our warning about verifying the intuitions that one gets from
  the maps. That said, in our maps we do have many more elections
  close to $\UN$ than close to $\ID$, but it does not seem to affect
  the visualizations strongly under the KK embedding (but for other
  embeddings---such as the Fruchterman-Reingold---this effect is
  stronger, see \Cref{app:embeddings}).
  
\item The maps of elections with just a few candidates---such as four
  or five---do not seem to be particularly informative (see
  \Cref{sec:scalability}). Indeed, as shown by
  \citet{fal-sor-szu:j:subelection-isomorphism}, elections with such few
  candidates simply tend to be quite similar to each other
  irrespective of which data source they come from. As a result, maps that rely
  on analyzing election similarity cannot distinguish them well.

\item The results of various experiments may show, e.g., that
  elections that are close on the map have very different features. In
  some sense, this is not a disadvantage of the maps per se, but
  rather it is an indication that a given feature depends on much more
  subtle properties of elections than those captured by the considered
  distance measure. While such cases are interesting in themselves,
  they may require the use of other tools than maps.\footnote{As a
    mild form of this phenomenon, it may be the case that the value of
    a particular election feature is more closely tied to, e.g., the
    statistical culture from which the election was generated than to
    its position on the map. As statistical cultures are typically
    clustered on the maps, in such a case the maps are still useful,
    but the particular conclusions that one may draw may not be as
    strong as otherwise.}
\end{enumerate}

\subsection{Specific Conclusions}\label{sec:conclusions-specific}

Our work also leads to a number of more specific conclusions. For
example, we have found that elections generated according to
particular statistical cultures are fairly similar to each other.
This is a strong argument that when performing numerical experiments
using synthetic data, one should use a set of statistical
cultures that is as diverse as possible. This was further confirmed, e.g., by our
experiments regarding the running time of ILP-based winner determination
procedures for $\np$-hard voting rules. We observed that while for the
HB rule the time needed to find a winning committee in a given
election is correlated with its distance from $\ID$ (or, on a higher
level of abstraction, it is correlated with the diversity of the votes
in the election), for Dodgson this relation is less clear. In
particular, we have observed that computing the Dodgson winner for
caterpillar group-separable elections usually takes much more time
than for elections from other statistical cultures. This illustrates
that even if some statistical culture does not seem particularly
likely to generate realistic elections, it is still useful to include
it in experiments to get a sense for all possible behaviors that might occur.

Further, in our experiments regarding approximation algorithms, we
have compared the Sequential~CC and Removal~CC
algorithms for the CC rule. In the satisfaction-based model, the
former algorithm has much stronger approximation guarantees than the
latter, yet our experiments show that Removal~CC performs very well on
some data, more often than not achieving better results than
Sequential~CC.\footnote{Interestingly, there are two other
  approximation algorithms for CC, namely Ranging~CC and Banzhaf~CC,
  that also have equally strong approximation guarantees as
  Sequential~CC. Our initial results, not shown in the paper, indicate
  that in practice they perform notably worse than the other two
  algorithms on our datasets.}
This illustrates the value of experimental evaluation, which can give
quite different insights than theoretical studies.  More importantly
from our point of view, we have shown that the map of elections is
very useful in analyzing for which elections a particular algorithm
performs better.

Finally, in Section~\ref{sec:real-life} we have shown that the map of
elections framework can also be used for learning parameters of
statistical cultures that lead to elections most similar to real-life
ones. The idea is that we seek such a parameter of a given culture
under which the average distance between the generated elections and
the input real-life ones is as small as possible. In particular, we
have shown which normalized dispersion parameters for the Mallows
model are best for generating political, sport, and survey elections.
This idea is further developed in the work of
\citet{boe-bre-elk-fal-szu:c:frequency-matrices}.

\subsection{Follow-Up and Future Work}

The work presented in this paper has already lead to a number of
follow-up works, which we have discussed in \Cref{sec:related}. One of
the most pressing directions for future research is to extend the map
framework with the ability to deal with elections of different sizes
and with partial preference data. In other words, instead of filling
in incomplete votes, as we did in \Cref{sec:real-life}, it would be
far better to have distances that could work with partial preference
orders ``natively.''  Consequently, it would be possible to prepare a
map of all elections from Preflib~\citep{mat-wal:c:preflib} and see
how they relate to each other.  It would be very interesting to see
if, for example, there are particularly many elections in certain
areas of the map, if nearby elections have similar features, and if
conclusions we found regarding real-life elections would still hold.

It would also be interesting to extend the framework beyond ordinal
elections. This has already been done for approval
elections~\citep{szu-fal-jan-lac-sli-sor-tal:c:map-approval}, stable
roommates and stable marriages
problems~\citep{boe-hee-szu:c:map-stable-roommates}, voting
rules~\citep{fal-lac-sor-szu:c:map-of-rules}, and to visualize
internal structure of
elections~\citep{fal-kac-sor-szu-was:c:microscope}. However, there are
many more types of objects studied in computational social choice (and beyond) for
which appropriate maps could help in designing and conducting
experiments.

\subsubsection*{Acknowledgments}
We are very grateful to the AAMAS and IJCAI reviewers who provided
comments on the papers on which this work is based, and to the AIJ
reviewers. Their comments greatly helped in improving this paper.
This project has received funding from the European Research Council
(ERC) under the European Union’s Horizon 2020 research and innovation
programme (grant agreement No 101002854), and from the French
government under the management of Agence Nationale de la Recherche as
part of the France 2030 program, reference ANR-23-IACL-0008. Niclas
Boehmer was supported by the DFG project MaMu (NI 369/19).  In part,
Stanisław Szufa was supported by the Foundation for Polish Science
(FNP). \medskip

\begin{center}
  \includegraphics[width=3cm]{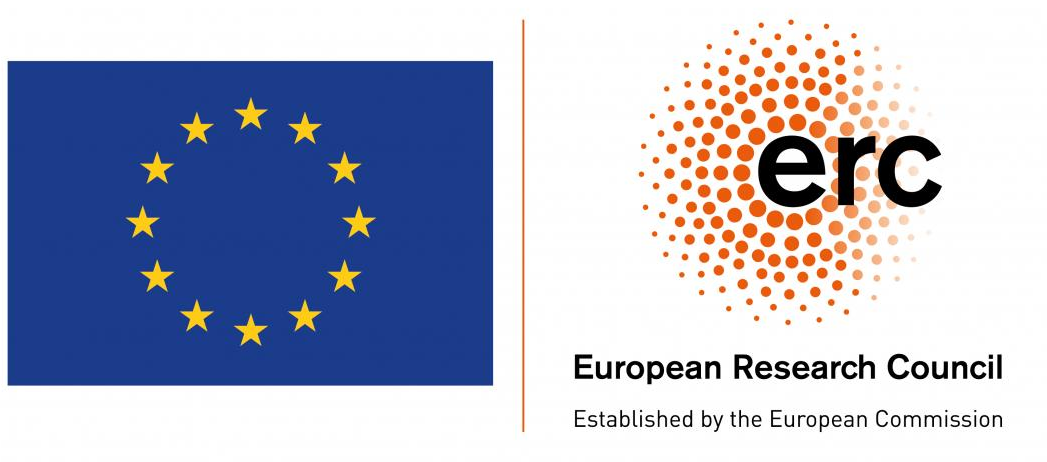}
\end{center}

\bibliographystyle{plainnat}

\newpage 

\appendix

{\centering
  \huge\bf
  Appendix

\bigskip
}

\section{Missing Proofs} \label{se:recov_app}

\polyalgo*
\begin{proof}
  Let $X$ be our input $m \times m$ matrix and let
  $C = \{c_1, \ldots, c_m\}$ be a set of candidates.  Our algorithm
  creates an election $E=(C, V)$ iteratively, as follows. In each
  iteration we first create a bipartite graph $G$ with vertex sets
  $A=\{a_1,\dots ,a_m\}$ and $B=\{b_1,\dots ,b_m\}$. For each
  $i, j \in [m]$, if $x_{i,j}$ is nonzero, then we put an edge
  between~$a_i$ and~$b_j$ (vertices in $A$ correspond to rows of $X$
  and vertices in $B$ correspond to the columns).  Next, we compute a
  perfect matching~$M$ in~$G$ (we will see later that it is guaranteed
  to exist).  Let $v$ be the vote that ranks $c_j$ on position $i$
  exactly if $M(a_i)=b_j$ ($v$ is well-defined because $M$ is a
  perfect matching).  Let $P$ be the position matrix corresponding to
  vote $v$, i.e., to election $(C,(v))$, and let $z$ be the largest
  integer such that $X-zP$ contains only non-negative entries.  Then,
  we add $z$ copies of $v$ to $V$ and set $X:=X-zP$.  We proceed to
  the next iteration until $X$ becomes the zero matrix.
  
  To prove the correctness of the algorithm, we show that at each
  iteration the constructed graph $G$ has a perfect matching.
  Let us assume that
  this is not the case.  Note that each row and each column in the
  current $X$ sums up to the same integer, say $n'$.  Since there is
  no perfect matching, by Hall's theorem, there is a subset of
  vertices $A'\subseteq A$ such that the neighborhood $B'\subseteq B$
  of $A'$ in $G$ contains fewer than $|A'|$ vertices. Yet, we have
  that $\sum_{a_i\in A', b_j\in B'} x_{i,j}=n' |A'|$, as we sum up all
  the nonzero entries of each row corresponding to a vertex from $A'$.
  However, this implies that $|B'| \geq |A'|$ because each column
  corresponding to a vertex from $B'$ sums up to $n'$, but we do not
  necessarily include all its nonzero entries. This is a
  contradiction.
  
  The algorithm terminates after at most $m^2-m+1$ steps (in each step
  at least one more entry of $X$ becomes zero, and in the last step,
  $m$ entries become zero).  Each step requires $\mathcal{O}(m^{2.5})$
  time to compute the matching, so the overall running time is
  $\mathcal{O}(m^{4.5})$.  This implies that $V$ contains at most
  $m^2-m+1$ different votes; indeed, using a similar argument as in  \citet{leep1999marriage} it can be shown that the algorithm always terminates
  after at most $m^2-2m+2$ steps.
\end{proof}

\dist*
\begin{proof}
\textbf{$\mathbf{\textbf{ID}_m}$ and $\mathbf{\textbf{UN}_m}$:} We start by computing the distance between $\ID_m$ and $\UN_m$. Note that $\UN_m$ always remains the same matrix independent of how its columns are ordered. Thus, we can compute the distance between these two matrices using the identity permutation between the columns of the two matrices: 
    $\POS(\ID_m,\UN_m) = \sum_{i=1}^m \EMD((\ID_m)_i, (\UN_m)_i)=\textstyle\sum_{i=1}^m (\textstyle\sum_{j=1}^{i-1} \frac{j}{m} + \textstyle\sum_{j=1}^{m-i} \frac{j}{m}) 
\\ = \frac{1}{m} \textstyle\sum_{i=1}^m ( \frac{1 + (i-1)}{2}(i-1) + \frac{1 + (m-i)}{2}(m-i) ) 
\\ = \frac{1}{2m} \textstyle\sum_{i=1}^m (2i^2 - 2i - 2mi + m^2  + m) 
\\ = \frac{1}{2m} (2\frac{m(m+1)(2m+1)}{6} - m(m+1)-m^2(m+1)  + m(m^2 + m))
\\ = \frac{1}{2m} (\frac{(m^2+m)(2m+1)}{3} - (m+1)(m+m^2)  + m(m^2 + m))
\\ = \frac{m+1}{2} (\frac{(2m+1)}{3} - (m+1)  + m) = \frac{(m+1)(m-1)}{3} = \frac{1}{3}(m^2-1).$

In the following, we use $(*)$ when we omit some calculations analogous to the calculations for $\POS(\ID_m,\UN_m)$.

\medskip
\noindent \textbf{$\mathbf{\textbf{UN}_m}$ and $\mathbf{\textbf{ST}_m}$:} Similarly, we can also directly compute the distance between $\UN_m$ and $\ST_m$ using the identity permutation between the columns of the two matrices. In this case, all column vectors of the two matrices have indeed the same $\EMD$ distance to each other:

 \noindent
 $\POS(\UN_m,\ST_m) =  m\cdot (\frac{1}{2}+2\cdot\textstyle\sum_{i=1}^{\frac{m}{2}-1} \frac{i}{m})= \frac{m}{2}+\frac{m}{2}(\frac{m}{2}-1) = \frac{m^2}{4}.$
 
  \medskip
\noindent \textbf{$\mathbf{\textbf{UN}_m}$ and $\mathbf{\textbf{AN}_m}$:} Next, we compute the distance between $\UN_m$ and $\AN_m$ using the identity permutation between the columns of the two matrices.  Recall that $\AN_m$ can be written as:
\[
  \AN_m = 0.5\begin{bmatrix}
    \ID_{\nicefrac{m}{2}} & \rID_{\nicefrac{m}{2}} \\
    \rID_{\nicefrac{m}{2}} & \ID_{\nicefrac{m}{2}}
  \end{bmatrix}.
\] Thus, it is possible to reuse our ideas from computing the distance between identity and uniformity:

  \noindent
 $\POS(\UN_m,\AN_m) = 4 \textstyle\sum_{i=1}^{\frac{m}{2}}(\textstyle\sum_{j=1}^{i-1} \frac{j}{m} + \textstyle\sum_{j=1}^{\frac{m}{2}-i} \frac{j}{m}) = (*) = \frac{2}{3}(\frac{m^2}{4}-1).$
 
 \medskip
\noindent \textbf{$\mathbf{\textbf{ID}_m}$ and $\mathbf{\textbf{ST}_m}$:}
There exist only two different types of column vectors in $\ST_m$, i.e.,  $\frac{m}{2}$ columns starting with $\frac{m}{2}$ entries of value $\frac{2}{m}$ followed by $\frac{m}{2}$ zero-entries and $\frac{m}{2}$ columns starting with $\frac{m}{2}$ zero entries followed by $\frac{m}{2}$ entries of value $\frac{2}{m}$. In $\ID_m$, $\frac{m}{2}$ columns have a one entry in the first $\frac{m}{2}$ rows and $\frac{m}{2}$ columns have a one entry in the last $\frac{m}{2}$ rows. Thus, again the identity permutation between the columns of the two matrices minimizes the $\EMD$ distance:

  \noindent
 $\POS(\ID_m,\ST_m) = 2\cdot \POS(\ID_{\frac{m}{2}},\UN_{\frac{m}{2}}) = \frac{2}{3}(\frac{m^2}{4}-1)$
 
  \medskip
\noindent \textbf{$\mathbf{\textbf{AN}_m}$ and $\mathbf{\textbf{ST}_m}$:} We now turn to computing the distance between $\AN_m=(\an_1,\dots , \an_m)$ and $\ST_m=(\stt_1,\dots , \stt_m)$. As all column vectors of $\AN_m$ are palindromes, each column vector of $\AN_m$ has the same $\EMD$ distance to all column vectors of $\ST_m$, i.e., for $i\in [m]$ it holds that $\EMD(\an_i,\stt_j)=\EMD(\an_i,\stt_{j'})$ for all $j,j'\in [m]$. Thus, the distance between $\AN_m$ and $\ST_m$ is the same for all permutation between the columns of the two matrices. Thus, we again use the identity permutation. 
We start by computing $\EMD(\an_i,\stt_i)$ for different $i\in [m]$ separately distinguishing two cases. Let $i\in [\frac{m}{4}]$. Recall that $\an_i$ has a $0.5$ at position $i$ and position $m-i+1$ and that $\stt_i$ has a $\frac{2}{m}$ at entries $j\in [\frac{m}{2}]$. We now analyze how to transform $\an_i$ to $\stt_i$. For all $j\in [i-1]$, it is clear that it is optimal that the value $\frac{2}{m}$ moved to position $j$ comes from position $i$. The overall cost of this is $\textstyle\sum_{j=1}^{i-1} \frac{2j}{m}$. Moreover, the remaining surplus value at position $i$ (that is, $\frac{1}{2}-\frac{2i}{m}$) needs to be moved toward the end. Thus, for $j\in [i+1,\frac{m}{4}]$, we move value $\frac{2}{m}$ from position $i$ to position $j$. The overall cost of this is  $\textstyle\sum_{j=1}^{\frac{m}{4}-i} \frac{2j}{m}$. Lastly, we need to move value $\frac{2}{m}$ to positions $j\in [\frac{m}{4}+1,\frac{m}{2}]$. This needs to come from position $m-i+1$. Thus, for each $j\in [\frac{m}{4}+1,\frac{m}{2}]$, we move value $\frac{2}{m}$ from position $m-i+1$ to position $j$. The overall cost of this is $\frac{1}{2}\cdot (\frac{m}{2}-i)+\textstyle\sum_{j=1}^{\frac{m}{4}} \frac{2j}{m}=\frac{1}{2}(\frac{m}{2}-i)+\frac{m}{16}+\frac{1}{4}$ 

Now, let $i\in [\frac{m}{4}+1,\frac{m}{2}]$. For $j\in [\frac{m}{4}]$, we need to move value $\frac{2}{m}$ from position $i$ to position $j$. The overall cost of this is $\frac{1}{2}\cdot (i-\frac{m}{4}-1)+\textstyle\sum_{j=1}^{\frac{m}{4}} \frac{2j}{m}=\frac{1}{2}\cdot (i-\frac{m}{4}-1)+\frac{m}{16}+\frac{1}{4}$. For $j\in [\frac{m}{4}+1,\frac{m}{2}]$, we need to move value $\frac{2}{m}$ from position $m-i+1$ to position $j$. The overall cost of this is $\frac{1}{2}\cdot (\frac{m}{2}-i)+\textstyle\sum_{j=1}^{\frac{m}{4}} \frac{2j}{m}=\frac{1}{2}\cdot (\frac{m}{2}-i)+\frac{m}{16}+\frac{1}{4}$. 

Observing that the case $i\in [\frac{3m}{4}+1,m]$ is symmetric to $i\in [\frac{m}{4}]$ and the case $i\in [\frac{m}{2}+1,\frac{3m}{4}]$ is symmetric to $i\in [\frac{m}{4}+1,\frac{m}{2}]$ the $\EMD$ distance between $\AN_m$ and $\ST_m$ can be computed as follows:

\noindent$\POS(\AN_m,\ST_m) = 2\cdot ( A + \frac{1}{2}\cdot (\sum_{i=1}^{\frac{m}{4}} \frac{m}{2}-i)  + \frac{m}{4}\cdot (\frac{m}{16}+\frac{1}{4})+ \frac{1}{2} \cdot (\sum_{i=\frac{m}{4}+1}^{\frac{m}{2}} \cdot (i-\frac{m}{4}-1))+\frac{m}{4}\cdot (\frac{m}{16}+\frac{1}{4})+ \frac{1}{2}\cdot (\sum_{i=\frac{m}{4}+1}^{\frac{m}{2}} \frac{m}{2}-i)+\frac{m}{4}\cdot (\frac{m}{16}+\frac{1}{4}))\\
=\frac{m^2}{48} - \frac{1}{3} + \frac{3m^2-4m}{32}  + \frac{m}{2}\cdot (\frac{m}{16}+\frac{1}{4})+  \frac{m^2-4m}{32}+\frac{m}{2}\cdot (\frac{m}{16}+\frac{1}{4})+ \frac{m^2-4m}{32}+\frac{m}{2}\cdot (\frac{m}{16}+\frac{1}{4}))\\
=\frac{m^2}{48}-\frac{1}{3}+\frac{3m^2-4m}{32}+\frac{3m}{2}(\frac{m}{16}+\frac{1}{4})+\frac{m^2-4m}{16}=(\frac{1}{48}+\frac{3}{32}+\frac{3}{32}+\frac{1}{16})m^2+(-\frac{4}{32}+\frac{3}{8}-\frac{4}{16})m\frac{1}{3}=\frac{13}{48}m^2-\frac{1}{3}$

\smallskip
\noindent with 

\noindent$A =  \textstyle\sum_{i=1}^{\frac{m}{4}}(\textstyle\sum_{j=1}^{i-1} \frac{2j}{m} + \textstyle\sum_{j=1}^{\frac{m}{4}-i} \frac{2j}{m}) = (*) = \frac{1}{6}(\frac{m^2}{16}-1) = \frac{1}{2}(\frac{m^2}{48} - \frac{1}{3})$

  \medskip
  
\noindent \textbf{$\mathbf{\textbf{ID}_m}$ and $\mathbf{\textbf{AN}_m}$:} Lastly, we consider $\ID_m=(\id_1,\dots , \id_m)$ and $\AN_m=(\an_1,\dots , \an_m)$. Note that, for $i\in [m]$, $\id_i$ contains a $1$ at position $i$ and $\an_i$ contains a $0.5$ at position $i$ and position $m-i$. Note further that for $i\in [\frac{m}{2}]$ it holds that $\an_i=\an_{m-i+1}$.
Fix some $i\in [\frac{m}{2}]$. For all $j\in [i,m-i+1]$ it holds that $\EMD(\an_i,\id_j)=\frac{m-2i+1}{2}$ and for all $j\in [1,i-1]\cup [m-i+2,m]$ it holds that $\EMD(\an_i,\id_j)>\frac{m-2i+1}{2}$. 
That is, for every $i\in [m]$, $\an_i$ has the same distance to all column vectors of $\ID_m$ where the one entry lies in between the two $0.5$ entries of $\an_i$ but a larger distance to all column vectors of $\ID_m$ where the one entry is above the top $0.5$ entry of $\an_i$ or below the bottom $0.5$ entry of $\an_i$. Thus, it is optimal to choose a mapping of the column vectors such that for all $i\in [m]$ it holds that $\an_i$ is mapped to a vector $\id_j$ where the one entry of $\id_j$ lies between the two $0.5$ in $\an_i$. This is, among others, achieved by the identity permutation, which we use to compute:

 \noindent
 $\POS(\ID_m,\AN_m) = 2 \textstyle\sum_{i=1}^{\frac{m}{2}} (\frac{1}{2} (m-2i+1)) 
  = \frac{m}{2}m-\frac{m}{2}(\frac{m}{2}+1)+\frac{m}{2} = \frac{m^2}{4}$
\end{proof}

\section{Different Embeddings and Their Robustness}\label{app:embeddings}

In this section, we provide the results of the accuracy experiments
from \Cref{sec:robustness} for the embedding algorithms of
Fruchterman-Reingold (FR)~\citep{fru-rei:j:graph-drawing},
Kamada-Kawai (KK)~\citep{kam-kaw:j:embedding}, and for
Multidimensional Scaling algorithm (MDS)~\citep{kru:j:mds,lee:j:mds}.
The results for KK from the main body of the paper are repeated in
this section, for comparison.  Overall, we find that KK is either
superior or nearly as good as each of the other ones.

\begin{figure}[]
  \centering

  \begin{subfigure}[b]{0.49\textwidth}
    \centering \includegraphics[width=8cm, trim={0.2cm 0.2cm 0.2cm
      0.2cm}, clip]{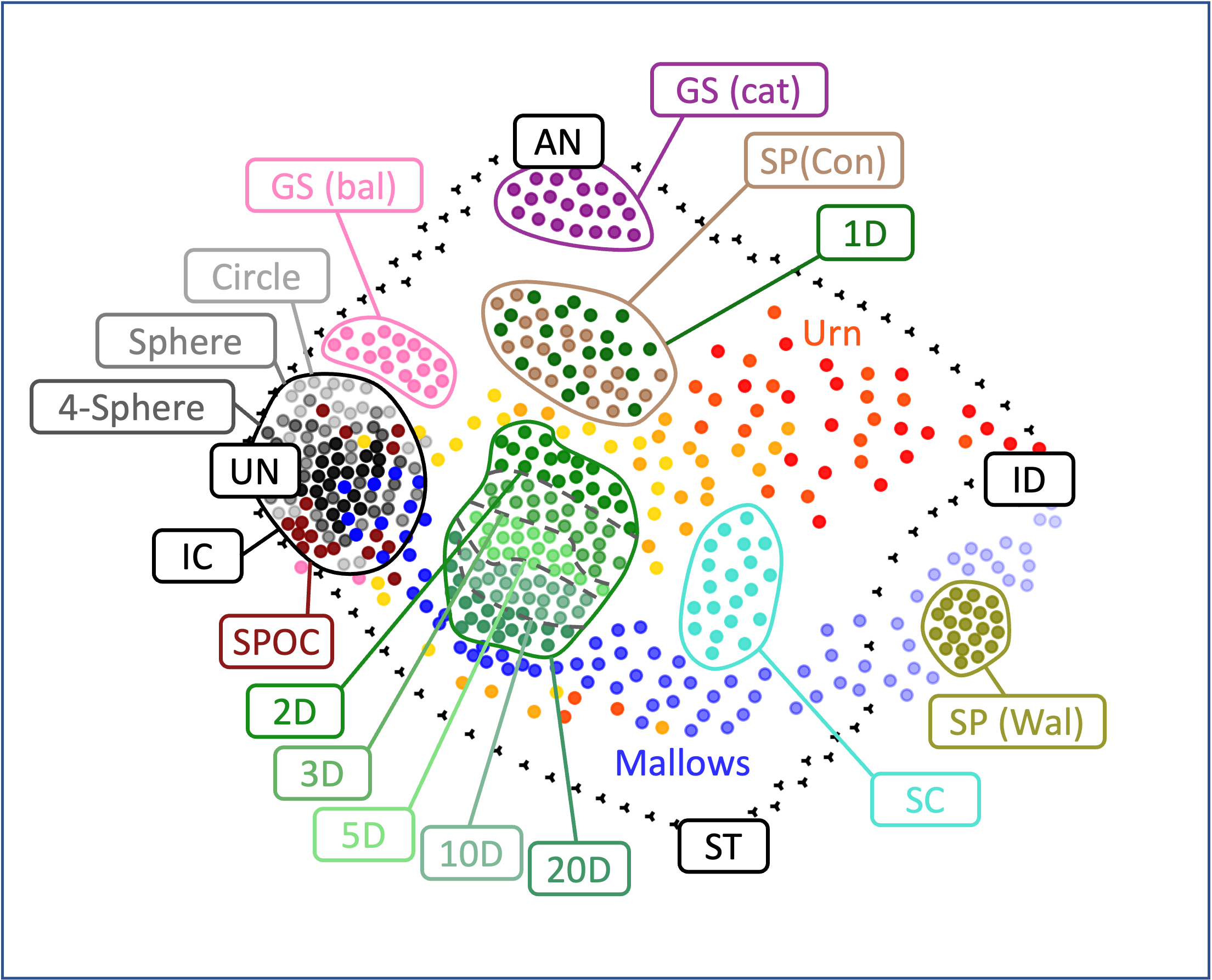}
    \caption{FR}
  \end{subfigure}
  \begin{subfigure}[b]{0.49\textwidth}
    \centering \includegraphics[width=8cm, trim={0.2cm 0.2cm 0.2cm
      0.2cm}, clip]{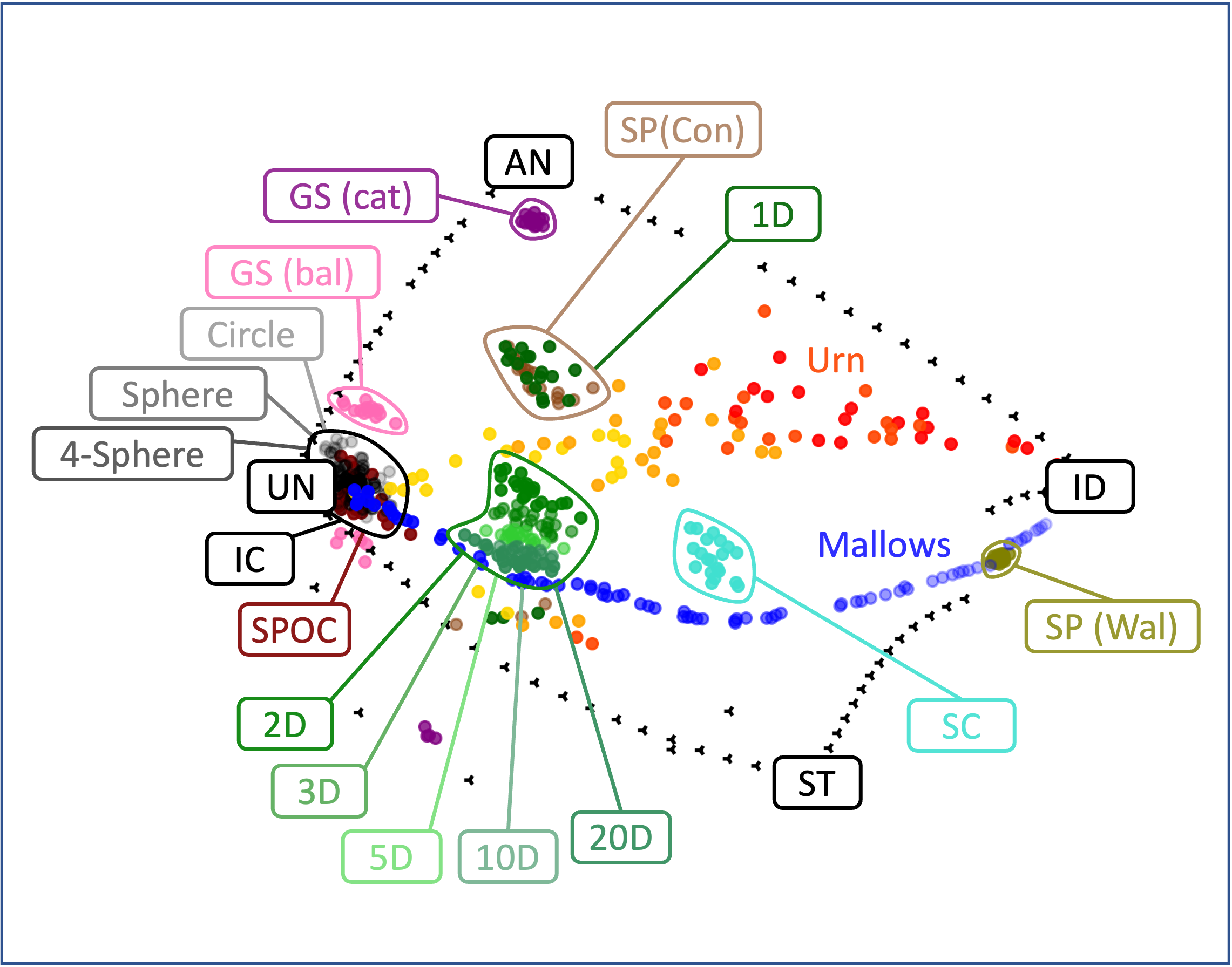}
    \caption{MDS}
  \end{subfigure}

  \begin{subfigure}[b]{0.49\textwidth}
    \centering \includegraphics[width=8cm, trim={0.2cm 0.2cm 0.2cm
      0.2cm}, clip]{images/embed/kamada_emd-positionwise.png}
    \caption{KK}
  \end{subfigure}
      
  \caption{Maps of elections obtained using the (a)
    Fruchterman-Reingold algorithm (FR), (b) the Multidimensional
    Scaling algorithm (MDS), and the Kamada-Kawai algorithm (KK). The
    colors of the dots correspond to the statistical culture from
    which the elections are generated. Elections generated from the Mallows
    model use the blue color, and the more pale they are, the closer
    is the $\normphi$ parameter to $0$. The urn elections use 
    yellow-red colors, and the more red they are, the larger is the
    contagion parameter $\alpha$.}
    \label{fig:emb:maps100x100}   
\end{figure}

In \Cref{fig:emb:maps100x100} we show maps of elections of the 100x100
dataset from \Cref{sec:map-cultures} obtained using the three
embedding algorithms. While the three maps in
\Cref{fig:emb:maps100x100} are different in some ways, for example,
the MDS one has the most compact clusters of points and the FR one has
the most scattered ones, on the high level they are similar. For
example, they all have IC, SPOC, and 4-Sphere elections very close to
the $\UN$ matrix, they all have a path of Mallows elections forming an
arc between the $\UN$ and $\ID$ matrices, and they all have hypercube
elections in the same area and in the same relation with respect to
each other (and with respect to the other elections).

\subsection{Correlation Between the Positionwise and Euclidean Distances}\label{app:pcc}

\begin{table}
  \centering
  \begin{tabular}{c|lll}
    \toprule
                     & \multicolumn{3}{c}{average PCC values} \\
    dataset          &  FR & MDS & KK \\
    \midrule
    $4   \times 100$ & \numberbarPCC{0.9226}{0.0083} & \numberbarPCC{0.9664}{0.0031} &  \numberbarPCC{0.9661}{0.0044} \\
    $10  \times 100$ & \numberbarPCC{0.9095}{0.0098}  & \numberbarPCC{0.9649}{0.0032} & \numberbarPCC{0.9686}{0.0061} \\
    $20  \times 100$ & \numberbarPCC{0.9135}{0.0105}  & \numberbarPCC{0.9657}{0.0036} & \numberbarPCC{0.9745}{0.0015} \\
    $100 \times 100$ & \numberbarPCC{0.9174}{0.0088}  & \numberbarPCC{0.9731}{0.0051} & \numberbarPCC{0.9735}{0.0115} \\
    \bottomrule
  \end{tabular}
  \caption{\label{tab:emb:pcc}Average values of the PCC between the
    original and embedded distances for collections of datasets of
    various sizes. After the $\pm$ signs we report the standard
    deviations.}
\end{table}

For each of the maps from \Cref{fig:emb:maps100x100}, we have computed
the PCC between the vector of the original distances and the vector of
the embedded ones. For the KK method the PCC is the highest and is
equal to~$0.9805$, for the MDS method it is equal to~$0.9748$ and for
the FR method it is equal~$0.9364$.  In \Cref{tab:emb:pcc} we report
average PCC values and standard deviations for the three embedding
algorithms and collections of datasets of various sizes (for our three
maps of the $100 \times 100$ dataset we chose very good embeddings,
whose PCC values are superior to the averages reported in
\Cref{tab:pcc}).

\begin{figure}[t]
  \centering
    
    \begin{subfigure}[b]{0.49\textwidth}
      \centering \includegraphics[width=8cm, trim={0.2cm 0.2cm 0.2cm
        0.2cm}, clip]{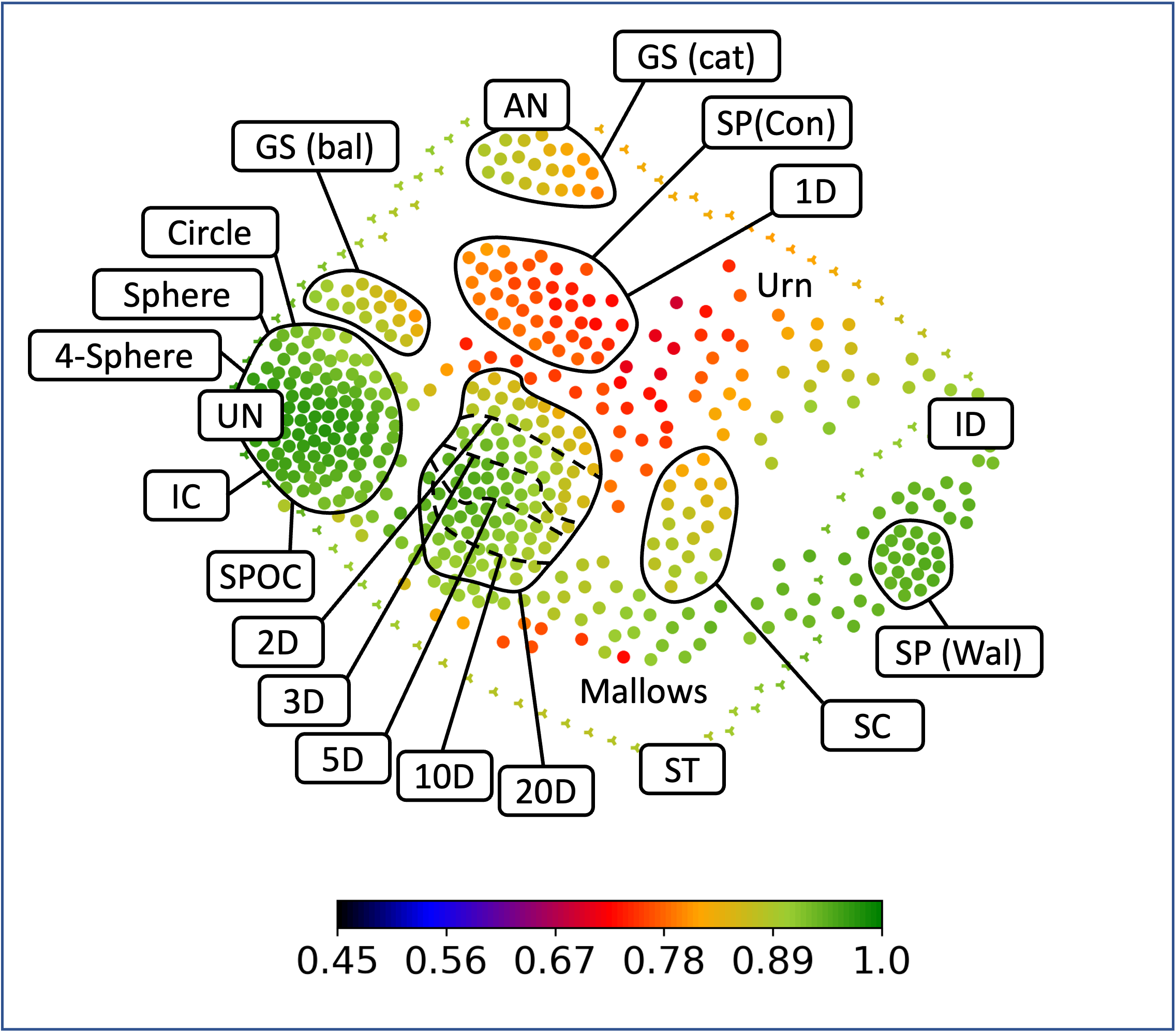}
      \caption{FR}
    \end{subfigure}%
    \begin{subfigure}[b]{0.49\textwidth}
      \centering \includegraphics[width=8cm, trim={0.2cm 0.2cm 0.2cm
        0.2cm}, clip]{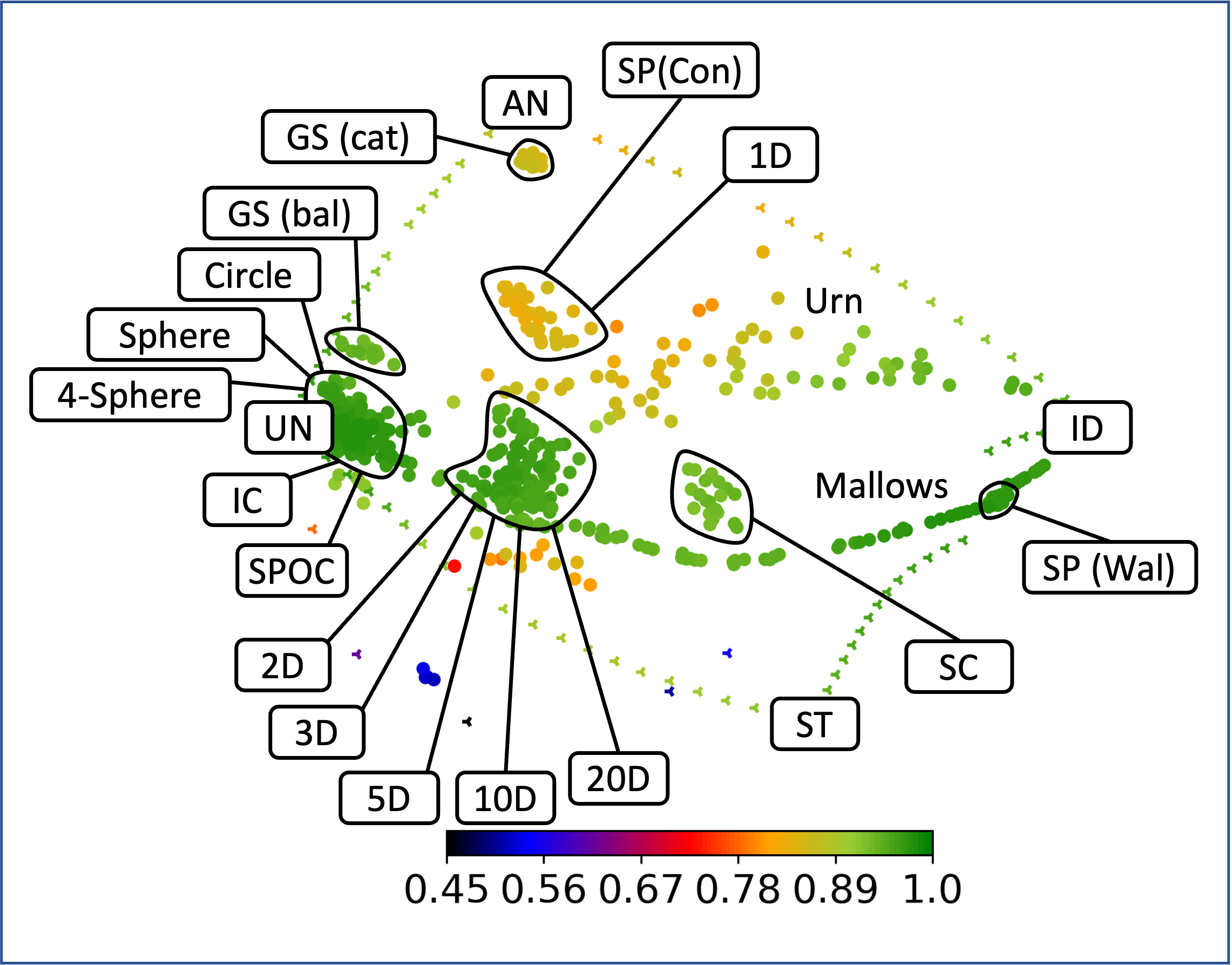}
      \caption{MDS}
    \end{subfigure}

    \begin{subfigure}[b]{0.49\textwidth}
      \centering \includegraphics[width=8cm, trim={0.2cm 0.2cm 0.2cm
        0.2cm}, clip]{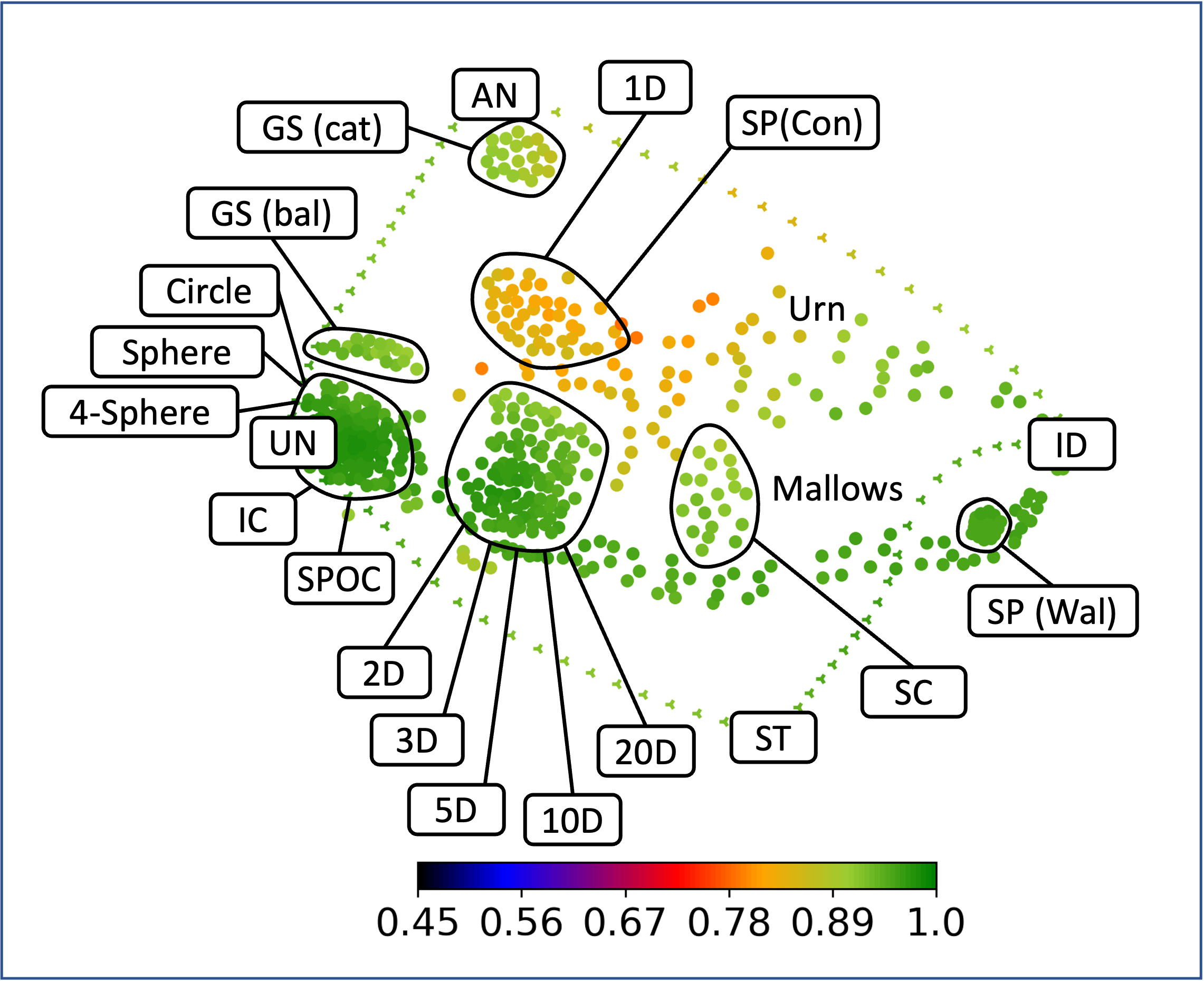}
      \caption{KK}
    \end{subfigure}

    \caption{Monotonicity coloring for the three emdeddings from
      \Cref{fig:maps100x100}. Each election $X$ (i.e., each point $X$)
      has a color reflecting its value $\mu(X)$ for the respective
      embedding.}
    \label{fig:monotonicity}
\end{figure}

\subsection{Monotonicity}\label{app:monotonicity}

Our next evaluation metric is what we call \textit{monotonicity}. The
intuition is that if the original distance between elections~$X$
and~$Y$ is larger than the original distance between elections~$X$
and~$Z$, then we expect that the same will hold for the embedded
distances.

Formally, for a given
embedding summary~$Q=(\mathcal{E}, d_\mathcal{M}, d_\Euc)$ and a given
election~$X\in \mathcal{E}$, we define the average monotonicity of this
election in the embedding summary to be:
\[
\textstyle \mu_{Q}(X) = \frac{1}{|{|\mathcal{E}|-1\choose 2}|}\sum_{Y,Z \in \mathcal{E}\setminus{X}} \Delta_{X}(Y,Z),
\]
where~$\Delta_{X}(Y,Z)$ is equal to 1, if 
\[
\sgn(d_\Euc(X,Y)-d_\Euc(X,Z)) = \sgn(d_\mathcal{M}(X,Y)-d_\mathcal{M}(X,Z)),
\]
and is equal to~$0$ otherwise. Positive (negative) signs in the above equation mean that
both the original and the embedded distances between~$X$ and~$Y$ were
larger (smaller) than the distances between~$X$ and~$Z$.  The larger
is the average monotonicity, the better.

In \Cref{fig:monotonicity} we present the three maps from
\Cref{fig:maps100x100} where each point (election) is colored
accordingly to its monotonicity value. The larger (the
closer to green) the value, the better, and the lower (the closer to
black) the value, the worse. Monotonicity equal to~$1$ means that all
inequalities are maintained after the embedding. For all three maps,
the main message is the same: Elections from the IC, SPOC, Mallows,
Walsh, and multidimensional Euclidean models are nicely
embedded. Then, elections from the single-crossing and group-separable
models are still fine, but on average worse than the previously
mentioned models. Finally, we have elections from the 1D-Interval,
Conitzer, and urn models, which are the worst embedded ones (not
counting some  caterpillar group-separable elections for MDS
embedding, which are obviously wrong).

\newcommand{\numberbarMonotonicity}[2]{\tikz{
    \fill[blue!17] (0,0) rectangle (#1*250mm-200mm,10pt);
    \node[inner sep=0pt, anchor=south west] at (0,0) {#1 $\pm$ #2};}
}

\begin{table}
  \centering
  \begin{tabular}{c|lll}
    \toprule
                     & \multicolumn{3}{c}{total average monotonicity values} \\
    dataset          &  FR & MDS & KK \\
    \midrule
    $4   \times 100$ & \numberbarMonotonicity{0.8684}{0.0051}  & \numberbarMonotonicity{0.8936}{0.0044} & \numberbarMonotonicity{0.8968}{0.0068} \\
    $10  \times 100$ & \numberbarMonotonicity{0.8637}{0.0112}  & \numberbarMonotonicity{0.9045}{0.0054} & \numberbarMonotonicity{0.9130}{0.0065} \\
    $20  \times 100$ & \numberbarMonotonicity{0.8649}{0.0105}  & \numberbarMonotonicity{0.9091}{0.0062} & \numberbarMonotonicity{0.9207}{0.0040} \\
    $100 \times 100$ & \numberbarMonotonicity{0.8689}{0.0065}  & \numberbarMonotonicity{0.9224}{0.0076} & \numberbarMonotonicity{0.9225}{0.0096} \\
    \bottomrule
  \end{tabular}
  \caption{\label{tab:monotonicity}Total average monotonicity values
    for collections of datasets of various sizes. After the $\pm$
    signs we report the standard deviations.}
\end{table}

In \Cref{tab:monotonicity} we provide the total average monotonicity
values for our collections of dataset. We see that with respect to
monotonicity, KK and MDS embeddings perform best, followed by
FR. This is in sync with the results regarding the Pearson correlation
coefficient.

We mention that instead of defining monotonicity for triples of
elections, we could have defined it for quadruples. Given elections
$X$ and $Y$ and elections $A$ and $B$, we could ask if the fact that
the original distance between $X$ and $Y$ is smaller than the original
distance between $A$ and $B$ implies the same relation between the
embedded distances. This would have lead to a somewhat more accurate
analysis, but would require significant amount of computation without
providing a clear advantage (in each of our datasets, we would have to analyze
nearly $480^4 \approx 53\cdot 10^9$ quadruples of elections).

\subsection{Distortion}\label{app:distortion}
\begin{figure}[t]
    \centering
    
     \begin{subfigure}[b]{0.49\textwidth}
        \centering
        \includegraphics[width=6.cm, trim={0.2cm 0.2cm 0.2cm 0.2cm}, clip]{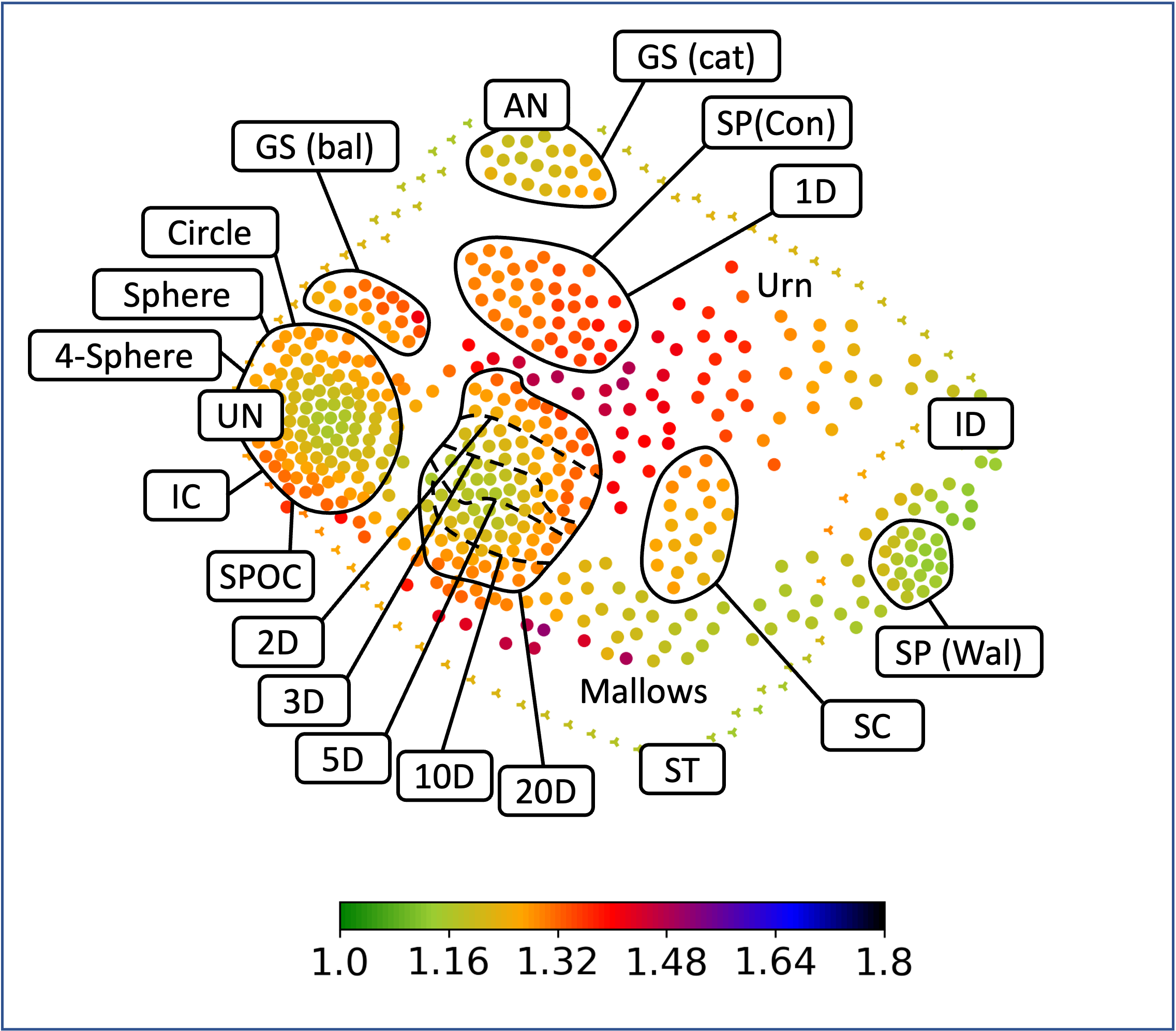}
        \caption{FR}
    \end{subfigure}
    \begin{subfigure}[b]{0.49\textwidth}
        \centering
        \includegraphics[width=6.cm, trim={0.2cm 0.2cm 0.2cm 0.2cm}, clip]{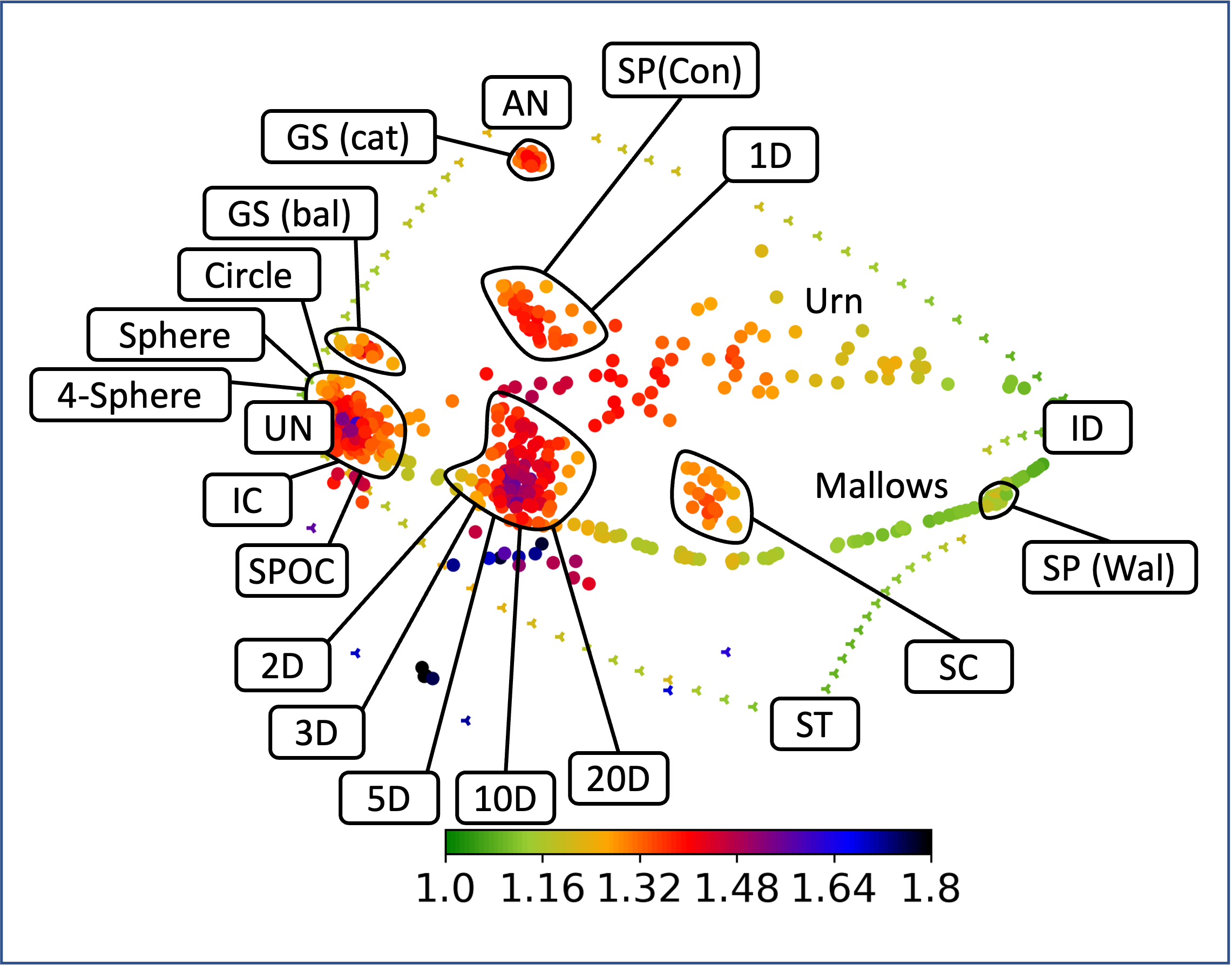}
        \caption{MDS}
    \end{subfigure}

     \begin{subfigure}[b]{0.49\textwidth}
        \centering
        \includegraphics[width=6.cm, trim={0.2cm 0.2cm 0.2cm 0.2cm}, clip]{images/embed/distortion_kamada_value.png}
        \caption{KK}
    \end{subfigure}     

    \caption{\label{fig:emb:distortion} Distortion coloring for the three
      emdeddings from \Cref{fig:maps100x100}. Each election $X$ (i.e.,
      each point $X$) has a color reflecting its value $\AMR(X)$ for
      the respective embedding.}
\end{figure}

In \Cref{fig:emb:distortion}, we present the map from
\Cref{fig:emb:maps100x100} with elections colored according to their distortion (i.e.,
according to their $\AMR$ vaules; recall \Cref{sec:distortion}). The best distortion is witnessed by
the KK embedding, followed by those for FR and MDS. In the cases of KK
and FR, urn elections, as well as 1D-Interval and Conitzer ones, have the
highest average distortion. This is in sync with the results
regarding monotonicity. In case of MDS, we additionally see elections
with high average distortion in the middle of the impartial culture
cluster and within the highly dimensional Euclidean clusters.

\begin{table}
  \centering
  \begin{tabular}{c|lll}
    \toprule
                     & \multicolumn{3}{c}{total average distortion values} \\
    dataset          &  FR & MDS & KK \\
    \midrule
    $4   \times 100$ & \numberbarDistortion{1.3213}{0.0157}  & \numberbarDistortion{1.3099}{0.0076} & \numberbarDistortion{1.2612}{0.0158} \\
    $10  \times 100$ & \numberbarDistortion{1.3119}{0.0194}  & \numberbarDistortion{1.3531}{0.0108} & \numberbarDistortion{1.2625}{0.0125} \\
    $20  \times 100$ & \numberbarDistortion{1.2979}{0.0195}  & \numberbarDistortion{1.3545}{0.0126} & \numberbarDistortion{1.2406}{0.0060} \\
    $100 \times 100$ & \numberbarDistortion{1.3006}{0.0256}  & \numberbarDistortion{1.3225}{0.0194} & \numberbarDistortion{1.2119}{0.0123} \\
    \bottomrule
  \end{tabular}
  \caption{\label{tab:emb:distortion}Total average distortion values
    (i.e., $\AMR$ values) for collections of datasets of various
    sizes. After the $\pm$ signs we report the standard deviations.}
\end{table}

In \Cref{tab:emb:distortion} we present the average values of the average
distortions for the three embedding methods and datasets with
different numbers of candidates. For each dataset, we compute the
average $\AMR$ value for its elections, as well as the standard
deviation.  We observe two patterns. The first one is related to the
embedding methods: KK is always the best, followed by FR, with MDS
being the worst (except for the $4 \times 100$ dataset). The second pattern is
that for KK the more candidates there are, the lower is the distortion
value. For MDS and FR this relation is not as clear.  Overall, for
each method the distances in the embedding are, on average, off by
$20-30\%$.

\section{Selection of Preflib Datasets}\label{app:selection}

In \Cref{tab:preflib} we list the datasets that were available on
Preflib when conducting our experiment on real-life elections (at the
end of 2020), including information whether we chose to include it in
our study and the reason for rejection (in case we deemed a given
dataset unsuitable).

\begin{table*}[t!]
  \centering \resizebox{\textwidth}{!}{\begin{tabular}{c c c c c c c
        c} \toprule
                                         Preflib ID & Name & Inst. & $m$ & Type & Selected &  Reason to reject \\
                                         \midrule
                                         1 & Irish & 3 & 9,12,14 & soi & yes &  -\\
                                         2 & Debian & 8 & 4-9 & toc & no & Too few candidates\\
                                         3 & Mariner & 1 & 32 & toc & no & Too many ties\\
                                         4 & Netflix & 200 & 3,4 & soc & no &Too few candidates \\
                                         5 & Burlington & 2 & 6 & toi & no & Too few candidates  \\
                                         6 & Skate & 48 & 14-30 & toc & yes &  -\\
                                         7 & ERS	& 87 & 3-29 & soi & yes &  -\\
                                         8 & Glasgow	& 21 & 8-13 & soi & yes&- \\
                                         9 & AGH & 2 & 7,9 & soc & no &  Too few candidates \\
                                         10.1 & Formula & 48 & 22-62 & soi & no &  Incomplete and few votes\\
                                         10.2 & Skiing & 2 & $\sim$50 & toc & no & Few votes and many ties \\
                                         11.1 & Webimpact & 3 &	103, 240, 242  & soc & no & Too many candidates and too few votes ($\sim$5)\\
                                         11.2 & Websearch & 74 &	100-200, $\sim$2000 & soi & no & Too few votes ($\sim$4)	\\
                                         12 & T-shrit & 1 & 11 & soc & yes &  -\\
                                         13 & Anes &	19 & 3-12 & toc &	no & Too many ties\\
                                         14 & Sushi & 1 & 10 & soc & yes & -\\
                                         15 & Clean Web & 79 & 10-50, $\sim$200 & soc & no & Too few votes ($\sim$4)	\\
                                         16 & Aspen & 2 & 5,11 & toc & yes & 		-\\
                                         17 & Berkeley & 1 & 4 & toc & no & Too few candidates 	\\
                                         18 & Minneapolis & 4 & 7,9,379,477 & soi & no & Incomplete votes	\\
                                         19 & Oakland & 7 & 4-11 & toc & no & Incorrect data (votes like: 1,1,1)\\
                                         20 & Pierce	& 4 & 4,5,7 & toc & no & Too few candidates			\\
                                         21 & San Francisco & 14 & 4-25 & toc & no & Incorrect data (votes like: 1,1,1)\\
                                         22 & San Leonardo & 3 & 4,5,7 & toc & no & Too few candidates \\
                                         23 & Takoma	& 1 & 4 & toc & no & Too few candidates \\
                                         24 & MT Dots & 4 & 4 & soc & no & Too few candidates					\\
                                         25 & MT Puzzles	& 4 & 4 & soc & no & Too few candidates				\\
                                         26 & Fench Presidential &	6 & 16 & toc & no & Approval ballots			\\
                                         27 & Proto French & 1 & 15 & toc & no & Approval ballots					\\
                                         28 & APA & 12 & 5 & soi & no & Too few candidates				\\
                                         29 & Netflix NCW & 12 & 3,4 & soc & no & Too few candidates						\\
                                         30 & UK labor party	& 1 & 5 & soi & no & Too few candidates				\\
                                         31 & Vermont & 15 & 3-6 & toc & no & Approval ballots \\
                                         32 & Cujae	& 7 & 6,32 & soc/soi/toc & no & Multiple issues				\\
                                         33 & San Sebastian Poster & 2 & 17 & toc & no & Approval ballots \\
                                         34 & Cities survey & 2 & 36, 48 & soi & yes &  -\\
                                         \bottomrule
	\end{tabular}}
      \caption{\label{tab:preflib} Overview of all election datasets
        that are part of the Preflib database (as of the end of 2020).
        ``Inst.'' stands for the number of elections in the dataset,
        ``$m$'' for the number of candidates and ``Type'' for the type
        of the votes in the dataset (``soc'' means that all votes are
        strict complete orders; ``soi'' means that all votes are
        strict incomplete orders; ``toc'' means that all votes are
        weak incomplete orders).}
      \end{table*}

\end{document}